\PassOptionsToPackage{dvipsnames}{xcolor}
\PassOptionsToPackage{pagebackref,draft=false}{hyperref}
\documentclass{article}

\usepackage{geometry}              
\usepackage{authblk}               

\usepackage{amsfonts}              
\usepackage{amsmath}               
\usepackage{amssymb}               
\usepackage{amsthm}                
\usepackage{bm}                    
\usepackage{centernot}             
\usepackage{colonequals}           
\usepackage{dsfont}                
\usepackage{mathtools}             
\usepackage{nicefrac}              
\usepackage{scalerel}              
\usepackage{xparse}                

\usepackage{arydshln}              
\usepackage{booktabs}              
\usepackage{colortbl}              
\usepackage{diagbox}               
\usepackage{makecell}              
\usepackage{multirow}              
\usepackage{tabularray}            
	\UseTblrLibrary{diagbox}

\usepackage{enumitem}              
	\setlist[itemize]{itemsep=0pt,topsep=4pt}
	\setlist[enumerate]{itemsep=0pt,topsep=4pt}
	\setlist[enumerate,1]{label=(\roman*)}
	\setlist[enumerate,2]{label=(\alph*)}

\usepackage{rotating}              
\usepackage{tikz}                  
\usepackage{tikz-cd}               
\usepackage{tikzit}                

\tikzstyle{morphism}=[fill=white, draw=black, shape=rectangle]
\tikzstyle{medium box}=[fill=white, draw=black, shape=rectangle, minimum width=0.7cm, minimum height=0.7cm]
\tikzstyle{large morphism}=[fill=white, draw=black, shape=rectangle, minimum width=1.7cm, minimum height=1cm]
\tikzstyle{bn}=[fill=black, draw=black, shape=circle, inner sep=1.5pt]
\tikzstyle{state}=[fill=white, draw=black, regular polygon, regular polygon sides=3, minimum width=0.8cm, shape border rotate=180, inner sep=0pt]
\tikzstyle{medium state}=[fill=white, draw=black, regular polygon, regular polygon sides=3, minimum width=1.3cm, inner sep=0pt, shape border rotate=180]
\tikzstyle{large state}=[fill=white, draw=black, regular polygon, regular polygon sides=3, minimum width=2.2cm, shape border rotate=180, inner sep=0pt]
\tikzstyle{wide state}=[fill=white, draw=black, shape=isosceles triangle, minimum width=0.8cm, shape border rotate=270, inner sep=1.4pt, minimum height=0.5cm, isosceles triangle apex angle=80]
\tikzstyle{wn}=[fill=white, draw=black, shape=circle, inner sep=1.5pt]
\tikzstyle{blue morphism}=[fill=white, draw={rgb,255: red,15; green,0; blue,150}, shape=rectangle, text={rgb,255: red,15; green,0; blue,150}, tikzit category=blue]
\tikzstyle{red morphism}=[fill=white, draw={rgb,255: red,150; green,0; blue,2}, shape=rectangle, text={rgb,255: red,150; green,0; blue,2}, tikzit category=red]
\tikzstyle{blue state}=[fill=white, draw={rgb,255: red,15; green,0; blue,150}, shape=circle, regular polygon, regular polygon sides=3, minimum width=0.8cm, shape border rotate=180, inner sep=0pt, text={rgb,255: red,15; green,0; blue,150}, tikzit category=blue]
\tikzstyle{blue node}=[fill={rgb,255: red,15; green,0; blue,150}, draw={rgb,255: red,15; green,0; blue,150}, shape=circle, tikzit category=blue, inner sep=1.5pt]
\tikzstyle{blue}=[text={rgb,255: red,15; green,0; blue,150}, tikzit draw={rgb,255: red,191; green,191; blue,191}, tikzit category=blue, tikzit fill=white, inner sep=0mm]
\tikzstyle{blue wide state}=[fill=white, draw={rgb,255: red,15; green,0; blue,150}, text={rgb,255: red,15; green,0; blue,150}, shape=isosceles triangle, minimum width=0.8cm, shape border rotate=270, inner sep=1.4pt, minimum height=0.5cm, isosceles triangle apex angle=80]
\tikzstyle{red node}=[fill={rgb,255: red,150; green,0; blue,2}, draw={rgb,255: red,150; green,0; blue,2}, shape=circle, inner sep=1.5pt]
\tikzstyle{Purple node}=[fill={rgb,255: red,120; green,0; blue,120}, draw={rgb,255: red,120; green,0; blue,120}, text={rgb,255: red,120; green,0; blue,120}, shape=circle, inner sep=1.5pt]
\tikzstyle{red}=[text={rgb,255: red,150; green,0; blue,2}, inner sep=0mm, tikzit fill=white, tikzit draw={rgb,255: red,191; green,191; blue,191}]
\tikzstyle{purple}=[text={rgb,255: red,150; green,0; blue,150}, inner sep=0mm, tikzit fill=white, tikzit draw={rgb,255: red,191; green,191; blue,191}]
\tikzstyle{white morphism}=[fill=white, draw=white, shape=rectangle, tikzit draw={rgb,255: red,139; green,139; blue,139}]
\tikzstyle{leak morphism}=[fill=white, draw={rgb,255: red,120; green,0; blue,85}, shape=rectangle, text={rgb,255: red,120; green,0; blue,85}, tikzit category=leak]
\tikzstyle{leak}=[text={rgb,255: red,120; green,0; blue,85}, inner sep=0mm, tikzit fill=white, tikzit draw={rgb,255: red,191; green,191; blue,191}, tikzit category=leak]
\tikzstyle{leak node}=[fill={rgb,255: red,120; green,0; blue,85}, draw={rgb,255: red,120; green,0; blue,85}, shape=circle, inner sep=1.5pt, tikzit category=leak]
\tikzstyle{horiz state}=[fill=white, draw=black, regular polygon, regular polygon sides=3, minimum width=1cm, shape border rotate=90, inner sep=0pt]

\tikzstyle{arrow}=[->]
\tikzstyle{dashed box}=[-, dashed]
\tikzstyle{blue arrow}=[-, draw={rgb,255: red,15; green,0; blue,150}, tikzit category=blue]
\tikzstyle{red arrow}=[-, draw={rgb,255: red,150; green,0; blue,2}, tikzit category=red]
\tikzstyle{purple arrow}=[->, draw={rgb,255: red,120; green,0; blue,120}, >=stealth, shorten <=2pt, shorten >=2pt]
\tikzstyle{protected purple arrow}=[->, draw={rgb,255: red,120; green,0; blue,120}, >=stealth, shorten <=2pt, shorten >=2pt, preaction={line width=1.8pt, white, draw}]
\tikzstyle{mapsto}=[{|->}]
\tikzstyle{double wire}=[-, double]
\tikzstyle{protected double wire}=[-, double, preaction={line width=2.5pt,white,draw}]
\tikzstyle{triple wire}=[-, draw, line width=0.4pt, preaction={-, draw, line width=1.4pt, white, preaction={-, draw, line width=2.2pt}}]
\tikzstyle{protected}=[-, preaction={line width=1.8pt,white,draw}]
\tikzstyle{leak arrow}=[-, tikzit draw={rgb,255: red,150; green,0; blue,120}]
\tikzstyle{protected leak arrow}=[-, tikzit draw={rgb,255: red,150; green,0; blue,120}]
\tikzstyle{hollow arrow}=[-, very thin, white, preaction={line width=0.7pt,draw={rgb,255: red,120; green,0; blue,85}}, tikzit category=leak, tikzit draw={rgb,255: red,150; green,0; blue,120}]
\tikzstyle{protected hollow arrow}=[-, very thin, white, preaction={line width=0.7pt,draw={rgb,255: red,120; green,0; blue,85},preaction={line width=2.1pt,white,draw}}, tikzit category=leak, tikzit draw={rgb,255: red,150; green,0; blue,120}]
\tikzstyle{over arrow}=[-, black, preaction={draw=white, double}]
\tikzstyle{purple line}=[-, draw={rgb,255: red,120; green,0; blue,120}]
\tikzstyle{orange line}=[-, draw={rgb,255: red,255; green,100; blue,0}]
\tikzstyle{red line}=[-, draw={rgb,255: red,150; green,0; blue,2}]
\tikzstyle{blue line}=[-, draw={rgb,255: red,15; green,0; blue,150}]
\tikzstyle{blue double arrow}=[-, double, draw={rgb,255: red,15; green,0; blue,150}, tikzit category=blue]
\tikzstyle{red double arrow}=[-, double, draw={rgb,255: red,150; green,0; blue,2}, tikzit category=red]
\tikzstyle{purple double line}=[-, double, draw={rgb,255: red,120; green,0; blue,120}]
\tikzstyle{orange double line}=[-, double, draw={rgb,255: red,255; green,100; blue,0}]
\tikzstyle{curly brace}=[decorate, decoration={brace,amplitude=5pt}]

\tikzstyle{d-wire1 plate}=[-, double=red!20!white]
\tikzstyle{d-wire2 plate}=[-, double=red!32!white]
\tikzstyle{dotted_plate}=[-, densely dotted, draw=red, fill opacity=0.4, fill=red!50!white, rounded corners]
\tikzstyle{twire1}=[-, draw, line width=0.4pt, preaction={-, draw, line width=1.4pt, red!20!white, preaction={-, draw, line width=2.2pt}}]
\tikzstyle{twire2}=[-, draw, line width=0.4pt, preaction={-, draw, line width=1.4pt, red!32!white, preaction={-, draw, line width=2.2pt}}]

	\usetikzlibrary{arrows.meta, positioning}
	\usetikzlibrary{patterns.meta}
	\usetikzlibrary{shapes, shadows, calc}

\usepackage{xcolor}                

\usepackage{doi}                   
\usepackage{nameref}               
\usepackage[numbers,sort]{natbib}  
\usepackage[nottoc]{tocbibind}	   

\usepackage[verbose=true]{microtype} 

\usepackage{calc}                  
\usepackage{comment}               
\usepackage{etoolbox}              
\usepackage{todonotes}             

\allowdisplaybreaks                

\definecolor{myurlcolor}{rgb}{0,0,0.3}
\definecolor{mycitecolor}{rgb}{0,0.3,0}
\definecolor{myrefcolor}{rgb}{0.3,0,0}
\usepackage[pagebackref,draft=false]{hyperref}
\hypersetup{colorlinks,
linkcolor=myrefcolor,
citecolor=mycitecolor,
urlcolor=myurlcolor}

\usepackage[noabbrev,capitalize]{cleveref}	

\definecolor{amber}{rgb}{1.0, 0.75, 0.0}
\definecolor{amaranth}{rgb}{0.9, 0.17, 0.31}
\definecolor{indiagreen}{rgb}{0.07, 0.53, 0.03}

\newcommand{\tablered}{\cellcolor{amaranth!30}}
\newcommand{\tableyellow}{\cellcolor{amber!30}} 
\newcommand{\tablegreen}{\cellcolor{indiagreen!30}}

\newcounter{todocounter}
\makeatletter
\def\cref@thmoptarg[#1]#2#3#4{%
	\ifhmode\unskip\unskip\par\fi%
	\normalfont%
	\trivlist%
	\let\thmheadnl\relax%
	\let\thm@swap\@gobble%
	\thm@notefont{\fontseries\mddefault\upshape}%
	\thm@headpunct{.}
	\thm@headsep 5\p@ plus\p@ minus\p@\relax%
	\thm@space@setup%
	#2
	\@topsep \thm@preskip               
	\@topsepadd \thm@postskip           
	\def\@tempa{#3}\ifx\@empty\@tempa%
	\def\@tempa{\@oparg{\@begintheorem{#4}{}}[]}%
	\else%
	\refstepcounter[#1]{#3}
	\@namedef{cref@#3@alias}{#1}
	\def\@tempa{\@oparg{\@begintheorem{#4}{\csname the#3\endcsname}}[]}%
	\fi%
	\@tempa}%
\makeatother

\newtheorem{theorem}{Theorem}[section]
\newtheorem*{theorem*}{Theorem}
\newtheorem{proposition}[theorem]{Proposition}
\newtheorem{lemma}[theorem]{Lemma}
\newtheorem{corollary}[theorem]{Corollary}
\newtheorem{definition}[theorem]{Definition}
\theoremstyle{definition}
\newtheorem{example}[theorem]{Example}
\newtheorem{remark}[theorem]{Remark}

\numberwithin{equation}{section}

\let\originalleft\left
\let\originalright\right
\renewcommand{\left}{\mathopen{}\mathclose\bgroup\originalleft}
\renewcommand{\right}{\aftergroup\egroup\originalright}

\newcommand{\V}{\mathsf{V}}
\newcommand{\I}{\mathsf{I}}
\newcommand{\Z}{\mathsf{Z}}
\newcommand{\Q}{\mathsf{Q}}
\renewcommand{\S}{\mathsf{S}}

\newcommand{\E}{\mathsf{E}}
\newcommand{\C}{\mathsf{C}}
\renewcommand{\P}{\mathsf{P}}
\newcommand{\F}{\mathsf{F}}
\newcommand{\W}{\mathsf{W}}

\newcommand{\K}{\mathsf{K}}
\newcommand{\R}{\mathsf{R}}
\newcommand{\B}{\mathsf{B}}
\newcommand{\T}{\mathsf{T}}
\newcommand{\Y}{\mathsf{Y}}
\newcommand{\G}{\mathsf{G}}
\newcommand{\sfH}{\mathsf{H}}

\newcommand{\ph}{\mathord{\rule[-0.05em]{0.6em}{0.05em}}}				

\newcommand{\id}{\mathds{1}}
\DeclareMathOperator{\im}{\mathsf{im}}
\newcommand{\swap}{\mathrm{swap}}

\newcommand{\comp}{ 		
	\mathchoice{\,}{\,}{}{} 	
}

\newcommand{\newterm}[1]{\textit{#1}}
\newcommand{\pre}{\mathrm{pre}}
\newcommand{\ret}{\mathrm{ret}}
\newcommand{\rank}{\mathrm{rank}}

\providecommand{\given}{}			
\newcommand{\SetSymbol}[1][]{%
	\nonscript\;\,#1\vert
	\allowbreak
	\nonscript\;\,
	\mathopen{}
}
\DeclarePairedDelimiterX{\Set}[1]{\{}{\}}{%
	\renewcommand{\given}{\SetSymbol[\delimsize]}
	#1
}
\makeatletter
\let\oldSet\Set
\def\Set{\@ifstar{\oldSet}{\oldSet*}}
\makeatother

\DeclarePairedDelimiterX{\Family}[1]{(}{)}{%
	\renewcommand{\given}{\SetSymbol[\delimsize]} 
	#1
}
\makeatletter
\let\oldFamily\Family
\def\Family{\@ifstar{\oldFamily}{\oldFamily*}}

\csundef{slashbox}
\newsavebox{\numbox}%
\newsavebox{\slashbox}%
\newsavebox{\denbox}%
\newlength{\slashlength}%
\newlength{\faktorscale}%
\DeclareDocumentCommand{\newfaktor}{m O{0.35} m O{-0.35}}{
\savebox{\numbox}{\ensuremath{#1}}
\savebox{\slashbox}{\ensuremath{\diagup}}
\savebox{\denbox}{\ensuremath{#3}}
\setlength{\faktorscale}{0.5\ht\numbox+0.5\ht\denbox}%
\setlength{\slashlength}{2pt+0.8\faktorscale+#2\faktorscale-#4\faktorscale}%
\raisebox{#2\ht\slashbox}{\usebox{\numbox}}
\mkern-2mu%
\rotatebox{-30}{\rule[#4\ht\denbox]{0.4pt}{\slashlength}}
\mkern9mu%
\hspace{-0.44\slashlength}%
\raisebox{#4\ht\denbox}{\usebox{\denbox}}
}
\DeclareDocumentCommand{\linefaktor}{m O{0.08} m O{-0.08}}{
\savebox{\numbox}{\ensuremath{#1}}
\savebox{\slashbox}{\ensuremath{\diagup}}
\savebox{\denbox}{\ensuremath{#3}}
\setlength{\faktorscale}{0.5\ht\numbox+0.5\ht\denbox}%
\setlength{\slashlength}{0.2\faktorscale+0.8\baselineskip}%
\raisebox{#2\ht\slashbox}{\usebox{\numbox}}
\mkern-1mu%
\raisebox{-0.8pt}{%
	\rotatebox{-25}{\rule[#4\ht\denbox]{0.4pt}{\slashlength}} 
}%
\mkern-1mu%
\hspace{-0.25\slashlength}%
\raisebox{#4\ht\denbox}{\usebox{\denbox}}
}

\definecolor{gridblue}   {RGB}{164,179,213}
\definecolor{gridgreen}  {RGB}{133,206,183}
\definecolor{gridred}    {RGB}{252,153,142}
\definecolor{gridpurple} {RGB}{201,153,202}
\definecolor{gridpink}   {RGB}{236,161,207}
\definecolor{gridolive}  {RGB}{184,224,118}
\definecolor{gridyellow} {RGB}{255,225, 89}
\definecolor{gridbrown}  {RGB}{193,163,131}
\definecolor{gridteal}   {RGB}{153,193,220}

\definecolor{gridgray}{gray}{0.4}

\colorlet{gridgraylight}	{gridgray!50!white}
\colorlet{gridbluelight}	{gridblue!50!white}
\colorlet{gridgreenlight}	{gridgreen!50!white}
\colorlet{gridredlight}		{gridred!50!white}
\colorlet{gridpurplelight}	{gridpurple!50!white}
\colorlet{gridpinklight}	{gridpink!50!white}
\colorlet{gridolivelight}	{gridolive!50!white}
\colorlet{gridyellowlight}	{gridyellow!50!white}
\colorlet{gridbrownlight}	{gridbrown!50!white}
\colorlet{gridteallight}	{gridteal!50!white}

\def\cellsize{0.6}
\def\panelsep{3.9}      

\newcommand{\drawgrid}{
	\draw (0,0) rectangle ++(3*\cellsize,3*\cellsize);
	\foreach \i in {1,2} {
		\draw (\i*\cellsize,0) -- (\i*\cellsize,3*\cellsize);
		\draw (0,\i*\cellsize) -- (3*\cellsize,\i*\cellsize);
	}
}

\newcommand{\panelaxes}{
	\node[left=0.2cm]  at (0,1.5*\cellsize) {$P$};
	\node[below=0.2cm] at (1.5*\cellsize,0) {$Q$};
}

\def\StripeAngle{-45}
\def\StripeDistance{3pt}

\newcommand{\filltwo}[5]{%
	
	\fill[#4] (#1,#2) rectangle ++(#3,#3);
	
	\fill[
	pattern={
		Lines[
		angle=\StripeAngle,
		distance=2*\StripeDistance,
		line width=\StripeDistance
		]
	},
	pattern color=#5
	] (#1,#2) rectangle ++(#3,#3);
}

\newcommand{\fillthree}[6]{%
	
	\fill[#4] (#1,#2) rectangle ++(#3,#3);
	
	\fill[
	pattern={
		Lines[
		angle=\StripeAngle,
		distance=3*\StripeDistance,
		line width=2*\StripeDistance
		]
	},
	pattern color=#5
	] (#1,#2) rectangle ++(#3,#3);
	
	\fill[
	pattern={
		Lines[
		angle=\StripeAngle,
		distance=3*\StripeDistance,
		line width=\StripeDistance,
		xshift=\StripeDistance
		]
	},
	pattern color=#6
	] (#1,#2) rectangle ++(#3,#3);
}

\begin{document}

\title{The View from Within: \\ What Can Embedded Observers (Not) Learn?}

\author[1]{Tom\'a\v{s} Gonda\textsuperscript{*}} 
\author[2]{Johannes Fankhauser\textsuperscript{*}} 
\author[3,4]{Gemma De les Coves\textsuperscript{*}} 

\affil[1]{Institute for Mathematics, University of Innsbruck, Austria}
\affil[2]{Institute for Theoretical Physics, University of Innsbruck, Austria}
\affil[3]{Departament d'Enginyeria, Universitat Pompeu Fabra, Carrer T\`anger 122, 08018 Barcelona}
\affil[4]{ICREA, Instituci\'o Catalana d'Estudis i Recerca Avan\c{c}ats, Passeig Llu\'is Companys 23, 08010 Barcelona}
\date{\today}
\renewcommand\Affilfont{\itshape\small}

\maketitle

\begingroup
	\renewcommand{\thefootnote}{\fnsymbol{footnote}}
	\footnotetext[1]{\, TG and JF contributed equally to this work. \\
		Email contacts: \href{mailto:tomas.gonda@uibk.ac.at}{\itshape tomas.gonda@uibk.ac.at}, \href{mailto:johannes.j.fankhauser@gmail.com}{\itshape johannes.j.fankhauser@gmail.com}, \href{mailto:gemma.delescoves@upf.edu}{\it gemma.delescoves@upf.edu}
		}
\endgroup

\begin{abstract}
	Physics is usually done from a third-person perspective, as if the world were described from outside. 
	Yet, observers are themselves physical systems within the world they observe. 
	Here, we investigate this tension using a toy model of classical particles. 
	Observers are physical systems characterised by a choice of (i) a manifest variable, whose value constitutes their empirical record, and (ii) ready states, which provide initial information. 
	We ask what such observers (subjects) can learn about the world through interactions that establish a correlation with another system (the object). 
	For each combination of the three types of learning (about the past, the future, or both), manifest variables, and ready states, we determine what the subject can learn.
	This shows that many subjects face an epistemic horizon{\,---\,}a limitation to what they can learn about the world{\,---\,}even though the model is classical and deterministic. 
	For example, a subject with a complete manifest variable{\,---\,}one whose ontic state is the empirical record{\,---\,}and partial initial information can learn at most half the object's variables, recovering Spekkens' knowledge-balance principle, i.e.\ an analogue of Heisenberg's uncertainty principle. 
	More generally, we find that subjects face more severe epistemic horizons when predicting than when retrodicting, and that learning by repeatable measurements can be more limited than either prediction or retrodiction alone. 
	Our work provides language and tools to study how the first-person perspective can differ from the third-person perspective beyond the toy model studied here, and invites explanations of the inherent uncertainty in quantum theory from the standpoint of embedded observers.
\end{abstract}

\newpage
\tableofcontents
\newpage

\section{Introduction}\label{sec:intro}

The mechanistic, third-person worldview of classical physics is predicated on a theoretical reality{\,---\,}the physical world{\,---\,}governed by universal laws that explain phenomena observable through empirical data. 
This reality is typically considered independent of the observer, composed of objects whose properties are revealed by measurement. 
In essence, what exists and what is known to exist are considered separate. 
This distinction has a long history in philosophy of science \cite{Kant,Va80} and remains central to modern discussions of the ontology of physical theories \cite{Be87,Ma07,Bh13,Sc20}.

For example, the Newtonian paradigm traditionally posits that specifying the positions and momenta of all particles at a given time, combined with the laws of motion, fully determines their entire future (and past).\footnotemark{}
\footnotetext{We gloss over well-known subtleties concerning determinism in Newtonian mechanics, such as Norton's dome \cite{Norton2003Causation} and related cases in which the standard uniqueness assumptions for solutions to the equations of motion fail.}%
It is often assumed that these quantities are, in principle, precisely measurable. 
One typically abstracts away from the measurement process, assuming any disturbance to the system is both determinable and correctable. 
If such correction were always possible, an observer could, in principle, access arbitrary information about a system.

Yet this classical idealisation leaves open how the third-person description provided by a physical theory connects to empirical observation. Einstein repeatedly emphasised that the relation between theoretical concepts and observations is conceptually mediated rather than simply given by experience \cite{EinsteinSolovine}. More recently, Maudlin has argued that a satisfactory physical theory should explain how observation arises from the physical degrees of freedom described by the theory itself \cite{Maudlin2025}. 

Relatedly, participatory and relational approaches emphasise that knowledge and objectivity should be understood from within the physical world rather than from an external standpoint, often in terms of correlations or relative information between physical systems \cite{Wh89,Rovelli-RQM,barbado2024relational,Co25,Di26}. 

The gap can be stated as follows: a theory may specify what physical states are possible and how they evolve, and yet fail to explain how they can function as information records for an embedded observer, or what the observer can infer about the world from such records.

In \cite{Fa24}, we took a first step towards addressing this gap in a deliberately simple classical setting. 
We introduced nomic toy theory, in which subjects are modelled as physical systems and information gathering as a physical interaction between subject and object. 
We showed that such subjects can face epistemic horizons{\,---\,}fundamental limitations on what they can learn{\,---\,}even though the underlying theory is fully deterministic. 
Because the information-gathering process is not abstract but physical, it is contingent both on the possible interactions and on the properties of the manifest variables whose values carry the observer's empirical records. 

The closest formal point of comparison is Spekkens' toy theory \cite{Sp07,spekkens2016quasi,Pu12,Ca17e,Ha21c}, an epistemically restricted operational theory with a non-contextual ontological model \cite{Ha10}.\footnote{For a contextual ontological model of $n$-qubit stabiliser quantum mechanics, see \cite{Hi22}.} 
The physical theory underlying our investigation coincides with the ontological model of Spekkens' toy theory, but the epistemic restriction arises differently. 
Whereas Spekkens' knowledge-balance principle restricts the allowed epistemic states by postulate, we model subjects as physical systems\footnotemark{} and measurements as physical interactions and ask what such subject can (not) learn.
That is, we study what limitations to subjects' knowledge follow as a result and find that, under certain assumptions, they coincide with the epistemic restriction of Spekkens' toy theory.
\footnotetext{For a different approach to modelling physical observers in Spekkens' toy theory, see \cite{Ha23b}.}%

In the spirit of physicalism \cite{sep-physicalism}, our toy model makes the following commitments about learning:
\begin{itemize}
	\item The subject is a physical system, and its empirical record is instantiated in its physical state.
	
	\item Empirical information about an object is represented by correlations between the state of the object and the subject's record.
	
	\item These correlations are established through physical interactions between the subject and the object.
\end{itemize}

This implies, for instance, that a record by itself does not determine what the subject knows about the object. 
The inference from the record to information about the object's state additionally presupposes knowledge of the details of the physical interaction that took place and the possible initial states of the ready system that enters the interaction. 
We treat this description as background information available to the subject, which need not itself be part of the record.\footnote{An analogous distinction appears in Wolpert's ``self-aware devices'', which represent separately the question an inference device takes itself to be answering and the answer it produces \cite[Sec.~9]{Wo08}.}

We use the term ``subject'' in a deliberately minimal sense: 
It denotes a physical system equipped with physical variables whose values can serve as empirical records.\footnote{This scope can be understood against Ismael's distinction between observational and agentive perspectives \cite{Is25b}: our framework studies the acquisition of information through physical interaction, while leaving the choice of interactions and the purposive use of information outside the model.}

A crucial assumption in our earlier work was that empirical records are carried by Poisson manifest variables, canonically exemplified by position projections.
Because such a manifest variable distinguishes only a coarse-graining of the subject's ontic state, we may interpret the resulting empirical record as providing only partial self-knowledge.
This raises the question:
\textit{Do subjects face an epistemic horizon when their manifest variable is complete, i.e.\ when their record is given by their own ontic state?} 

In this paper, we show that the answer to this question is nuanced, as it requires disambiguating several concepts which give rise to a varied landscape of epistemic horizons. 
The aspects that constrain learning in our model are (see \cref{fig:domains_for_epistemic_horizons}):
\begin{itemize} 
		\item The \textit{fundamental physical theory}: the kinematical and dynamical assumptions of the model.
		
		\item The \textit{manifest variable}: the variable whose value constitutes the subject's empirical record.
		Equivalently, one may say that the record supervenes on those distinctions in the subject's physical state selected by the manifest variable.
		We use \emph{self-knowledge} only as an interpretation of how finely this record distinguishes the subject's own ontic state.
	
		\item The \textit{type of learning}: whether the subject learns through prediction, retrodiction, or repeatable learning. 
		This aspect concerns whether knowledge is about the past state of the object, its future state, or both simultaneously. 
		
		\item Resources of \textit{initial information}: the constraints imposed on the initial ready state of systems used by the subject. 
		This concerns the question of whether the subject has access to systems that carry information, in the sense that their ontic state is constrained to a subset of their full state space. 
\end{itemize}

In particular, the fundamental physical theory in our investigation is fixed to be that of classical particles. 
We then study how different choices of manifest variable and different resources of initial information constrain what can be learned for each of the three types of learning. 

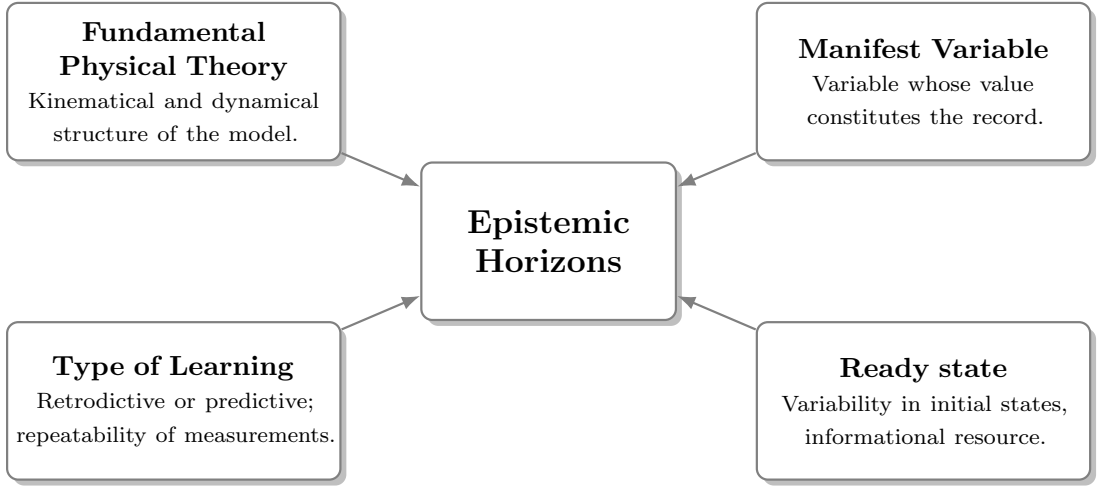
\begin{figure}[h]
	\centering
	\resizebox{\textwidth}{!}{
	\begin{tikzpicture}[
		base/.style = {rectangle, draw=black!50, thick, rounded corners, text centered, text width=4cm, minimum height=2cm, drop shadow},
		center/.style = {base, fill=white, text width=3cm, font=\bfseries\large},
		domain1/.style = {base, fill=white},
		domain2/.style = {base, fill=white},
		domain3/.style = {base, fill=white},
		domain4/.style = {base, fill=white},
		line/.style = {draw, -Latex, thick, gray}
		] 
		
		\node[center] (core) {Epistemic Horizons};
		
		\node[domain1, above left=0cm and 1cm of core] (phy) 
		{\textbf{Fundamental Physical Theory}\\ \footnotesize Kinematical and dynamical structure of the model.};
		
		\node[domain2, above right=0cm and 1cm of core] (sup) 
		{\textbf{Manifest Variable}\\ \footnotesize Variable whose value constitutes the record.}; 
		
		\node[domain3, below left=0cm and 1cm of core] (learn) 
		{\textbf{Type of Learning}\\ \footnotesize Retrodictive or predictive; repeatability of measurements.};
		
		\node[domain4, below right=0cm and 1cm of core] (res) 
		{\textbf{Ready state}\\ \footnotesize Variability in initial states, informational resource.}; 
		
		\draw[line] (phy) -- (core) ;
		\draw[line]  (sup) -- (core) ;
		\draw[line] (learn) -- (core);
		\draw[line]  (res) -- (core);
	\end{tikzpicture}}
	\caption{Four aspects of learning in nomic toy theory that determine a toy subject's epistemic horizons.}
	\label{fig:domains_for_epistemic_horizons}
\end{figure}

The answer to the question posed above is then: Under the constraint that measurements be repeatable (i.e.\ the same outcome is obtained if the measurement is repeated), a non-trivial epistemic horizon persists even for subjects with full access to their own ontic state. 
We establish this key result as \cref{thm:repeatable_complete_MV_Poisson_RS}.
More generally, all the epistemic horizons we find are summarised in \cref{tab:retrodiction,tab:repeatable,tab:prediction}. 

Related formal approaches obtain epistemic limitations from different sources. 
Wolpert derives constraints on observation, prediction, and recollection that apply to physical inference devices independently of the precise physical laws \cite{Wo08}, while Szangolies relates epistemic horizons in quantum theory to self-reference and undecidability \cite{Sz18}. 
By contrast, the horizons studied here arise from the allowed interactions and informational resources within a specified dynamical theory.

\paragraph{Structure of the paper.}
In \cref{sec:stage}, we introduce nomic toy theory, the physical theory used throughout the paper, together with the other three aspects of learning shown in \cref{fig:domains_for_epistemic_horizons}. 
In \cref{sec:EH}, we present the resulting epistemic horizons for different choices of manifest variables, ready states, and learning types. 
In \cref{sec:conclusions}, we summarise and interpret these results, and discuss open questions and directions for future work. 
The appendices contain the formal details of nomic toy theory (Appendix~\ref{sec:definitions}) and the proofs of the main results (Appendix~\ref{sec:proof_EH}).

\section{Nomic Toy Theory}\label{sec:stage}

To study epistemic horizons, i.e.\ fundamental limitations on learning for embedded observers, we formalise the four aspects of learning shown in \cref{fig:domains_for_epistemic_horizons}. The full details are given in Appendices \ref{sec:NTT}, \ref{sec:subjects}, and \ref{sec:learning}.

The first ingredient is the physical theory specifying the systems under consideration and their dynamics. We use nomic toy theory, introduced in \cite{Fa24}, and largely follow the original setup (see \cref{sec:toy_systems}).
Since our aim is to study learning about physical systems, \cref{sec:variables} explains how such information is represented in the model. \Cref{sec:toy_subjects} then introduces our minimal model of embedded observers, including their manifest variables{\,---\,}whose values constitute records{\,---\,}and their ready states. 
Finally, \cref{sec:learning_types} defines measurements as physical interactions between an observer and another system, as well as the information acquired through such interactions.

\subsection{Toy Systems and Transformations}\label{sec:toy_systems}
In nomic toy theory, a physical system may be modelled as $n$ classical particles, each with a position and a momentum. Formally, the state space of such a system $B$ is a $2n$-dimensional symplectic vector space $\B$. We call $B$ a \newterm{toy system}. 
The underlying field may be continuous, such as $\mathbb{R}$, or finite, such as $\mathbb{Z}_\ell$ for a prime $\ell$. 

Each symplectic vector space can be given a basis $(q_1,\dots,q_n,p_1,\dots,p_n)$ of what we may call position and momentum coordinates, such that their Poisson brackets 
\begin{equation}
	\forall \, i,j \in \{1,\ldots,n\} : \quad \{q_i, q_j\} = 0, \qquad \{q_i,p_j\} = \delta_{ij}, \qquad  \{p_i, p_j\} = 0
\end{equation}
can be expressed as
\begin{equation}\label{eq:sympl_form}
	\{b_i, b_j\} = b_i^T \Omega b_j,  \qquad \text{where} \qquad
		\Omega_\B = 
		\begin{pmatrix}
			0 &  \id \\
			- \id & 0
		\end{pmatrix}
\end{equation}
where $b_i$ and $b_j$ are basis elements for $i,j \in \{1,\ldots,2n\}$, and $\id$ denotes the identity matrix (of size $n\times n$). 
In the block matrix expression of \eqref{eq:sympl_form}, we use a decomposition $\B = \Q \oplus \P$, where $\Q$ is spanned by the positions $q_i$ and $\P$ by the momenta $p_i$. 

Given two systems, $B$ and $R$, their composite is another physical system with state space given by the direct sum $\B \oplus \R$.

Reversible dynamics in nomic toy theory align with those of classical mechanics for (at most) quadratic Hamiltonians. 
For a given system $B$, the allowed \newterm{physical transformations} on $B$ include all symplectic maps $\B \to \B$, i.e.\ linear maps that preserve the symplectic structure.
Additionally, we also allow shifts by a constant and operations that discard a subsystem, e.g.\ $\pi_\B \colon \B \oplus \R \to \B$ (see \cref{sec:NTT}, and \cref{ex:Poisson_map} specifically, for more details).
In \cref{thm:dilation} we show that a physical transformation can be decomposed into a (reversible) symplectic map on the input space followed by an (irreversible) discarding operation.

\subsection{Variables}\label{sec:variables}

Properties of toy systems are modelled by variables. 
A \newterm{variable} is a function $Z \colon \B \to \Z$ from the ontic state space $\B$ of the system to a set of possible values $\Z$. 
For example, if $Z$ encodes the predicate ``Is the position of the toy object equal to zero?" then $\Z$ is the set $\{\text{\texttt{yes}, \texttt{no}}\}$ and $Z(b) = \texttt{yes}$ holds if and only if the position of $b$ is $0$. 
For a visual representation of the information carried by a variable, see the diagrams in \cref{fig:variables} depicting the partition of the state space induced by the variable.

\begin{figure}[htb]\centering
	\begin{tikzpicture}
		\begin{scope}[shift={(0,0)}]
			\node at (1.5*\cellsize,3.7*\cellsize) {(a) Constant};
			
			\fill[gridblue] (0,0) rectangle ++(3*\cellsize,3*\cellsize);
			
			\drawgrid
			\panelaxes
		\end{scope}
		
		\begin{scope}[shift={(\panelsep,0)}]
			\node at (1.5*\cellsize,3.7*\cellsize) {(b) Position: $\pi_\Q$};
			
			\foreach \y in {0,1,2}
			{
				\fill[gridblue] (0,\y*\cellsize) rectangle ++(\cellsize,\cellsize);
				\fill[gridredlight] (1*\cellsize,\y*\cellsize) rectangle ++(\cellsize,\cellsize);
				\fill[gridgreenlight] (2*\cellsize,\y*\cellsize) rectangle ++(\cellsize,\cellsize);
			}
			
			\drawgrid
			\panelaxes
		\end{scope}
		
		\begin{scope}[shift={(2*\panelsep,0)}]
			\node at (1.5*\cellsize,3.7*\cellsize) {(c) $\pi_\Q + \pi_\P$};
			
			\fill[gridblue] (0*\cellsize,2*\cellsize) rectangle ++(\cellsize,\cellsize);
			\fill[gridblue] (1*\cellsize,1*\cellsize) rectangle ++(\cellsize,\cellsize);
			\fill[gridblue] (2*\cellsize,0*\cellsize) rectangle ++(\cellsize,\cellsize);
			
			\fill[gridredlight] (0*\cellsize,1*\cellsize) rectangle ++(\cellsize,\cellsize);
			\fill[gridredlight] (1*\cellsize,0*\cellsize) rectangle ++(\cellsize,\cellsize);
			\fill[gridredlight] (2*\cellsize,2*\cellsize) rectangle ++(\cellsize,\cellsize);
			
			\fill[gridgreenlight] (0*\cellsize,0*\cellsize) rectangle ++(\cellsize,\cellsize);
			\fill[gridgreenlight] (1*\cellsize,2*\cellsize) rectangle ++(\cellsize,\cellsize);
			\fill[gridgreenlight] (2*\cellsize,1*\cellsize) rectangle ++(\cellsize,\cellsize);
			
			\drawgrid
			\panelaxes
		\end{scope}
		
		\begin{scope}[shift={(3*\panelsep,0)}]
			\node at (1.5*\cellsize,3.7*\cellsize) {(d) Identity};
			
			\fill[gridblue] (0*\cellsize,2*\cellsize) rectangle ++(\cellsize,\cellsize);
			\fill[gridteallight] (1*\cellsize,2*\cellsize) rectangle ++(\cellsize,\cellsize);
			\fill[gridbrownlight] (2*\cellsize,2*\cellsize) rectangle ++(\cellsize,\cellsize);
			
			\fill[gridyellowlight] (0*\cellsize,1*\cellsize) rectangle ++(\cellsize,\cellsize);
			\fill[gridpinklight] (1*\cellsize,1*\cellsize) rectangle ++(\cellsize,\cellsize);
			\fill[gridgreenlight] (2*\cellsize,1*\cellsize) rectangle ++(\cellsize,\cellsize);
			
			\fill[gridolivelight] (0*\cellsize,0*\cellsize) rectangle ++(\cellsize,\cellsize);
			\fill[gridredlight] (1*\cellsize,0*\cellsize) rectangle ++(\cellsize,\cellsize);
			\fill[gridpurplelight] (2*\cellsize,0*\cellsize) rectangle ++(\cellsize,\cellsize);
			
			\drawgrid
			\panelaxes
		\end{scope}
	\end{tikzpicture}
	
	\caption[variables]{Diagrams depicting information about a toy system $B$, akin to those used for epistemic states in Spekkens' toy theory \cite{Sp07}.	
		Individual squares in a grid are elements of the ontic state space $\B$.
		In this example, $B$ is a single particle with a finite configuration space, so that 
		its position $Q$ (labelling columns) and its momentum $P$ (labelling rows) each have three possible values, i.e.\ $\Q \cong \P \cong \mathbb{Z}_3$.
		Each shaded region of a single colour represents a subset of the state space for a given value of a variable $Z \colon \B \to \Z$, for (a) a constant variable $Z(b) = 0$, (b) the position variable $Z = \pi_\Q$, (c) the sum of position and momentum $Z = \pi_\Q+\pi_\P$, and (d) the identity variable $Z=\id_{\B}$.}
	\label{fig:variables}
\end{figure}
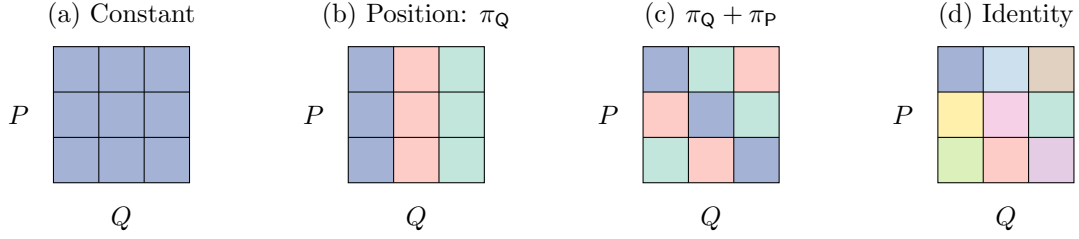

We denote variables by capital letters (such as $Z$ and $E$) and spaces of their values in sans-serif (such as $\Z$ and $\E$). 
In most cases of interest, the values form a vector space and the variable is a linear map. 

An important kind of variable is a \newterm{Poisson}\footnotemark{} variable{\,---\,}a surjective linear variable $Z$ satisfying
\footnotetext{See \cref{foo:Poisson} for a justification of the terminology.}%
\begin{equation}\label{eq:Poisson}
	Z \Omega_\B Z^T=0. 
\end{equation}
For a physical system $\B = \Q \oplus \P$, an example of a Poisson variable is the projection $\pi_\Q$ onto the configuration space $\Q$ of the positions of $n$ particles, see \cref{fig:variables} (b).
Condition \eqref{eq:Poisson} then says that the Poisson bracket between any two positions of the individual particles should vanish.
This is indeed the case and highlights why, more generally, Poisson variables model a set of compatible degrees of freedom of the system (cf.\ \cref{lem:Poisson}).
Note that in Spekkens' toy theory \cite{spekkens2016quasi,Sp07}, valid epistemic states are precisely those sets of ontic states that correspond, in our terminology, to values of Poisson variables.

Here, ``Poisson'' denotes a mathematical property of variables on any toy system.
When the variable in question is a subject's manifest variable, we speak of a \emph{Poisson manifest variable} and interpret it as a record-carrying variable.

\subsection{Toy Subjects}\label{sec:toy_subjects}

In our model, observers are embedded in the sense that they are themselves described as physical systems within the theory.
To represent the empirical records available to such observers, we equip their physical state spaces with manifest variables (cf.\ \cref{fig:domains_for_epistemic_horizons}).
This allows us to study what information such observers can acquire about other systems.

\paragraph{Manifest variables.}
We define a \newterm{toy subject} to be a toy system $S$ equipped with a manifest variable.\footnote{We often write `subject' instead of `toy subject' and `system' instead of `toy system'.}

In our previous work \cite[Definition 2.5]{Fa24}, the manifest variable was required to be Poisson, i.e.\ to satisfy \cref{eq:Poisson}. 
Here we allow a more general notion: 
A manifest variable is any surjective linear map $E \colon \S \to \E$, where $\E$ is the vector space of its possible values.

For a given ontic state $s\in\S$, the value $E(s) \in \E$ constitutes the subject's empirical record. 
In other words, the manifest variable specifies which distinctions in the subject's physical state are not relevant for the information it has about the world{\,---\,}they are pairs of states that correspond to the same empirical record.

We say that the subject knows a property provided that its record entails that property.
This is illustrated in \cref{fig:interaction}, where the interaction generates correlations such that the $E(s)=0$ implies that the initial position of the toy object $B$ was $0$, and also that its final momentum is $1$. 

In this paper, we consider three distinct types of manifest variables:
\begin{enumerate}
	\item For a \newterm{trivial} manifest variable, $E$ is constant. 
	The canonical example is the unique linear map $\S \to \{0\}$. 
	Every ontic state of the subject gives rise to the same record, which therefore carries no information.
	
	\item For a \newterm{Poisson} manifest variable, $E$ is a maximal Poisson variable (\cref{def:max_Poisson}).
	A chief example is the position projection $\pi_\Q \colon \S \to \Q$. 
	Its records distinguish one half of the subject's underlying degrees of freedom.
	We may interpret this as partial self-knowledge.
	
	\item A \newterm{complete} manifest variable is the identity map $\id_\S \colon \S \to \S$.
	The record is uniquely specified by the subject's ontic state.
	We may interpret this as complete self-knowledge.
\end{enumerate}

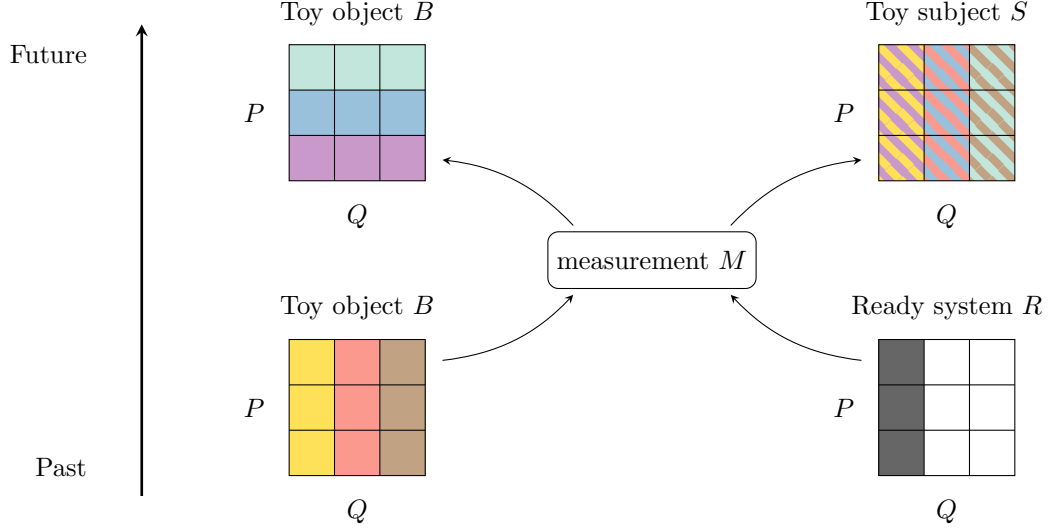
\begin{figure}[htb]\centering
	\begin{tikzpicture}[>=stealth]
		\begin{scope}[shift={(0,0)}]
			\node at (1.5*\cellsize,3.7*\cellsize) {Toy object $B$};
			\foreach \y in {0,1,2}
			{
				\fill[gridyellow] (0,\y*\cellsize) rectangle ++(\cellsize,\cellsize);
				\fill[gridred] (1*\cellsize,\y*\cellsize) rectangle ++(\cellsize,\cellsize);
				\fill[gridbrown] (2*\cellsize,\y*\cellsize) rectangle ++(\cellsize,\cellsize);
			}
			\drawgrid
			\panelaxes
			\node (past object) at (3*\cellsize,2.5*\cellsize) {};
		\end{scope}
		
		\begin{scope}[shift={(2*\panelsep,0)}]
			\node at (1.5*\cellsize,3.7*\cellsize) {Ready system $R$};
			\foreach \y in {0,1,2}
			{
				\fill[gridgray] (0,\y*\cellsize) rectangle ++(\cellsize,\cellsize);
				\fill[white] (1*\cellsize,\y*\cellsize) rectangle ++(\cellsize,\cellsize);
				\fill[white] (2*\cellsize,\y*\cellsize) rectangle ++(\cellsize,\cellsize);
			}
			\drawgrid
			\panelaxes
			\node (past ready system) at (0*\cellsize,2.5*\cellsize) {};
		\end{scope}
		
		\begin{scope}[shift={(0,1*\panelsep)}]
			\node at (1.5*\cellsize,3.7*\cellsize) {Toy object $B$};
			\foreach \y in {0,1,2}
			{
				\fill[gridpurple] (\y*\cellsize,0*\cellsize) rectangle ++(\cellsize,\cellsize);
				\fill[gridteal] (\y*\cellsize,1*\cellsize) rectangle ++(\cellsize,\cellsize);
				\fill[gridgreenlight] (\y*\cellsize,2*\cellsize) rectangle ++(\cellsize,\cellsize);
			}
			\drawgrid
			\panelaxes
			\node (future object) at (3*\cellsize,0.5*\cellsize) {};
		\end{scope}

		\begin{scope}[shift={(2*\panelsep,1*\panelsep)}]
			\node at (1.5*\cellsize,3.7*\cellsize) {Toy subject $S$};
			\foreach \y in {0,1,2}
			{
				\filltwo{0*\cellsize}{\y*\cellsize}{\cellsize}{gridpurple}{gridyellow};
				\filltwo{1*\cellsize}{\y*\cellsize}{\cellsize}{gridteal}{gridred};
				\filltwo{2*\cellsize}{\y*\cellsize}{\cellsize}{gridgreenlight}{gridbrown};
			}
			\drawgrid
			\panelaxes
			\node (future subject) at (0*\cellsize,0.5*\cellsize) {};
		\end{scope}
		
		\begin{scope}[shift={(1*\panelsep,0.5*\panelsep)}]
			\node[draw, rounded corners,
				rectangle,
				minimum width=3*\cellsize,
				minimum height=0.75cm,
				align=center,
				name=measurement
				] at (1.5*\cellsize,1.5*\cellsize) {measurement $M$};
				
			\draw[->, shorten >=3pt, shorten <=3pt] (past object) to [bend right=20] ($(measurement.south west)!0.15!(measurement.south east)$);
			\draw[->, shorten >=3pt, shorten <=3pt] (past ready system) to [bend left=20] ($(measurement.south west)!0.85!(measurement.south east)$);
			\draw[->, shorten >=3pt, shorten <=3pt] ($(measurement.north west)!0.15!(measurement.north east)$) to [bend right=20] (future object);
			\draw[->, shorten >=3pt, shorten <=3pt] ($(measurement.north west)!0.85!(measurement.north east)$) to [bend left=20] (future subject);
		\end{scope}
		
		\begin{scope}[shift={(-0.5*\panelsep,0.5*\panelsep+1.5*\cellsize)}]
			\draw[->,line width=1pt] (0,-0.8*\panelsep) -- (0,0.8*\panelsep);
			\node[left] at (-1*\cellsize,-0.7*\panelsep) {Past};
			\node[left] at (-1*\cellsize,0.7*\panelsep) {Future};
		\end{scope}
	\end{tikzpicture}
	\caption{Illustration of a measurement interaction, with individual state spaces as in \cref{fig:variables}. 
	The manifest variable of the toy subject is $\pi_\Q$ and subject's states with identical records are represented by matching colours.
	The ready state variability $\V$ is $\P$ (allowed initial states are in grey). 
	Matching colours between the subject and the two instances of the object depict the correlation facilitated by the interaction.
	Namely, the position of the subject being $0$ (yellow/purple) implies that the initial position of the object is $0$ (yellow) and also that the final momentum of the object is $0$ (purple), and similarly for other colours.
	The inaccessible subspaces are $\I_\ret = \P$ and $\I_\pre = \Q$.
	The retrodicted variable is thus $\pi_\Q$ and the predicted variable is $\pi_\P$.}
	\label{fig:interaction}
\end{figure}

\paragraph{Ready states.} 
The learning capabilities of embedded observers also depend on the resources of initial information available to them.
To give a full account of an observer for our purposes, we need to specify what information is available about the ready system $R$, which is used as the input of a measurement interaction (cf.\ \cref{fig:interaction}).

The \newterm{ready state} is defined by a subspace $\V$ of the ontic state space $\R$.
Essentially, $\V$ represents the set of possible initial states of $\R$, which is why we also refer to it as the ready state variability.
Intuitively, the smaller the variability, the greater the initial resource of non-uniformity \cite{La61,Gour2015}, which enables more precise learning. 
We distinguish three types of ready states based on `how large' their variability is: 
\begin{enumerate}
	\item A \newterm{trivial} ready state corresponds to no constraint on the initial state, i.e.\ the variability is the whole space, $\V=\R$. 
	
	\item For a ready state of \newterm{Poisson} type, the variability is the kernel of a Poisson variable.\footnote{By \cref{lem:Poisson}, this means the variability of a Poisson ready state is an arbitrary coisotropic subspace of $\R$.}
	
	\item The most powerful is a \newterm{complete} ready state, whose variability is the trivial subspace, $\V=\{0\}$. 
	The initial ontic state is fixed. 
\end{enumerate}

A trivial ready state contains no information{\,---\,}any initial state is possible{\,---\,}so that $\V$ would be given by the blue region of \cref{fig:variables} (a). 
A Poisson ready state corresponds to partial information; the blue regions of \cref{fig:variables} (a), (b) and (c) provide examples thereof.
Finally, a complete ready state expresses perfect information about the initial state, as is the case for the blue region of \cref{fig:variables} (d).  

When we say that a toy subject has access to ready states of a certain type, we mean that every conceivable ready system $R$ has an initial state whose variability $\V$ is of that type. 

It is worth mentioning that all processes in nomic toy theory are surjective (\cref{lem:ops_surj}) and thus they do not lead to any information gain. Indeed, no allowed physical transformation can turn a ready state into a more informative one; the dimension of ready state variability cannot decrease. 

\subsection{Learning}\label{sec:learning_types}

We conceptualise learning about a system $B$, which we call the \newterm{toy object}, as a physical interaction between $B$ and the world.
Because we allow discarding among physical transformations, we do not assume that the input and output of this process must coincide. 
Specifically, the environmental input is a ready system $R$, which supplies the resource of initial information in the form of its ready state.
The environmental output is a toy subject $S$, whose final state carries an empirical record determined by its manifest variable $E$.\footnotemark{}
\footnotetext{In a concrete measurement model, $S$ may be identified as a subsystem of $R$. 
Degrees of freedom in $\R$ complementary to $\S$ are then viewed as irrelevant to the subject's record.}%
See \cref{fig:interaction} for a schematic representation.

\paragraph{Measurements.} 
A \newterm{measurement interaction} is consequently a physical transformation of the form
\begin{equation}
	M \colon \B \oplus \R \to \B\oplus \S.
\end{equation}
When we refer to a measurement $M$, the ready system $R$ comes with a choice of ready state variability $\V$, and the toy subject $S$ comes with a choice of manifest variable $E$ (see \cref{sec:toy_subjects}).

The initial state of the object is left unconstrained, in contrast to that of the ready system, because we focus on learning without prior information about the object. 
In line with our aim of making learning resources explicit, any such prior information would itself have to be accounted for by an earlier learning process and included as part of $M$, rather than left as an implicit background condition.

Specific measurement interactions that play an important role in our results are collected in \cref{tab:basic-interactions}.

\begin{table}[tb]
	\centering
	\small
	\renewcommand{\arraystretch}{1.25}
	\begin{tblr}{
			width=\textwidth,
			colspec={Q[c,m] X[l,m] X[l,m] Q[c,m]},
			row{1}={font=\bfseries},
			vline{2}={-}{},
			hline{2}={-}{},
		}
		Interaction & Action & Main use & Definition \\
		
		$M^{\rm Pson}$
		&
		Encodes the value of a Poisson variable in the record
		&
		Repeatable learning of compatible properties
		&
		\eqref{eq:M_Pson}
		\\
		
		$M^{\rm swap}$
		&
		Exchanges the object and ready system
		&
		Optimal retrodiction, without repeatability
		&
		\eqref{eq:swap}
		\\
		
		$M^{\rm full}$
		&
		Encodes the initial ontic state in the record
		&
		Optimal repeatable learning (for a complete ready state)
		&
		\eqref{eq:M_full}
		\\
		
		$M^{\rm pred}$
		&
		Encodes the final ontic state in the record
		&
		Optimal prediction (for a complete manifest variable)
		&
		\eqref{eq:M_pred}
	\end{tblr}
	\caption{The main measurements we use and their role in our analysis of epistemic horizons.}
	\label{tab:basic-interactions}
\end{table}

We associate two linear maps
\begin{align}
	\label{eq:beta} \beta &\colon \B \oplus \V \to \B;  &  b\oplus v &\mapsto \pi_\B \comp M (b\oplus v),\\
	\label{eq:epsilon} \epsilon &\colon \B \oplus \V \to \E;  &  b\oplus v &\mapsto E \comp \pi_\S \comp M(b\oplus v),
\end{align} 
to each measurement $M$, where $\pi_{\B}$ and $\pi_\S$ are projections that reduce the output of $M$ to $\B$ and $\S$, respectively.
Intuitively, $\beta$ describes the disturbance of the toy object by the measurement when the ready system is restricted to initial states specified by the ready state variability $\V$. 
Similarly, $\epsilon$ describes how the subject's record depends on the joint initial state prior to the measurement interaction.

\paragraph{Subject's knowledge.}
In our setup for learning, we assume that the subject has certain background knowledge. In particular, the subject knows the ontic state spaces of the systems involved, the interaction $M$, and the set $\V$ of allowed initial states of the ready system. 
We call this the subject's \emph{non-empirical} knowledge, to distinguish it from the \emph{empirical} knowledge inferred from its record after the measurement. 

\paragraph{Types of learning.}
Each measurement interaction $M$ facilitates two basic types of learning.
Namely, the subject may learn about the initial state of $B$, which we refer to as \newterm{retrodiction}.
Simultaneously, it may obtain information about the final state of the object and thus engage in \newterm{prediction}. 
In the special case when $M$ is \newterm{repeatable}, retrodicted and predicted information coincides and thus we do not need to distinguish them. 

A crucial feature of $M$ that we can use to characterise learned information is the kernel \mbox{of $\epsilon$}, which consists of those variations of initial states that do not influence the subject's record.
That is, two elements of $\B \oplus \V$ result in the same record if and only if they differ by an element of $\ker(\epsilon)$. 
This subspace of $\B \oplus \V$ thus consists precisely of those joint variations of the object and ready system that remain invisible to the subject.
For example, in \cref{fig:interaction}, coloured subsets of $\S$ indicate distinct records.
Some pairs of joint initial states produce the same record after the interaction and thus cannot be distinguished given a value of the record.

\paragraph{Retrodiction.}
Projecting this subspace onto the state space of the object gives the \newterm{retrodictively inaccessible subspace}
\begin{equation}
	\I_\ret \coloneq \pi_\B(\ker\epsilon)
\end{equation}
that tells us which initial ontic states of $B$ cannot be distinguished by the subject.
The retrodictive information learned by the subject is then characterised by\footnote{See \cref{prop:concrete_measurable} for this characterisation.} what we call the variable \newterm{retrodicted by} $M$, given by the quotient map
\begin{equation}
		L_\ret \colon \B \to \newfaktor{\B}{\I_\ret},  \qquad \qquad  L_\ret(b) \coloneq \Set{ b + i  \given  i \in \I_\ret }.
\end{equation}
For example, in \cref{fig:interaction}, changing only the initial momentum of the object does not change the subject's record, i.e.\ we have $\I_\ret=\P$. 
The subject can therefore use such a measurement to learn the initial position of the object, but not the full initial ontic state.
The retrodicted variable is equivalent to the projection $\pi_\Q$.

\paragraph{Prediction.}
Unlike the initial state of the object, its final state may be constrained by the subject's non-empirical knowledge.
In particular, if the ready state variability is not the whole of $\R$, then it is conceivable that not all elements of $\B$ can occur as possible final states after the measurement interaction.
The set of those that can is precisely the image of $\beta$, which we call the \newterm{preparation support} of $M$.
It thus characterises the non-empirical knowledge of the subject about the final state of $B$. 

Mapping the indistinguishable states in $\ker(\epsilon)$ to the future instance of the object then yields the \newterm{predictively inaccessible subspace}
\begin{equation}\label{eq:I_pre_main}
	\I_\pre \coloneq \beta(\ker\epsilon)
\end{equation}
that tells us which final ontic states of $B$ cannot be distinguished by the subject on the basis of its empirical knowledge.
The predictive information gathered by the subject is then described by the pair consisting of the preparation support and the predictively inaccessible subspace\footnote{See \cref{prop:fixed_from_measured} for this characterisation.}, which we can collect to define the variable \newterm{predicted by} $M$, as given by the quotient map
\begin{equation}\label{eq:predicted_main}
	L_\pre \colon \im(\beta) \to \newfaktor{\im(\beta)}{\I_\pre},  \qquad \qquad  L_\pre(b) \coloneq \Set{ b + i  \given  i \in \I_\pre },
\end{equation}
which captures the correlation between the post-measurement state of the object and the subject's record.
For the interaction in \cref{fig:interaction}, the subject's record distinguishes final momenta of the object but not its final positions. 
Hence, we have $\I_\pre = \Q$ and the predicted variable is equivalent to $\pi_\P$. 
The preparation support is the set of final object states that can arise from an arbitrary initial state of $\B \oplus \V$, which in this example is the set of all final ontic state of $\B$. 

In general, the predicted variable partitions the preparation support into subsets, each of which contains all the final ontic states of the object compatible with a particular record.
See \cref{fig:preparation_support} for an illustration.

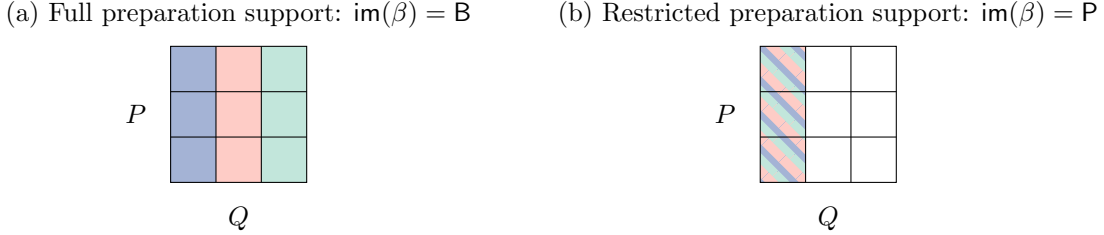
\begin{figure}[htb]\centering
	\begin{tikzpicture}
		\begin{scope}[shift={(0,0)}]
			\node at (1.5*\cellsize,3.7*\cellsize) {(a) Full preparation support: $\im(\beta) = \B$};
			
			\foreach \y in {0,1,2}
			{
				\fill[gridblue] (0,\y*\cellsize) rectangle ++(\cellsize,\cellsize);
				\fill[gridredlight] (1*\cellsize,\y*\cellsize) rectangle ++(\cellsize,\cellsize);
				\fill[gridgreenlight] (2*\cellsize,\y*\cellsize) rectangle ++(\cellsize,\cellsize);
			}
			
			\drawgrid
			\panelaxes
		\end{scope}
		
		\begin{scope}[shift={(2*\panelsep,0)}]
			\node at (1.5*\cellsize,3.7*\cellsize) {(b) Restricted preparation support: $\im(\beta) = \P$};
			
			\foreach \y in {0,1,2}
			{
				\fillthree{0}{\y*\cellsize}{\cellsize}{gridblue}{gridredlight}{gridgreenlight};
			}
			
			\drawgrid
			\panelaxes
		\end{scope}
	\end{tikzpicture}
	\caption[prepsupport]{An illustration of predicted variables in terms of diagrams from \cref{fig:variables}.
		Each of the three colours depicts the set of possible final ontic states of $B$ compatible with a given record, i.e.\ a given value of the subject's manifest variable.
		In both cases, the predictively inaccessible subspace is $\I_\pre = \P$.
		In (a) all ontic states of $B$ are possible post-measurement states{\,---\,}the preparation support is $\im(\beta) = \B$.
		The subject has no non-empirical knowledge about $B$ but has empirical knowledge of the position of $B$.
		The predicted variable is $\pi_\Q$, which is a Poisson variable.
		In (b) only states with zero position are possible post-measurement states{\,---\,}the preparation support is $\im(\beta) = \P$.
		The subject has no empirical knowledge about $B$ but has non-empirical knowledge of the final position of $B$.
		The predicted variable is the identity map on $\im(\beta)$, which is a Poisson variable again.
	}
	\label{fig:preparation_support}
\end{figure}

\paragraph{Repeatable measurements.}
Some measurement interactions, such as the swap operation introduced in \cref{eq:swap}, allow the subject to learn about the initial state of $B$, but this information is no longer present in the state of $B$ after the measurement.
Indeed, if a second subject were to perform the same measurement immediately afterwards, the information learned by the two subjects would be uncorrelated. 
The swapping interaction does not, in this sense, produce reliable records of the state of the object.

To capture the idea of the information learned about the past persisting in the object after the measurement, we introduce the notion of repeatability.
A measurement is said to be \newterm{repeatable} if it satisfies
\begin{equation}\label{eq:repeatable}
	L_{\rm ret}(b) = L_{\rm ret}\bigl( \beta(b \oplus v) \bigr),
\end{equation}
for all $b \in \B$ and all $v \in \V$.
The left-hand side expresses the information retrodicted by $M$, while the right-hand side is the information retrodicted by a second subject that applies $M$ immediately after the first one does.
As we show in \cref{sec:bitemporal}, variables retrodicted by a repeatable measurement can also be predicted (\cref{cor:retrodiction_implies_prediction}).
In this sense, repeatable measurements enable simultaneous learning about the past and the future.

\begin{remark}[Comparison with nomic toy theory presented in \cite{Fa24}]\label{rem:comparison}
	The formulation used in this paper generalises our earlier account \cite{Fa24} in several respects. 
	First, we consider physical transformations which include all discarding maps.
	This allows us to explicitly treat irreversible operations, and hence to distinguish the ready system $R$ entering the interaction from the subject $S$ carrying the manifest variable.
	
	Second, we consider subjects of many different kinds.
	In \cite{Fa24}, both the manifest variable and the ready state were necessarily Poisson.
	Moreover, we justified the constraint on the initial state by appealing to pre-selection by the subject.
	This is, however, a problematic assumption (see \cref{rem:ready_state_issues} for more details).
	Here, manifest variables and ready states are independent choices: 
	The manifest variable specifies subject's empirical records, while the ready state specifies which resources of information are available prior to the interaction.
	
	Finally, some of the terminology has changed. 
	What we called a measured variable in \cite{Fa24} corresponds here to a variable retrodicted by a measurement interaction. 
	More precisely, for the measurements considered there, \cref{prop:concrete_measurable} shows that the measured variable \cite[Definition 2.9]{Fa24} coincides with the retrodicted variable $L_{\rm ret}$. 
	The epistemic horizon in \cite[Theorem 3.1]{Fa24} is thus recovered here as the retrodictive case with Poisson manifest variables and Poisson ready states.
\end{remark}

\section{Resources of Learning and Epistemic Horizons}\label{sec:EH}

As outlined in \cref{fig:domains_for_epistemic_horizons}, limitations on what a subject can learn about the world stem from several assumptions. 
To summarise, besides the physical theory and its allowed interactions, the relevant ones are
\begin{enumerate}
	\item the \emph{learning-type} assumption about the kind of information learned by the subject,
	
	\item the \emph{ready state} assumption about the nature of information sources available to the subject prior to learning, and
	
	\item the \emph{manifest variable} assumption specifying which distinctions in the subject's physical state carry the empirical record.
\end{enumerate}

We consider three options for each of these assumptions.
As detailed in \cref{sec:learning_types}, for a generic measurement the subject can either learn about the past or the future. 
Additionally, we can restrict the interaction to repeatable measurements, in which case the subject learns about both past and future simultaneously.
These correspond to the three learning types: retrodictive, repeatable, and predictive learning.
The three options for allowed ready states and for manifest variables have been introduced in \cref{sec:toy_subjects}.
They can be trivial, Poisson, or complete.

In each of the 27 cases, we ask what information the corresponding toy subject can learn about toy objects.
A fundamental limitation on learning is called an \newterm{epistemic horizon}.
We find three broad types of epistemic horizons (see \cref{fig:EH_pictures}):
\begin{enumerate}
	\item[(a)] In the case of the \newterm{total} epistemic horizon, the subject cannot learn anything, i.e.\ the learned variables are constant.
	
	\item[(b)] In the \newterm{knowledge-balance} epistemic horizon,\footnotemark{} a variable can be learned if and only if it is a Poisson variable.
	In particular, the subject can learn at most half of the degrees of freedom of the toy object. 
	\footnotetext{The name refers to the \emph{knowledge-balance principle} \cite{Sp07}, because of the connection between Poisson variables and epistemic states in Spekkens' toy theory (\cref{sec:variables}).}%
	
	\item[(c)] There is \newterm{no} epistemic horizon when the subject can learn the precise ontic state of the object. 
\end{enumerate}

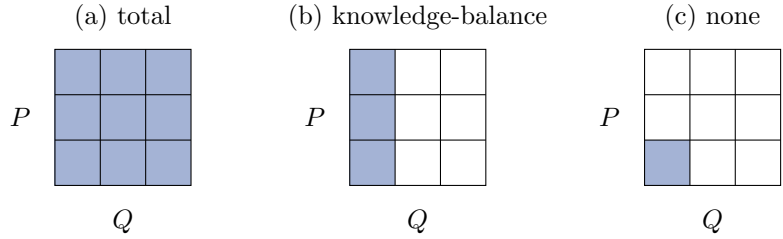
\begin{figure}[htb]\centering
	\begin{tikzpicture}		
		\begin{scope}[shift={(0,0)}]
			\node at (1.5*\cellsize,3.7*\cellsize) {(a) total};
			
			\fill[gridblue] (0,0) rectangle ++(3*\cellsize,3*\cellsize);
			
			\drawgrid
			\panelaxes
		\end{scope}
		
		\begin{scope}[shift={(\panelsep,0)}]
			\node at (1.5*\cellsize,3.7*\cellsize) {(b) knowledge-balance};
			
			\foreach \y in {0,1,2}
			\fill[gridblue] (0,\y*\cellsize) rectangle ++(\cellsize,\cellsize);
			
			\drawgrid
			\panelaxes
		\end{scope}
		
		\begin{scope}[shift={(2*\panelsep,0)}]
			\node at (1.5*\cellsize,3.7*\cellsize) {(c) none};
			
			\fill[gridblue] (0,0) rectangle ++(\cellsize,\cellsize);
			
			\drawgrid
			\panelaxes
		\end{scope}		
	\end{tikzpicture}
	\caption[EH]{Three types of epistemic horizons (total, knowledge-balance and none) in terms of the diagrams from \cref{fig:variables}.
		The shaded region depicts the set of possible ontic states of the toy object compatible with the subject's record. 
		(a) If the epistemic horizon is total, the toy subject cannot learn anything. 
		(b) If it is knowledge-balance, the record can determine the value of a Poisson variable, such as $\pi_\Q$ in this example.
		(c) If there is no epistemic horizon, the subject can learn the precise ontic state of the toy object.}
	\label{fig:EH_pictures}
\end{figure}

Our results are summarised in \cref{tab:retrodiction,tab:repeatable,tab:prediction}. 
We omit the cases with trivial manifest variables, because none of them allows the subject to learn anything{\,---\,}they give rise to a total epistemic horizon irrespective of any other assumptions. 

This classification of epistemic horizons into three types does not tell us everything about what the subjects can and cannot learn.
For example, in the case of prediction, it does not specify what the preparation support is and therefore cannot distinguish the two cases depicted in \cref{fig:preparation_support}.
We provide a detailed account of our results including these nuances in \cref{sec:retrodiction,sec:repeatable,sec:prediction}.

\subsection{Retrodiction}\label{sec:retrodiction}

We begin with a presentation of epistemic horizons for retrodiction in nomic toy theory, as summarised in \cref{tab:retrodiction}. 
The proofs of all results in this section are provided in Appendix~\ref{sec:proofs_retrodictive}.  

\begin{table}[htb]
	\centering
	\renewcommand{\arraystretch}{1.5}~
	\begin{tabular}{cl|cc}
		\multicolumn{2}{c|}{\multirow{2}{*}{\makecell{\textbf{Retrodiction}}}} & \multicolumn{2}{c}{\textbf{Manifest variable}} \\
		\multicolumn{2}{c|}{} & \textit{Poisson} & \parbox{\widthof{knowledge-balance (\ref{thm:repeatable_complete_MV_Poisson_RS})}}{\centering \textit{Complete}} \\ 
		\hline
		\multirow[m]{3}{*}{\begin{sideways}\makecell{\textbf{Ready} \\ \textbf{state}}\end{sideways}} & \textit{Trivial} & \tableyellow knowledge-balance (\ref{thm:Poisson_MV_trivial_RS}) & \tablegreen none (\ref{cor:complete_MV_retrodictive})  \\
		& \textit{Poisson} & \tableyellow knowledge-balance (\ref{thm:Poisson_MV_Poisson_RS}) & \tablegreen none (\ref{cor:complete_MV_retrodictive})  \\
		& \textit{Complete} & \tablegreen none (\ref{thm:Poisson_MV_fixed_RS})  & \tablegreen none (\ref{cor:complete_MV_retrodictive}) 
	\end{tabular}
	\caption{Epistemic horizons for retrodiction in nomic toy theory with links to the theorems.}
	\label{tab:retrodiction}
\end{table}

\paragraph{Poisson manifest variable.}
First, we consider the case of partial self-knowledge in the form of subjects with Poisson manifest variables.
In the context of Poisson ready states, we showed in \cite[Theorem 3.1]{Fa24} that a variable can be retrodicted if and only if it is Poisson.
Since the set-up of nomic toy theory here differs slightly from that of \cite{Fa24} (see \cref{rem:comparison}), we include the new proof in full in the Appendix.

\begin{theorem}[Poisson ready state]\label{thm:Poisson_MV_Poisson_RS}
	For subjects with Poisson manifest variables and Poisson ready states, a variable can be retrodicted if and only if it is Poisson.
\end{theorem} 

One might expect that replacing the Poisson ready state by a trivial one would further restrict what the subject can retrodict. 
This is not the case, as can be seen from a simple swapping interaction. 
If the only aim is to correlate the object's initial state with the subject's final state, one can simply swap the states of the object and the subject. 

Specifically, suppose that the object $B$, the ready system $R$, and the subject $S$ have the same ontic state space. The interaction $M^{\rm swap} \colon \B \oplus \R \to \B \oplus \S$ that swaps\footnotemark{} the two inputs is 
\footnotetext{For our purposes, it is irrelevant whether $M^{\rm swap}$ exchanges the physical systems or just the information carried by their ontic states. To establish the epistemic horizon, all that matters is that swapping is a valid transformation in nomic toy theory.}%
\begin{equation}\label{eq:swap}
	\begin{tikzpicture}
	\begin{pgfonlayer}{nodelayer}
		\node [style=none] (0) at (1.75, -1.125) {};
		\node [style=none] (1) at (1.75, -0.875) {};
		\node [style=none] (2) at (1.75, 0.875) {};
		\node [style=none] (3) at (1.75, 1.125) {};
		\node [style=none] (4) at (3.75, 0.875) {};
		\node [style=none] (5) at (3.75, 1.125) {};
		\node [style=none] (6) at (3.75, -0.875) {};
		\node [style=none] (7) at (3.75, -1.125) {};
		\node [style=none] (8) at (1.75, -1.5) {$\B$};
		\node [style=none] (9) at (1.75, 1.5) {$\B$};
		\node [style=none] (10) at (3.75, 1.5) {$\S$};
		\node [style=none] (11) at (3.75, -1.5) {$\R$};
		\node [style=none] (12) at (0, 0) {$=$};
		\node [style=morphism] (13) at (-2.75, 0) {$\; M^{\rm swap} \;$};
		\node [style=none] (14) at (-3.5, 1.125) {};
		\node [style=none] (15) at (-2, 1.125) {};
		\node [style=none] (16) at (-3.5, 1.5) {$\B$};
		\node [style=none] (17) at (-2, 1.5) {$\S$};
		\node [style=none] (18) at (-3.5, -1.125) {};
		\node [style=none] (19) at (-2, -1.125) {};
		\node [style=none] (20) at (-3.5, -1.5) {$\B$};
		\node [style=none] (21) at (-2, -1.5) {$\R$};
	\end{pgfonlayer}
	\begin{pgfonlayer}{edgelayer}
		\draw (0.center) to (1.center);
		\draw [in=-90, out=90] (1.center) to (4.center);
		\draw (4.center) to (5.center);
		\draw [style=protected, in=-90, out=90] (6.center) to (2.center);
		\draw [style=protected] (7.center) to (6.center);
		\draw [style=protected] (2.center) to (3.center);
		\draw [style=protected] (18.center) to (14.center);
		\draw [style=protected] (19.center) to (15.center);
	\end{pgfonlayer}
\end{tikzpicture}
 \qquad \qquad \text{i.e.\ we have} \quad M^{\rm swap}(b \oplus r) = r \oplus b.
\end{equation}
This establishes a perfect correlation between the ontic state of the object before the interaction and that of the subject after the interaction, \emph{irrespective} of the ready state. 
But it also resets the object such that the final state of $B$ carries no information about its initial state.

Since the subject's record is constrained to be the value of a Poisson manifest variable, it can nevertheless encode only coarse-grained information about the object's initial state. 
We therefore obtain the same knowledge-balance epistemic horizon.

\begin{theorem}[trivial ready state]\label{thm:Poisson_MV_trivial_RS}
	For toy subjects whose manifest variable is Poisson and ready state is trivial, a variable can be retrodicted if and only if it is Poisson.
\end{theorem}

A subject with access to trivial ready states thus cannot learn more about the object than one with access to a Poisson ready state. 
Conversely, we would expect that a subject with access to complete ready states can learn at least as much as one with access to Poisson ready states.
Surprisingly, this is not the best one can do. 
As the following theorem shows, \emph{any} information can in fact be retrodicted in this case, even though the manifest variable is restricted to be Poisson.

\begin{theorem}[complete ready state]\label{thm:Poisson_MV_fixed_RS}
	For toy subjects whose manifest variable is Poisson and ready state is complete, every surjective linear variable can be retrodicted.
\end{theorem}

Since the identity variable $\id_\B$ is surjective and linear, such subjects have the ability to retrodictively learn the precise ontic state of any toy system. 
In other words, there is no epistemic horizon (cf.\ \cref{tab:retrodiction}).
We thus find that the epistemic horizon for retrodiction disappears when subjects have access to complete ready states, even if the manifest variable remains Poisson. 

It is intriguing to note that such subjects can retrodict arbitrary information about other toy systems even though they only have partial self-knowledge: 
Their records are given by Poisson manifest variables, which are necessarily coarse-grained. 
The apparent tension disappears once one accounts for the dimensions of the systems involved. 
For example, a subject with an ontic state space of dimension $4n$ can, by the construction of $M^{\rm full}$ in \eqref{eq:M_full}, retrodict the full ontic state of any toy object of dimension at most $2n$. 
The same subject cannot do this for larger objects, since the value of a Poisson manifest variable can encode at most half of the independent degrees of freedom of the subject's ontic state space. 
In particular, such a record cannot uniquely specify the subject's own ontic state.

\paragraph{Complete manifest variable.}

When the subject's manifest variable is the identity map, its record is in one-to-one correspondence with its ontic state; i.e.\ the subject has complete self-knowledge. 
In this case, we can induce a perfect correlation via the swapping interaction $M^{\rm swap}$ from \eqref{eq:swap}.
We thus obtain that there is no epistemic horizon, regardless of the ready states available (last column of \cref{tab:retrodiction}).

\begin{corollary}[trivial, Poisson, and complete ready states]\label{cor:complete_MV_retrodictive}
	For toy subjects with a complete manifest variable, the identity variable of any toy object can be retrodicted, independently of whether trivial, Poisson, or complete ready states are used. 
\end{corollary}

This follows from \cref{thm:complete_MV_retrodictive}, which shows that every surjective linear variable whose kernel is a symplectic subspace of $\B$ can be retrodicted. 
This includes the identity variable since $\ker(\id_\B)=\{0\}$ is a symplectic subspace. 
Hence, the subject can retrodict the precise initial ontic state of the object and every other variable of $B$ is available by post-processing (cf.\ \cref{prop:fixed_from_measured}).

Note that the swap interaction used above erases the ontic state of the object, i.e.\ after the interaction, the state of $B$ coincides with the initial state of the ready system. 
If the same measurement were applied again, it would therefore not reproduce the same outcome. 
Such non-repeatable measurements are unsatisfactory if a measurement is meant to produce a reliable record of the state of a system. 
We now turn to repeatable measurements, which impose this additional robustness requirement.

\subsection{Repeatable Measurements}\label{sec:repeatable}
Our aim is to identify properties that can be learned in a way that they stay true about the object also after the interaction.
That is, applying the same measurement in succession should yield the same outcome. 
\Cref{tab:repeatable} summarises the limitations on retrodiction in nomic toy theory, under the additional assumption that the measurement interaction used is repeatable.
All proofs of results from this subsection are given in \cref{sec:proofs_repeatable}.

\begin{table}[htb]
	\centering
	\renewcommand{\arraystretch}{1.5}
	\begin{tabular}{cl|cc}
		\multicolumn{2}{c|}{\multirow{2}{*}{\makecell{\textbf{Repeatable} \\ \textbf{measurements}}}} 
		& \multicolumn{2}{c}{\textbf{Manifest variable}} 
		\\
		\multicolumn{2}{c|}{} 
		& \textit{Poisson} 
		& \textit{Complete} 
		\\ \hline
		\multirow[m]{3}{*}{\begin{sideways}\makecell{\textbf{Ready} \\ \textbf{state}}\end{sideways}} 
		& \textit{Trivial} 
		& \tablered total (\ref{thm:repeatable:Poisson_MV_trivial_RS}) 
		& \tablered total (\ref{thm:repeatable:Poisson_MV_trivial_RS})  
		\\
		& \textit{Poisson} 
		& \tableyellow knowledge-balance (\ref{thm:repeatable_Poisson_MV_Poisson_RS})
		& \tableyellow knowledge-balance (\ref{thm:repeatable_complete_MV_Poisson_RS}) 
		\\
		& \textit{Complete} 
		& \tablegreen none (\ref{thm:repeatable_Poisson_MV_complete_RS})
		& \tablegreen none (\ref{thm:repeatable_complete_MV_complete_RS})
	\end{tabular}
	\caption{Epistemic horizons for learning by repeatable measurements in nomic toy theory with links to the theorems.}
	\label{tab:repeatable}
\end{table}

\paragraph{Poisson manifest variable.}
The swap interactions introduced in \eqref{eq:swap} are indeed not repeatable, unless the manifest variable is trivial.
To see this, note that for this interaction (and manifest variable $E$) we have 
\begin{equation}
	\beta(b \oplus v) = v  \qquad \text{and} \qquad L_{\rm ret} = E,
\end{equation}
so that \cref{eq:repeatable} can only be satisfied if $E$ is constant.
This suggests that subjects with trivial ready states may not be able to learn information repeatably. 
Indeed, the following result establishes that such subjects face a total epistemic horizon (cf.\ \cref{tab:repeatable}).

\begin{theorem}[trivial ready state]
	\label{thm:repeatable:Poisson_MV_trivial_RS}
	For toy subjects whose manifest variable is Poisson and ready state is trivial, a variable can be repeatably retrodicted if and only if it is constant.
	The same holds for subjects with complete manifest variables.
\end{theorem}

For Poisson ready states, the repeatable interaction $M^{\rm Pson}$ defined via \cref{eq:M_Pson} can be used to encode the value of an arbitrary Poisson variable of the object in the subject's record. 
Conversely, \cref{thm:Poisson_MV_Poisson_RS} shows that other variables cannot be retrodicted, which also means that they cannot be repeatably retrodicted. 

\begin{theorem}[Poisson ready state]
	\label{thm:repeatable_Poisson_MV_Poisson_RS}
	For toy subjects whose manifest variable is Poisson and ready state is Poisson, a variable can be repeatably retrodicted if and only if it is Poisson.
\end{theorem}

Finally, we consider a subject with access to complete ready states. 
The situation is analogous to \cref{thm:Poisson_MV_fixed_RS}, since the measurement $M^{\rm full}$ used in its proof is also repeatable.

\begin{corollary}[complete ready state]
	\label{thm:repeatable_Poisson_MV_complete_RS}
	For toy subjects whose manifest variable is Poisson and ready state is complete, every surjective linear variable can be repeatably retrodicted.
\end{corollary}

\paragraph{Complete manifest variable.}
We now focus on subjects with complete self-knowledge, i.e.\ subjects whose records uniquely specify their own ontic states.
Arguably, the most important result in this paper is that even such subjects can face a non-trivial epistemic horizon.
\Cref{thm:repeatable_complete_MV_Poisson_RS} below says that when they have access to Poisson ready states, only Poisson variables can be repeatably retrodicted. 
This means that the restriction on self-knowledge in \cref{thm:Poisson_MV_Poisson_RS} is not necessary to obtain an epistemic horizon, as long as we are interested in simultaneous learning about the past and the future.

\begin{theorem}[Poisson ready state]
	\label{thm:repeatable_complete_MV_Poisson_RS}
	For toy subjects whose manifest variable is complete and ready state is Poisson, a variable can be repeatably retrodicted if and only if it is a Poisson variable. 
\end{theorem}

However, the limitation disappears when the subject has access to complete ready states.

\begin{corollary}[complete ready state]
	\label{thm:repeatable_complete_MV_complete_RS}
	For toy subjects whose manifest variable is complete and ready state is complete, every surjective linear variable can be repeatably retrodicted.
\end{corollary}

As we have mentioned previously, the assumption of repeatability rules out interactions that merely produce a record of the past.
Instead, it requires the retrodicted property to be present both in the final state of the toy object $B$ and in the subject's record.
In this sense, repeatably retrodicted variables describe information that can be retrieved without disturbance.
That is why they are also related to information that can be copied \cite[Definition B.9]{Fa24}.
The results of \cref{sec:repeatable} can therefore be interpreted as showing that the ability to broadcast information in nomic toy theory is tightly linked to the ready state resource and is not determined by the manifest variable alone. 
In particular, even for a complete manifest variable{\,---\,}so that the subject's record is its ontic state{\,---\,}the precise state of the object can be learned without disturbance only when complete ready states are available.
Epistemic horizons for robust information acquisition therefore do not follow merely as a consequence of limited self-knowledge.
They arise from a joint constraint involving the manifest variable, ready state resources, and repeatability. 

\subsection{Prediction}\label{sec:prediction}
We now turn to the future-oriented counterpart of retrodictive learning, namely prediction.
The correlations of interest are those between the subject's record{\,---\,}the value of its manifest variable{\,---\,}and the state of the object \textit{after} the measurement interaction. 
This completes the comparison of epistemic horizons across the different temporal notions of learning in nomic toy theory.

In contrast to retrodiction, and similarly to the repeatable case, prediction is sensitive to the disturbance induced by the measurement process itself.
However, unlike for repeatable measurements, the scope of allowed interactions between the subject and the object is unrestricted.

\begin{table}[htb]
	\centering
	\renewcommand{\arraystretch}{1.5}
	\begin{tabular}{cl|cc}
		\multicolumn{2}{c|}{\multirow{2}{*}{\makecell{\textbf{Prediction}}}} 
		& \multicolumn{2}{c}{\textbf{Manifest variable}} 
		\\
		\multicolumn{2}{c|}{} 
		& \textit{Poisson} 
		& \parbox{\widthof{knowledge-balance (\ref{thm:repeatable_complete_MV_Poisson_RS})}}{\centering \textit{Complete}}
		\\ \hline
		\multirow[m]{3}{*}{\begin{sideways}\makecell{\textbf{Ready} \\ \textbf{state}}\end{sideways}} 
		& \textit{Trivial} 
		& \tablered total (\ref{thm:prediction_Poisson_MV_trivial_RS})
		& \tablered total (\ref{thm:prediction_Poisson_MV_trivial_RS}) 
		\\
		& \textit{Poisson} 
		& \tableyellow knowledge-balance (\ref{thm:prediction_Poisson_MV_Poisson_RS})
		& \tablegreen none (\ref{cor:prediction_complete_MV_Poisson_RS}) 
		\\
		& \textit{Complete} 
		& \tablegreen none (\ref{thm:prediction_Poisson_MV_complete_RS})
		& \tablegreen none (\ref{thm:prediction_complete_MV_complete_RS})
	\end{tabular}
	\caption{Epistemic horizons for prediction in nomic toy theory with links to the theorems.}
	\label{tab:prediction}
\end{table}

Comparing \cref{tab:prediction} to \cref{tab:repeatable}, one can see that this nevertheless leads to epistemic limitations similar to those in the repeatable case.
One key difference is the absence of an epistemic horizon for subjects with complete manifest variables and access to Poisson ready states.

The other, more subtle difference is that prediction is characterised by both non-empirical and empirical information about the object, i.e.\ by the preparation support $\im(\beta)$ and the predictively inaccessible subspace $\I_\pre$, respectively.
For retrodiction, only the empirical part is relevant, which is why we could characterise the information learned purely by means of the retrodictively inaccessible subspace $\I_\ret$ in \cref{sec:retrodiction,sec:repeatable}.
All proofs of the following theorems can be found in \cref{sec:proofs_predictive}.

\paragraph{Poisson manifest variable.}
We begin by showing that prediction is impossible in the case of trivial ready states, independently of the subject's manifest variable.
This fact actually holds beyond the specifics of nomic toy theory.
It is a feature of any physical theory whose dynamics are reversible (or indeed right-invertible, as here).
The intuitive reason is as follows.
For trivial ready states, the joint initial state of $\B \oplus \R$ is unconstrained.
For non-trivial prediction, there must be a correlation between the two systems $B$ and $S$ after the interaction, which constrains the possible joint states of $\B \oplus \S$.
If the interaction process is modelled by a surjective map, as is the case in nomic toy theory, this is impossible.

\begin{theorem}[trivial ready state]
	\label{thm:prediction_Poisson_MV_trivial_RS}
	For toy subjects with Poisson manifest variables and trivial ready states, only constant variables can be predicted.
	The same holds for subjects with complete manifest variables.
\end{theorem}
Thus, such subjects cannot have any information about the post-measurement state of the object.
They face a total epistemic horizon (cf.\ \cref{fig:EH_pictures} (a)) for prediction.

For Poisson ready states, we have encountered the knowledge-balance epistemic horizon both for retrodiction (\cref{thm:Poisson_MV_Poisson_RS}) and for repeatable learning (\cref{thm:repeatable_Poisson_MV_Poisson_RS}).
For repeatable measurements, retrodicted and predicted information are closely related (\cref{thm:prediction_repeatable_kernels}). 
Moreover, the repeatable interaction $M^{\rm Pson}$ used in this case has full preparation support (cf.\ \cref{fig:preparation_support}). 
This suggests that every Poisson variable should be predictable.
The following theorem says that this is true and that, in the case of full preparation support (\cref{fig:preparation_support} (a)), this is the best one can do.
As such, we obtain the knowledge-balance epistemic horizon also for prediction.

\begin{theorem}[Poisson ready state]
	\label{thm:prediction_Poisson_MV_Poisson_RS}
	For toy subjects with Poisson manifest variables and Poisson ready states, a partial variable can be predicted if and only if it is Poisson.\footnotemark{} 
\end{theorem}
\footnotetext{Poisson partial variables are introduced in \cref{def:partial_variable}.}%

In the general case, i.e.\ for a restricted preparation support (\cref{fig:preparation_support} (b)), \cref{thm:prediction_Poisson_MV_Poisson_RS} says that the predictively inaccessible subspace $\I_\pre$ must be coisotropic (\cref{def:complement}).
One of the consequences is that the dimension of $\I_\pre$ is necessarily at least half the dimension of $\B$.
Since this subspace describes the directions in the ontic state space $\B$ that cannot be distinguished by the subject, we obtain that subjects of this type face a non-trivial epistemic horizon for prediction regardless of the preparation support.

To gain intuition for which variables can be predicted by subjects as in \cref{thm:prediction_Poisson_MV_Poisson_RS}, let us consider a few examples. 
Both variables depicted in \cref{fig:preparation_support} can be predicted, since the kernel of each is given by the subspace $\P$, which is indeed coisotropic.
On the other hand, even though the variable depicted in \cref{fig:impossible_prediction} is merely a restriction of the variable from \cref{fig:preparation_support}~(a) to the subspace $\Q$, it cannot be predicted, because its kernel is trivial and so it is not a Poisson partial variable.
We thus see that restricting the preparation support may turn a variable that can be predicted into one that cannot.

\begin{figure}[htb]\centering
	\begin{tikzpicture}
		\begin{scope}[shift={(0,0)}]			
			\fill[gridblue] (0*\cellsize,0*\cellsize) rectangle ++(\cellsize,\cellsize);
			\fill[gridredlight] (1*\cellsize,0*\cellsize) rectangle ++(\cellsize,\cellsize);
			\fill[gridgreenlight] (2*\cellsize,0*\cellsize) rectangle ++(\cellsize,\cellsize);
			\drawgrid
			\panelaxes
		\end{scope}		
	\end{tikzpicture}
	\caption[prediction]{An illustration of a variable that cannot be predicted by a subject with access to Poisson ready states, in terms of diagrams from \cref{fig:variables}.
		It has a restricted preparation support, $\im(\beta) = \Q$, and trivial predictively inaccessible subspace, $\I_\pre = \{0\}$.
		Neither the relevant condition from \cref{thm:prediction_complete_MV_Poisson_RS} nor the one from \cref{thm:prediction_Poisson_MV_Poisson_RS} is satisfied.}      
	\label{fig:impossible_prediction}
\end{figure}
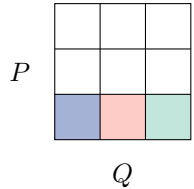

As already established, epistemic horizons for retrodiction disappear for complete ready states (\cref{thm:Poisson_MV_fixed_RS,thm:repeatable_Poisson_MV_complete_RS}).
Since the interaction $M^{\rm full}$, defined via \cref{eq:M_full}, used in these cases has full preparation support, \cref{thm:prediction_repeatable_kernels} implies the same holds for prediction.

\begin{corollary}[complete ready state]
	\label{thm:prediction_Poisson_MV_complete_RS}
	For toy subjects whose manifest variable is Poisson and ready state is complete, every surjective linear variable $\B \to \Z$ can be predicted.
\end{corollary}
In particular, since identity maps are surjective, for every toy object $B$ there exists a toy subject of the above type that can have perfect knowledge of the post-measurement ontic state of $B$.
Thus, there is no epistemic horizon (cf.\ \cref{fig:EH_pictures}).

\paragraph{Complete manifest variable.}
Finally, we consider prediction for subjects with a complete manifest variable. 
We already saw in \cref{thm:prediction_Poisson_MV_trivial_RS} that no prediction is possible in the case of trivial ready states. 

The case of Poisson ready states and complete manifest variables shows a key distinction between general retrodiction and learning via repeatable measurements.
Indeed, there is a non-trivial epistemic horizon for the latter (\cref{thm:repeatable_complete_MV_Poisson_RS}), but not for generic interactions (\cref{cor:complete_MV_retrodictive}).
In \cref{thm:prediction_complete_MV_Poisson_RS}, we show that prediction is not subject to the same limitation as learning by repeatable measurements.
In particular, the theorem has the following immediate corollary, which captures the case of full preparation support.

\begin{corollary}[Poisson ready state]
	\label{cor:prediction_complete_MV_Poisson_RS}
	For toy subjects whose manifest variable is complete and whose ready state is Poisson, every surjective linear variable $\B \to \Z$ can be predicted.
\end{corollary}

In particular, for every toy object $B$, there exists such a subject and a physical transformation that generates a perfect correlation between the final ontic state of $B$ and the subject's record.
The interaction that achieves this is $M^{\rm pred}$ given by \cref{eq:M_pred} (after setting $L = \id$ therein). 

The above corollary only tells us about variables that can be predicted in the case of full preparation support, i.e.\ in the case of no non-empirical knowledge (cf.\ \cref{sec:learning_types}).
The more intricate result (\cref{thm:prediction_complete_MV_Poisson_RS}) does, however, imply some constraints on prediction.
For example, it says that the preparation support must have dimension at least half the dimension of $\B$.
Thus it cannot, in general, be $\{0\}$ and so such subjects cannot prepare toy objects in a precise ontic state.

Another consequence of \cref{thm:prediction_complete_MV_Poisson_RS} is a trade-off between empirical and non-empirical knowledge achievable for prediction by such subjects. 
In order to have perfect knowledge of the final ontic state of the object, i.e.\ in order to achieve $\I_\pre = \{0\}$, we must  have full preparation support, i.e.\ $\im(\beta) = \B$.
The subject then has maximal empirical knowledge of the object and no non-empirical knowledge thereof.
On the other hand, to maximise the non-empirical knowledge under the constraint derived by the theorem, the preparation support must be a Lagrangian subspace of $\B$ (\cref{def:complement}), which implies $\I_\pre = \im(\beta)$.
Consequently, the subject has no empirical knowledge of the post-measurement state of the object.
The predicted variables in the two extreme cases described in this paragraph are like those depicted in \mbox{\cref{fig:variables} (d)} and \mbox{\cref{fig:preparation_support} (b)}, respectively.
An example of a variable that cannot be predicted according to \cref{thm:prediction_complete_MV_Poisson_RS} is shown in \cref{fig:impossible_prediction}.

These constraints may be relevant to the capabilities of such subjects. 
If we imagine a world in which ready states and measurement interactions can be engineered by the subject, then $\im(\beta) = \{0\}$ means that the subject can prepare a specific ontic state (e.g.\ $0 \in \B$) of the object with certainty.
In the above context where the preparation support has to be coisotropic, this is not possible.
Using the interaction $M^{\rm pred}$ allows the subject to infer the precise ontic state prepared from its record.
Nevertheless, it cannot manipulate the prepared state to match a given target ontic state. 
Indeed, if this were possible, it would effectively achieve a measurement with preparation support given by $\{0\}$, in contradiction to \cref{thm:prediction_complete_MV_Poisson_RS}.

Finally, we consider complete ready states. 
As expected, given the previous results, there is no epistemic horizon in this setting. 
The proof is essentially identical to that of \cref{thm:prediction_Poisson_MV_complete_RS}.

\begin{corollary}[complete ready state]
	\label{thm:prediction_complete_MV_complete_RS}
	For toy subjects whose manifest variable and ready state are both complete, every surjective linear variable $\B \to \Z$ can be predicted.
\end{corollary}
	
\section{Discussion and Outlook}\label{sec:conclusions}

In this section, we summarise our main conclusions and discuss their broader significance. 

\paragraph{Conditions for learning.}
An important lesson of this work is that epistemic horizons in nomic toy theory depend jointly on the subject's manifest variable, the available ready states, and the type of learning under consideration. 
The specific form of this dependence is summarised in \cref{tab:retrodiction,tab:repeatable,tab:prediction}. 
They show that an epistemic horizon is not a single phenomenon associated with a given physical theory, but also depends on the informational resources attributed to the embedded observer.

The epistemic horizon derived in our previous work \cite[Theorem 3.1]{Fa24} assumes partial self-knowledge as well as imperfect initial information about the ready system{\,---\,}the manifest variable and the ready state are both Poisson.
Whether this limitation persists under weaker assumptions was therefore left open. 
In this work we show that this restriction on self-knowledge is in fact not necessary:
A non-trivial epistemic horizon persists even for subjects with complete manifest variables, at least under the assumption of repeatability (\cref{thm:repeatable_complete_MV_Poisson_RS}).
By contrast, some restriction on the ready state is needed. 
For the two types of manifest variables considered, the limitation disappears whenever subjects have access to complete ready states. 
We therefore cannot regard partial self-knowledge alone as responsible for the epistemic horizon derived in \cite{Fa24}. 

Our model of learning also features a distinction between two modes of knowledge, both of which play an important role. 
On the one hand, the value of a manifest variable constitutes the empirical information available to the subject. 
On the other hand, information about the ready system, the measurement interaction, and the other features of the learning setup need not be physically encoded in the subject's record; we treat it as non-empirical background knowledge. 
A measurement interaction may provide the subject with new empirical information, but this constitutes learning about the world only in the presence of the relevant background knowledge.
Without it, the record could not be connected to a property of another system participating in the interaction.

\paragraph{Comparison of learning types.}
Our results reveal an asymmetry between knowledge of the past and knowledge of the future \cite{Al00,Wo24}: 
Retrodiction is generally less constrained than prediction. 
When a subject predicts a property of the object, there are effectively two copies of this information after the interaction. 
One is present in the final state of the object, which has the predicted property, and the other is encoded in the subject's manifest variable, whose value is correlated with that property. 
From this perspective, prediction amounts to broadcasting. 
In causal terms, successful prediction relies on correlations generated by common causes. 
A disturbance of the object can thus act as a resource for prediction, provided that it is a correlated disturbance that simultaneously affects the subject's record.
This is utilised in the interaction $M^{\rm pred}$ constructed in \cref{eq:M_pred}, which facilitates perfect prediction.

By contrast, retrodiction does not require the relevant information to be present in two systems at the same time. 
The correlation is instead between the initial state of the object and the final state of the subject. 
Retrodiction can therefore be implemented through a direct causal link even though the corresponding property is erased by the interaction, as for the swap operation $M^{\rm swap}$.

Repeatable measurements interpolate between the two extremes by connecting the retrodicted information to the measurement's effect on the object. 
The extra requirement of \cref{eq:repeatable} generally rules out both $M^{\rm pred}$ and $M^{\rm swap}$.
For the former, the disturbance which powers its predictive capability is the culprit, as it changes the retrodicted information of a successive measurement and thus spoils repeatability.
Similarly, while the disturbance of the swap is innocuous for retrodiction, it means that swapping cannot be repeatable.

As they form a restricted class of interactions, repeatable measurements impose constraints on learning at least as strong as retrodiction does. 
Since every variable retrodicted by a repeatable measurement is predictively fixed by the same interaction (\cref{cor:retrodiction_implies_prediction}), the same holds for prediction.

Notice also that for repeatable measurements and the two non-trivial types of manifest variables considered here, the epistemic horizon is determined entirely by the available ready states. 
Trivial ready states give rise to a total epistemic horizon, Poisson ready states to a knowledge-balance epistemic horizon, and complete ready states to no epistemic horizon (cf.\ \cref{tab:repeatable}).

However, one case stands out to highlight that we cannot think of epistemic horizons under repeatability as a simple intersection of the constraints for retrodiction and prediction respectively. 
Subjects with complete manifest variables and Poisson ready states face no epistemic horizon, both for retrodiction (\cref{cor:complete_MV_retrodictive}) and prediction (\cref{cor:prediction_complete_MV_Poisson_RS}). 
As shown in \cref{thm:repeatable_complete_MV_Poisson_RS}, such subjects can nevertheless repeatably retrodict only Poisson variables.
Therefore, they face a knowledge-balance epistemic horizon for learning via repeatable measurements. 

\paragraph{Relation to quantum theory.}
For certain classes of subjects, nomic toy theory recovers the knowledge-balance principle and thereby reproduces the epistemic features of Spekkens' toy model. 
In particular, at a knowledge-balance epistemic horizon, the maximal empirical content recoverable from the subject's record is the value of a Poisson variable of the object, which corresponds precisely to an epistemic state in Spekkens' toy theory. 
For odd-dimensional quantum systems, this in turn yields the stabiliser subtheory of quantum theory \cite{Ca17e}. 
The epistemic horizons derived here may therefore be viewed as giving rise to quantum-like features within a deterministic ontology without postulating the knowledge-balance principle directly \cite{Sp07}. 
This horizon does not come for free, however, since it depends on additional assumptions, most notably the fact that only Poisson ready states are available.

Extending the present analysis to full quantum theory is not straightforward. 
In nomic toy theory, the ontic state space and physical transformations are specified explicitly, and manifest variables and ready state resources are defined in terms of this structure. 
In quantum mechanics, by contrast, the nature of the ontic state space is itself a matter of interpretation, and the separability of subjects and objects cannot simply be taken for granted.

Bohmian mechanics provides a useful comparison. 
There, the joint configuration space of two systems is the Cartesian product of their individual configuration spaces, whereas their joint wave function need not be separable. 
Moreover, the standard arguments for Heisenberg uncertainty relations (so-called absolute uncertainty) and for the emergence of POVMs assume that information is ultimately grounded in configurations of physical systems \cite{DurrGoldsteinZanghi1992,DurrGoldsteinZanghi2004,Okon2025BohmianMechanics}. 
More generally, the distinction between the underlying Bohmian ontology and what can be accessed through measurement is central to the generic undetectability of Bohmian trajectories, as discussed in \cite{Fankhauser2025Undetectability}.
From the perspective of this work, grounding information in configurations amounts to selecting configurations as the physical variables carrying empirical records.
It therefore plays a role analogous to choosing a the position projection $\pi_\Q$ as the manifest variable in nomic toy theory. 

This comparison shows that the present definitions rely on ontological and epistemic aspects that are explicit in nomic toy theory but interpretation-dependent in quantum mechanics. 
Any generalisation must therefore specify how these structures are represented or replace them with more theory-independent notions.

\paragraph{Relation to perspectivalism.}
Nomic toy theory contains at least two levels of description of physical reality. 
At any given time, a physical system has a precise ontic state; this is the objective, third-person description provided by the theory. 
Relative to a subject's empirical record and background information, the system may also be assigned an epistemic state{\,---\,}the totality of subject's knowledge about the system. 
This gives a first-person description of how the system appears from the perspective of that subject.
Considering multiple subjective perspectives would require a further, intersubjective description that combines them \cite[Chapter 10]{Ad25b}. 
The present framework does not yet provide such a description, since our analysis considers arbitrary but only single subjects.
See also our discussion of multiple subjects and agency below.

On the other hand, thanks to the explicit way in which subjects and their knowledge are modelled, nomic toy theory provides a setting for studying how reality can appear from the perspective of a subject while retaining that a perspective-neutral description of the world exists in principle. 
It therefore falls under what Adlam calls ``moderate physical perspectivalism'' \cite{Ad24b}, as opposed to stronger forms of physical perspectivalism in which all facts are relativised.

Another approach to perspectives internal to a physical theory is provided by quantum reference frames \cite{An11,Va20}, where each reference frame is itself treated as a quantum system. 
However, the advantage of this approach, namely that it can be specified purely using notions from quantum theory at the operational level, comes at a cost.
The analogues of manifest variables and of ready state resources cannot be identified with the same granularity as in nomic toy theory. 
This may contribute to the underdetermination highlighted by inequivalent approaches and conceptual puzzles in the quantum-reference-frame literature \cite{Ca25,An11,Kr21,Br26}. 
Comparing these approaches with nomic toy theory may clarify which features of embedded observers necessitate ontological concepts and which can be formulated operationally. 
One concrete question is whether the classical version of the paradox of the third particle \cite{Ha26} can arise for subjects in nomic toy theory as well.

\paragraph{Justifying epistemic resources.}
One open question is why, for the manifest variables considered here, the resulting epistemic horizon for repeatable measurements is determined entirely by the available ready states. 
More generally, manifest variables and ready states are treated here as independent assumptions, and we derive the epistemic horizons that follow from different choices. 
It remains to be explored how these choices can be justified, and whether such justifications render the learning resources interdependent.

Independent grounds for particular ready states might come from thermodynamics, for example through low-entropy initial conditions, or from the evolutionary stability of systems whose records can reliably guide their subsequent behaviour (cf.\ \cref{rem:ready_state_resolution}). 
On the other hand, physically distinguished manifest variables may be identified by whether their values can serve as stable and reliable records of information, e.g.\ whether they can be copied.
A similar justification is given for the emergence of classical records in quantum Darwinism \cite{Bl06}.

\paragraph{Multiple subjects and agency.} \label{par:multiple_subjects}
Our analysis focused on a single embedded observer learning about an external system. 
Extending the framework to multiple subjects would allow one to ask under what conditions can different subjects consistently compare and combine their records and how information acquired by one subject should be represented from the perspective of another.

These questions become particularly salient in nested-observer scenarios such as Wigner's friend thought experiments \cite{schmid2023review, fankhauser2023epistemic}, where one subject is treated as part of the system observed by another. 
Nomic toy theory may provide a setting in which to examine how the epistemic restrictions of different observers relate to one another and to clarify how such restrictions should be applied. 

A further extension would be required to model agency more intricately.
Nomic toy theory, as described here, provides the infrastructure to study conditional operations and stable records, which have been argued to be a part of the minimal requirements for agency \cite{adlam2025agency}.
However, the act of choosing a measurement interaction in which the subject itself is partaking cannot be modelled within our current set-up.
See also \cite{Is25b} for the distinction between epistemic and agential attitudes.

\paragraph{Beyond nomic toy theory.}
A major direction for future work is the development of a more general framework. 
Category-theoretic approaches to measurement, such as the one developed in \cite{Fr26}, may provide a setting in which to reformulate the present concepts and compare epistemic horizons across different physical theories.

In this context, an ambitious question is whether quantum predictions could emerge as an effective description of the world according to some embedded observer.
This fits into a broader investigation of how epistemic limitations specified by an operational theory (e.g.\ quantum theory) jointly constrain the admissible fundamental physical theory (the ontological level), manifest variables, and ready states.
This differs from the usual task of constructing a hidden-variable model because the physical constitution of the observer and the informational resources available feature as a part of the explanation. 
Given a class of physical theories, a failure to recover the operational predictions motivates a revision of deeper assumptions of the model, such as deterministic dynamics or subject--object separability. 
The present paper does not answer such questions, but it provides a clear example in which some of the relevant concepts can be distinguished.

To extend the analysis of epistemic limitations beyond the current model of learning, which assumes a strict subject-object independence at the ontological level due to its roots in physicalism, one may take inspiration from the ideas of enactivism \cite{Va91}.
This relates to our discussion of agency, as enaction refers to the possibility that there are aspects of the subjective point of view are not merely passive observations.
Insofar as the metaphysics proposed by enactivism is related to that of QBism \cite{Ge24}, this may also be relevant for the path towards quantum epistemic horizons.

\section*{Acknowledgements}
J.F.\ gratefully acknowledges support from the European Union (ERC Advanced Grant, QuantAI, No. 101055129) and the Wittgenstein Award granted to H.\ J.\ Briegel [WIT9503323, DOI: 10.55776/WIT9503323]. 
The views and opinions expressed in this article are however those of the author(s) only and do not necessarily reflect those of the European Union or the European Research Council - neither the European Union nor the granting authority can be held responsible for them. 
T.G.\ has been funded by a grant from the Austrian Science Fund (FWF) 10.55776/ESP3451824.
We thank Robert W.\ Spekkens, Matthew Leifer, Y\`{\i}l\`{e} Y\={\i}ng, Phil A.\ LeMaitre, Antoine Soulas, and Lukas J.\ Fiderer for insightful discussions that have motivated and informed this work. 

Several aspects of the mathematical framework and results were discussed and brainstormed using Google Gemini 3.1 Pro. 
In particular, the construction of $M^{\rm pred}$ from \eqref{eq:M_pred} was suggested by Gemini. 

\bibliographystyle{abbrvnat}
\bibliography{biblio-epihorizons2}

\appendix

\section{Nomic Toy Theory: Formal Details}\label{sec:definitions}

To define a physical theory for the purposes of our study of epistemic horizons, we need to specify how to model 
\begin{itemize}
	\item individual physical systems,
	\item transformations between any two systems,
	\item a way to compose transformations in sequence and in parallel, and
	\item a special class of transformation that corresponds to discarding of (sub)systems.
\end{itemize} 
In mathematical terms, this amounts to a semicartesian (symmetric) monoidal category \cite{Ch10,Co11,Co14,Ho21,Co10}.

When we refer to \newterm{nomic toy theory}, we mean a specific choice of such a physical theory, which we introduce in \cref{sec:NTT} (see also \cref{sec:stage} for a less formal introduction).
In the terminology of \cref{fig:domains_for_epistemic_horizons}, nomic toy theory specifies the top left aspect of epistemic horizons{\,---\,}the fundamental physical theory.

Additionally, when we speak of a given nomic toy theory, we also understand that a choice of other aspects is given{\,---\,}namely, the supervenience domain and the ready state resource.
These two aspects correspond to different possible models of embedded observers with the physical theory (\cref{sec:subjects}).

For the purposes of studying epistemic horizons{\,---\,}limitations on learning for embedded observers{\,---\,}we need to further fix a specific type of learning for which to derive such a constraint.
We introduce the two main types of learning, about past and about future, in \cref{sec:learning}.

\Cref{sec:measured,sec:predicted} then include further details and results on retrodictive and predictive learning respectively. 
Finally, \cref{sec:bitemporal} establishes technical results about learning via repeatable measurements.

\subsection{Mathematical Preliminaries}\label{sec:math}

	In the interest of keeping the manuscript largely self-contained for readers unfamiliar with the terminology of symplectic linear algebra, we introduce a few standard concepts here.
	
	\begin{definition}[Symplectic vector space]\label{def:SVS}
		A symplectic vector space $\B$ over a field $\mathbb{F}$ is a vector space equipped with a symplectic form $\omega_{\B}$, i.e.\ a bilinear map $\B \times \B \to \mathbb{F}$ satisfying, for all $b \in \B$,
		\begin{equation}
			\omega_{\B}(b,b) = 0 \qquad \quad \text{and} \qquad \quad \bigl( \forall c \in \B \quad \omega_{\B}(b,c) = 0 \bigr) \;\; \implies \;\; b = 0.
		\end{equation} 
	\end{definition}
	A symplectic form is always skew-symmetric:
	\begin{equation}
		\omega_{\B}(b,c) = - \omega_{\B}(c,b),
	\end{equation}
	which also means that the dimension of $\B$ is even (whenever it is finite).
	Further, we sometimes use a chosen Darboux basis of $\B$, which provides a decomposition into two $n$-dimensional subspaces $\B = \Q \oplus \P$, such that the symplectic form $\omega_{\B}$ has the matrix representation $\Omega_{\B}$ from \cref{eq:sympl_form}.\footnote{Because of this choice, we generally do not distinguish between the symplectic form $\omega$ and its matrix representation $\Omega$.} 
	In \cref{eq:sympl_form}, we relate the symplectic form to a Poisson bracket between elements of $\B$.
	Strictly speaking, the Poisson bracket is defined on functionals, but we gloss over this distinction in this paper, because the chosen basis gives an isomorphism between $\B$ and its dual space.
	The subspaces $\Q$ and $\P$ are both Lagrangian (\cref{def:complement}).
	Moreover, the chosen Darboux basis provides a vector space isomorphism $\Q \cong \P$.
	
	\begin{definition}[Symplectic complement]\label{def:complement}
		Let $\V$ be a linear subspace of a symplectic vector space $\B$.
		We define the symplectic complement $\V^\omega$ to be the subspace
		\begin{equation}
			\V^\omega \coloneq \Set{ b \in \B  \given  \omega(b, v) = 0 \text{ for all } v \in \V }.
		\end{equation}
		The subspace $\V$ is 
		\begin{itemize}
			\item symplectic if $\V \cap \V^\omega = \{0\}$ holds,
			\item isotropic if $\V \subseteq \V^\omega$ holds,
			\item coisotropic if $\V^\omega \subseteq \V$ holds, and
			\item Lagrangian if we have $\V = \V^\omega$.
		\end{itemize}
	\end{definition}
	Note that if the dimension of $\B$ is $2n$, then:
	\begin{itemize}
		\item $\dim{\V} \leq n$ whenever $\V$ is isotropic.
		\item $\dim{\V} \geq n$ whenever $\V$ is coisotropic.
		\item $\dim{\V} = n$ whenever $\V$ is Lagrangian.
	\end{itemize}
	Symplectic subspaces are precisely those for which $\omega_\B$ restricts to a valid symplectic form on $\V$, so that they are symplectic vector spaces in their own right.
	
	Recalling the definition of a Poisson variable\footnotemark{} from \cref{eq:Poisson}, we can characterise such variables in terms of their kernel properties.
	\footnotetext{\label{foo:Poisson} The name stems from the connection to Poisson geometry.
	If we think of the symplectic vector space $\B$ as a Poisson manifold \cite[Definition B.3]{Cr21} and take the vector space $\Z$ to be a Poisson manifold with a vanishing Poisson tensor, then a linear map $\B \to \Z$ is a Poisson map in the sense of \cite[Definition 1.1]{Cr21} if and only if it is a Poisson variable in the sense of \cref{eq:Poisson}.}%
	\begin{lemma}[Characterisation of Poisson variables]\label{lem:Poisson}
		A linear map $Z \colon \B \to \Z$ is a Poisson variable if and only if its kernel $\K \coloneq \ker(Z)$ is a coisotropic subspace of $\B$.
	\end{lemma}
	\begin{proof}
		To show this, we can work in a Darboux basis of $\B$, which then also gives us an inner product on $\B$.
		We can use this choice to compute the symplectic complement $\K^\omega$ of $\K$, which is the set of those $b \in \B$ that satisfy 
		\begin{equation}\label{eq:complement_element}
			0 = b^T \Omega k = (-\Omega b)^T k
		\end{equation}
		for all $k \in \K$.
		This means that every $\Omega b$ is orthogonal to $\K$.
		That is, for each $b \in \K^\omega$, there exists a $w \in \Z$ such that we have
		\begin{equation}
			-\Omega b = Z^T w, \qquad \text{which implies} \qquad b = \Omega Z^T w.
		\end{equation} 
		Since every such $b$ satisfies \cref{eq:complement_element}, we can express $\K^\omega$ as the image of the linear map $\Omega Z^T$.
		
		By definition, the kernel of $Z$ is coisotropic, which means that every $b \in \K^\omega$ satisfies $Zb = 0$.
		Using the above characterisation, this is equivalent to
		\begin{equation}
			Z \Omega Z^T w = 0
		\end{equation}
		for all $w \in \Z$, which is precisely the condition that $Z$ is a Poisson variable.
	\end{proof}
	We can use this result to generalise the definition of Poisson variables to the case of properties that only pertain to a part of the state space.
	This is useful when studying prediction, because the predicted variable as introduced in \cref{eq:predicted_main} is only defined on a subspace of $\B$ and thus is not a variable in the strict sense of \cref{sec:variables}.
	\begin{definition}[Partial variable]\label{def:partial_variable}
		For a given toy object $B$, a \newterm{partial variable} is a function $Z \colon \F \to \Z$, where $\F$ is a subset of the ontic state space $\B$.
		It is a \newterm{Poisson partial variable} if $Z$ is surjective, linear, and its kernel is coisotropic as a subspace of $\B$.
	\end{definition}
	Clearly, for every Poisson partial variable, the domain of definition $\F$ is necessarily also a coisotropic subspace.
	The smallest possible coisotropic subspace is a Lagrangian one. 
	We use this to define a subclass of Poisson variables that retain the most information about the system $\B$ in their value.
	\begin{definition}[Maximal Poisson variable]\label{def:max_Poisson}
		If $Z$ is a surjective linear variable whose kernel is Lagrangian, we say that $Z$ is a maximal Poisson variable.
	\end{definition}
	
	Finally, maps that preserve the symplectic form can be defined as follows.
	\begin{definition}[Symplectic maps]\label{def:symplectic_map}
	Given two symplectic vector spaces $\B$ and $\C$, a linear map $M \colon \B \to \C$ is symplectic if it satisfies
	\begin{equation}\label{eq:symplectic}
		M^T \Omega_\C M = \Omega_\B.
	\end{equation}
	\end{definition}
	Note that symplectic maps are injective and thus they necessitate $\dim(\B) \leq \dim(\C)$.
	This is in contrast to Poisson maps (\cref{def:Poisson_map}), which are surjective (\cref{lem:ops_surj}) and include discarding transformations (\cref{ex:Poisson_map}).
	In particular, the category of symplectic vector spaces with symplectic maps between them does not constitute a suitable physical theory in our sense.
	One of the reasons is that it lacks discarding maps.
	The other is that every inclusion map $\{0\} \to \C$ is symplectic and including such maps as allowed transformations would provide free access to complete ready states (\cref{sec:toy_subjects}).
	
\subsection{Physical Systems and Transformations}\label{sec:NTT}
	
	A \newterm{physical system} $B$ in nomic toy theory is specified by a finite-dimensional symplectic vector space $\B$ with dimension $2n$. 
	
	Given two systems $B$ and $R$, their \newterm{composite} is given by $\B \oplus \R$, whose symplectic form is $\Omega_\B \oplus \Omega_\R$. 
	For any such composite, we have a projection map sending $b \oplus r \mapsto b$ for $b \in \B$ and $r \in \R$, which we denote by
	\begin{equation}\label{eq:projection}
		\pi_{\B} \colon \B \oplus \R \to \B.
	\end{equation}
	These can be interpreted as discarding operations in nomic toy theory.
	
	In order to describe general physical processes in nomic toy theory, we use the following concept from Poisson geometry \cite{Da99,Cr21}.
	\begin{definition}[Poisson maps]
		\label{def:Poisson_map}
		Let $\B$ and $\C$ be symplectic vector spaces.
		A linear map $M \colon \B \to \C$ is \newterm{Poisson}\footnotemark{} if it satisfies
		\footnotetext{\label{foo:Poisson_map}Just as for Poisson variables (\cref{foo:Poisson}), the terminology comes from Poisson geometry.
		If we think of the symplectic vector spaces $\B$ and $\C$ as Poisson manifolds, then a linear map between them is a Poisson map if and only if it satisfies \cref{eq:Poisson_map}.
		On a related note, the differential (at a point) of an arbitrary Poisson map between Poisson manifolds is a linear map of the tangent spaces, which restricts to the respective symplectic leaves \cite[Section 5]{Da99}, on which it further satisfies \cref{eq:Poisson_map}.}%
		\begin{equation}\label{eq:Poisson_map}
			M \Omega_\B M^T = \Omega_\C
		\end{equation}
		where $\Omega_\B$ and $\Omega_\C$ are the matrices of the respective symplectic forms.
	\end{definition}
	\begin{example}\label{ex:Poisson_map}
		There are two fundamental types of Poisson maps.
		The first one is a surjective symplectic transformation, i.e.\ a linear map $M \colon \B \to \B$ satisfying \cref{eq:symplectic}.
		Such an $M$ is necessarily a bijection since the symplectic form is non-degenerate.
		Moreover, its inverse is also symplectic and so it is reversible.
		The second one is a discarding map, i.e.\ a projection as in \eqref{eq:projection}.
		These are clearly not reversible.
		We show in \cref{thm:dilation} that all Poisson maps can be generated by a sequential composition of maps of these two types.
	\end{example}
	Even though we need to choose bases of $\B$ and $\C$ to make sense of \cref{eq:Poisson_map}, it is clear that this does not influence which maps are Poisson.
	This is made explicit also by the basis-independent formulation in \cref{lem:ops_forms}.
	
	A general \newterm{physical transformation} $m \colon \B \to \C$ in nomic toy theory is given by an affine map 
	\begin{equation}
		m(b) = M b + c
	\end{equation}
	where $c$ is any element of $\C$ and $M \colon \B \to \C$ is a Poisson map.
	To study learning in nomic toy theory, we can set $c = 0$ without loss of generality and work with linear maps (see \cref{rem:zero_affine_shift}).
	
	As we mention briefly in \cref{sec:toy_systems}, reversible transformations in nomic toy theory can be interpreted as describing the time evolution of classical particles with respect to an at most \emph{quadratic} Hamiltonian.
	Alternatively, we may also interpret the ontic state space as a tangent space to a symplectic manifold, thus representing infinitesimal fluctuations of a general phase space in classical mechanics. 
	In this case, dynamics generated by an \emph{arbitrary} Hamiltonian leads to a symplectic linear map between the tangent spaces, namely the pushforward of the Hamiltonian flow.
	General physical transformations, which also allow for discarding of subsystems, can be then interpreted as modelling linearised open-system dynamics of classical systems.

	Given two transformations $m \colon \B \to \C$ and $n \colon \S \to \T$, their parallel composite is the direct sum $m \oplus n \colon \B \oplus \S \to \C \oplus \T$, which acts as expected via
	\begin{equation}
		(m \oplus n)(b \oplus s) = \bigl( m(b) \bigr) \oplus \bigl( n(s) \bigr).
	\end{equation}

	It is clear from the definition that physical transformations are closed under both sequential and parallel composition, and that they include the discarding maps as introduced in \eqref{eq:projection}.
	
	The following result establishes that the dimension of the output of a Poisson map can be no larger than the dimension of its input.
	\begin{lemma}\label{lem:ops_surj}
		Every Poisson map is surjective.
	\end{lemma}
	\begin{proof}
		Let $2n$ be the dimension of $\C$ for a Poisson map $M \colon \B \to \C$.
		Then we have
		\begin{equation}
			2n = \rank(\Omega_{\C}) = \rank(M \Omega_\B M^T) \leq \rank(M).
		\end{equation} 
		Since the rank of $M$ cannot be larger than the dimension of its output, we must have ${\rank(M)=2n}$ and thus $M$ is surjective.
	\end{proof}
	Therefore, whenever $\B$ and $\C$ are isomorphic, a Poisson map must be a linear isomorphism. 
	By the standard argument which establishes that an invertible matrix is symplectic if and only if its transpose is symplectic, we get that in this case Poisson and symplectic maps coincide.
	For inequivalent input and output spaces, Poisson maps further include the possibility of information erasure, an idea that we formalise in \cref{thm:dilation} below.
	
	First, however, let us provide an alternative equivalent characterisation of Poisson maps.
	
	\begin{proposition}[Characterisation of Poisson maps]\label{lem:ops_forms}
		For a surjective linear map ${N \colon \R \to \S}$, the following are equivalent:
		\begin{enumerate}
			\item\label{it:ops} $N$ is Poisson.
			
			\item\label{it:ops_kernel} The kernel of $N$, denoted by $\K$, is a symplectic subspace of $\R$ and the restriction of $N$ to its symplectic complement, $\K^\omega$, is a symplectic isomorphism.
			
			\item\label{it:ops_forms} For all $a_1, a_2$ in $\K^\omega$, we have
			\begin{equation}\label{eq:ops_forms}
				\omega_\S\bigl(N(a_1), N(a_2) \bigr) = \omega_\R(a_1, a_2).
			\end{equation}
			Here, $\omega_\B$ and $\omega_\R$ are the symplectic forms of the two respective vector spaces.
		\end{enumerate}
	\end{proposition}
	\begin{proof}
		\textbf{\ref{it:ops} $\bm{\Rightarrow}$ \ref{it:ops_kernel}:}
		Recall that a subspace, such as $\K$, is a symplectic subspace if the intersection with its symplectic complement
		\begin{equation}
			\K^{\omega} \coloneq \Set{ a \in \R  \given  \omega_\R(a, k) = 0 \text{ for all } k \in \K }
		\end{equation}
		only contains the zero vector.
		
		To prove this, consider an element $a \in \K \cap \K^{\omega}$.
		It satisfies
		\begin{equation}
			Na = 0  \qquad \text{ and } \qquad  a^T \Omega_\R k = 0 \text{ for all } k \in \K.
		\end{equation}
		The latter says that $a^T \Omega_\R = (\Omega_\R^T a)^T$ is a functional that vanishes on the kernel of $N$.
		By the Fundamental Theorem of Linear Algebra, $a^T \Omega_\R$ must be in the ``row space'' of $N$, i.e.\ there exists a $s \in \S$ such that we have
		\begin{equation}\label{eq:row_space}
			a^T \Omega_\R = s^T N.
		\end{equation}
		Using the Darboux basis for $\R$, we have $\Omega_\R^T = \Omega_\R^{-1}$, and thus we get $a =  \Omega_\R N^T s$ from the transpose of \cref{eq:row_space}.
		Then, by $Na=0$, we obtain the first equality in 
		\begin{equation}
			0 = N \Omega_\R N^T s = \Omega_{\S} s,
		\end{equation}
		while the second uses the fact that $N$ is Poisson. 
		Since $\Omega_{\S}$ is invertible, we obtain $s = 0$ and conclude $a=0$.
		This shows that $\K$ is indeed a symplectic subspace.
		
		Let us denote the restriction of $N$ to  $\K^\omega$ by $\tilde{N} \colon \K^\omega \to \S$.
		$\tilde{N}$ is clearly a vector space isomorphism by \cref{lem:ops_surj} and $\K \cap \K^{\omega} = \{0\}$ proven above.
		Note that we have the splitting ${\R = \K \oplus \K^\omega}$, with respect to which $\Omega_{\R}$ is block diagonal.
		In particular, the map $\tilde{N}$ is Poisson and thus also symplectic, since it is invertible.
		
		\textbf{\ref{it:ops_kernel} $\bm{\Rightarrow}$ \ref{it:ops}:}
		Using the notation and splitting of $\R$ from the previous paragraph, we can write the left hand side of \cref{eq:Poisson_map} as
		\begin{equation}\label{eq:ops_split}
			\begin{pmatrix}
				0 & \tilde{N}
			\end{pmatrix}
			\begin{pmatrix}
				\Omega_{\K} & 0 \\
				0 & \Omega
			\end{pmatrix}
			\begin{pmatrix}
				0 \\ \tilde{N}^T
			\end{pmatrix} = \tilde{N} \Omega \tilde{N}^T,
		\end{equation}
		where $\Omega$ is a short-hand notation for $\Omega_{\K^{\omega}}$.
		The right-hand side of \cref{eq:ops_split} equals $\Omega_{\S}$ since $\tilde{N}$, being a symplectic isomorphism, is also Poisson.
		Thus, $N$ itself is Poisson.
		
		\textbf{\ref{it:ops_kernel} $\bm{\Rightarrow}$ \ref{it:ops_forms}:}	
		We denote the inverse of $\tilde{N}$ by $P \colon \S \to \K^\omega$, which is also a symplectic map, since $\tilde{N}$ a symplectic isomorphism.
		That is, we have
		\begin{equation}
			P^T \Omega P = \Omega_{\S}.
		\end{equation}
		Thus, for any $s_1, s_2 \in \S$, we can define $a_i \coloneq P s_i$, so that we have
		\begin{align}
			\omega_\S\bigl(s_1, s_2 \bigr) &= s_1^T \Omega_{\S} s_2 \\
			&= s_1^T P^T \Omega P s_2 \\
			&= a_1^T \Omega_{\R} a_2 = \omega_\R(a_1, a_2)
		\end{align}
		Since $N$ is a left inverse of $P$, we obtain \cref{eq:ops_forms}.
		
		\textbf{\ref{it:ops_forms} $\bm{\Rightarrow}$ \ref{it:ops_kernel}:}	
		Consider a $k \in \K \cap \K^{\omega}$.
		By \cref{eq:ops_forms} and $Na_1 = 0$, we have
		\begin{equation}
			\omega_\R (k, a) = \omega_{\S}(0,Na) = 0  \qquad \forall \, a \in \K^{\omega}.
		\end{equation} 
		Moreover, since $k$ is in $\K^{\omega}$, it also satisfies $\omega_\R (k, a) = 0$ for all $a \in \K$.
		Since the symplectic form is non-degenerate, $\K$ and $\K^{\omega}$ together span the whole vector space, and so we get $k = 0$ by the non-degeneracy of $\omega_\R$ again.
		Thus, the kernel of $N$ is indeed a symplectic subspace.
		
		Consequently, the restriction of $N$ to $\K^{\omega}$ has a trivial kernel and is thus a vector space isomorphism.
		By \cref{eq:ops_forms}, it is also a symplectic map.
	\end{proof}
	
	This helps us show that every Poisson map can be expressed as a symplectic isomorphism followed by a discarding map. 
	
	\begin{theorem}[Dilations of Poisson maps]\label{thm:dilation}
		For every Poisson map $N \colon \R \to \S$, there exists a symplectic vector space $\T$ and a symplectic isomorphism $M \colon \R \to \S \oplus \T$ such that 
		\begin{equation}
			N = \pi_\S \comp M
		\end{equation}
		holds, where $\pi_\S$ is the projection map $\S \oplus \T \to \S$.\footnote{This result shows that Poisson maps are basically linear analogues of symplectic reductions by the action of a Lie group on a symplectic manifold.}
	\end{theorem}
	\begin{proof}
		Note that by \cref{lem:ops_forms}, we have $\S \cong \ker(N)^\omega$ and $\R \cong \ker(N)^\omega \oplus \ker(N)$.
		We let $\T \coloneq \ker(N)$, so that we have 
		\begin{equation}
			\R \cong \S \oplus \T.
		\end{equation}
		With respect to this decomposition, we then define $M$, in block matrix notation, to be
		\begin{equation}
			M = \begin{pmatrix}
				\tilde{N} & 0 \\
				0 & \id_{\T}
			\end{pmatrix}.
		\end{equation}
		By \cref{lem:ops_forms}, it is bijective and also symplectic:
		\begin{equation}
			\begin{split}
				M^T \Omega M &= 
				\begin{pmatrix}
					\tilde{N}^T & 0 \\
					0 & \id 
				\end{pmatrix}
				\begin{pmatrix}
					\Omega_{\S} & 0 \\
					0 & \Omega_{\T}
				\end{pmatrix}
				\begin{pmatrix}
					\tilde{N} & 0 \\
					0 & \id_{\T}
				\end{pmatrix} \\
				&= 
				\begin{pmatrix}
					\Omega_{\S} & 0 \\
					0 & \Omega_{\T}
				\end{pmatrix}
			\end{split}
		\end{equation}
		Moreover, projecting onto $\S$ gives back $N$:
		\begin{equation}
			\pi_\S \comp M = 
			\begin{pmatrix}
				\id_{\S} & 0
			\end{pmatrix}
			\begin{pmatrix}
				0 & \tilde{N} \\
				\id_{\T} & 0
			\end{pmatrix} = N
		\end{equation}
		and thus the proof is complete.
	\end{proof} 

	In essence, \cref{thm:dilation} provides a dilation of a Poisson map by expressing it as a composition of a reversible Poisson map followed by discarding.
	It is thus similar to Stinespring dilation theorem, which says that every completely positive map can be expressed as an isometry followed by discarding (see \cite{Pa02} for instance).

	The following lemma establishes that any surjective linear map $\Q \to \Z$, where both the input and output are Lagrangian subspaces of symplectic vector spaces, can be extended to a Poisson map.
	There is a variant of this construction for non-surjective linear maps, in which case the off-diagonal blocks in \eqref{eq:Lag_rotation} cannot vanish.
	We do not need this extra generality and thus restrict to the simpler case.

	\begin{lemma}[Transforming a Lagrangian subspace]\label{lem:Lag_rotation}
		Let $\R = \Q \oplus \P$ and $\S = \Z \oplus \W$ be symplectic vector spaces whose symplectic form has the canonical matrix representation of \eqref{eq:sympl_form} with respect to these decompositions.
		For any surjective linear map $Z \colon \Q \to \Z$, there is a Poisson map $N \colon \R \to \S$ satisfying  
		\begin{equation}
			N(q) = Z(q)
		\end{equation}
		for all $q \in \Q$.
	\end{lemma}
	\begin{proof}
		First, we use the singular value decomposition\footnote{Here, we can choose arbitrary inner products on $\Q$ and $\Z$, and then use the same inner products for $\P$ and $\W$ respectively.} to write 
		\begin{equation}
			Z = U \Lambda V^T,
		\end{equation}
		where $U$ and $V$ are two orthogonal matrices and $\Lambda$ is a (rectangular) diagonal matrix with its diagonal entries denoted by $\lambda_i$.
		Note that since $Z$ is surjective, all the diagonal entries are non-zero.
		
		We also define a linear map $\Gamma \colon \P\to \W$, which is given by a (rectangular) diagonal matrix with its diagonal entries equal to $1/\lambda_i$. 
		Note that the transpose of $\Gamma$ is a right-inverse of $\Lambda$, i.e.\ we have $\Lambda \Gamma^T = \id$.
		
		We then construct the matrix $N$ as (we suppress the  isomorphisms $\Q \cong \P$ and $\W \cong \Z$):
		\begin{equation}\label{eq:Lag_rotation}
			N \coloneq 
			\begin{pmatrix}
				U \Lambda V^T & 0 \\
				0 & U \Gamma V^T
			\end{pmatrix} 
		\end{equation}
		
		Finally, let us show that $N$ is Poisson by direct computation:
		\begin{equation}
			\begin{split}
				N \Omega_{\R} N^T &= 
				\begin{pmatrix}
					U \Lambda V^T & 0 \\
					0 & U \Gamma V^T
				\end{pmatrix} 
				\begin{pmatrix}
					0 & \id \\
					-\id & 0
				\end{pmatrix}
				\begin{pmatrix}
					V \Lambda^T U^T & 0 \\
					0 & V \Gamma^T U^T
				\end{pmatrix}  \\
				&= 
				\begin{pmatrix}
					0 & U \Lambda V^T V \Gamma^T U^T \\
					- U \Gamma V^T V \Lambda^T U^T & 0
				\end{pmatrix} = \Omega_{\S}
			\end{split}
		\end{equation}
		where we have used the orthogonality of $U$ and $V$ as well as $\Lambda \Gamma^T = \id$.
	\end{proof}
	
\subsection{Toy Subjects}\label{sec:subjects}

	As explained in \cref{sec:toy_subjects}, a subject in nomic toy theory is a physical system equipped with a specification of self-knowledge.
	We denote the subject's ontic state space by $\S = \Q \oplus \P$, with symplectic form $\Omega_\S$ given by \eqref{eq:sympl_form} with respect to this decomposition. 
	In principle, one could consider any variable as specifying the information directly accessible to the subject.
	However, since we do not have anything specific to say about this general case, we restrict our attention to \newterm{manifest variables} that are given by a linear map 
	\begin{equation}
		E \colon \S \to \E.
	\end{equation}
	Its kernel can be interpreted as the degrees of freedom that are inaccessible to the subject{\,---\,}they are part of the physical system, but they do not directly encode subject's knowledge.
	On the other hand, the quotient space $\linefaktor{\S}{\ker(E)}$ corresponds to the supervenience domain of the toy subject (cf.\ \cref{fig:domains_for_epistemic_horizons}).
	It carries the physical record of the (empirical) information the subject has about the world and thus makes it into a memory system \cite{Wo92}.

	As depicted in \cref{fig:interaction}, the manifest variable is relevant after a measurement interaction, where it encodes the information learned by the subject.
	Additionally, we can express prior information by specifying a constraint on the initial state of a ready system $R$, which enters the measurement interaction as an input.
	
	Such constraint is given by a \newterm{ready state variability}, which is a subspace $\V$ of the ontic state space $\R$ of the ready system.
	It refers to the uncontrollable degrees of freedom of $R$.
	That is, the ready system can occupy any of the ontic states $v \in \V$ before the measurement interaction takes place.
	During a measurement interaction, these can disturb the measured system and spoil the correlation with the measurement outcome.
	On the other hand, the quotient space $\linefaktor{\R}{\V}$ corresponds to properties that are fixed among the possible initial ontic states.
	
	In general, one could consider affine subspaces of $\R$ (as opposed to linear ones) as ready states, but a constant shift in the initial state changes nothing about the subject's learning capabilities (see \cref{rem:zero_affine_shift}) and thus we omit these from our analysis.
	
	As described in \cref{sec:toy_subjects}, we focus on three types of manifest variables and three types of ready states: trivial, Poisson, and complete.
	
	\begin{remark}[Issues with the old interpretation of ready states]\label{rem:ready_state_issues}
		Note that in this paper we interpret ready states differently than in \cite{Fa24}.
		There, we justified the assumption of fixing the initial value of the subject's manifest variable by an argument roughly saying that the subject, by virtue of having direct access to its own manifest variable, can choose to perform the learning procedure conditional on this value.
		If such conditional operations were valid physical processes in nomic toy theory, one could indeed implement a protocol which goes through only for a specific initial value of the manifest variable, thus effectively rendering its initial value to be fixed.
		
		There are two issues with this argument.
		The first, and somewhat less important, is that pre-selecting the initial value of the manifest variable need not succeed. 
		This means that the whole learning protocol only goes through in a fraction of instances, with probability approaching $0$ as the space of possible values of the manifest variable increases in dimension.
		
		The second, and more severe, issue is the following. 
		To be consistent with our intention to model the learning process explicitly by physical processes,\footnotemark{} we ought to also model the conditioning on the initial value of the manifest variable by a physical process.
		\footnotetext{For example, our choice to model subjects as physical systems and their knowledge as a property of their physical state is one consequence.}%
		A careful analysis, however, reveals that conditioning the nature of the physical interaction on the initial value of (a Poisson variable of) one of the systems partaking in the interaction can result in a map that is not symplectic, and thus not physical in nomic toy theory.
		A similar conclusion was found in \cite{Ha23b}.
		
		The `correct' way to model the above protocol would be to assume no ready state (i.e.\ that the initial ontic state of the subject may be arbitrary) and consider all valid physical interactions (which would include any conditioning operation as mentioned above, as long as it gives rise to a valid physical process on arbitrary input).
		As we see in \cref{tab:retrodiction}, the resulting epistemic horizon in the case of retrodictive learning is not changed when the ready state is changed from Poisson to trivial.
		However, it becomes a crucial distinction in the case of repeatable learning (\cref{tab:repeatable}).
	\end{remark}
	
	\begin{remark}[New interpretation of ready states]\label{rem:ready_state_resolution}
		If we do not use the previous justification for the constraint on initial ontic states, the question then arises how should one think of ready states here.
		We leave this question to a large part unanswered.
		
		Indeed, no physical process in nomic toy theory can turn an unconstrained initial state into a constrained one (cf.\ \cref{thm:prediction_Poisson_MV_trivial_RS}).
		Therefore, these resources of initial information have to be justified from outside of nomic toy theory itself.
		
		There are numerous possibilities.
		For example, according to a ``past hypothesis'' \cite{Re57,Al00}, the universe itself starts with a low entropy initial state, which then provides a resource of information to be harvested by subjects in such a universe. 
		A specific version of this possibility, relevant for nomic toy theory (which we think of as a deterministic physical theory), is that there is a fact of the matter about what the precise initial ontic state is.
		In this case the variability in the initial state is merely a way to talk about the fact that one does not know what this state is.
		
		The ready state assumption is then precisely that{\,---\,}an assumption or a working hypothesis, which cannot be proven by deduction, but for which one may obtain supportive evidence by induction.
		A reliable source of low energy states (such as sunlight) can be then used by subjects (such as living organisms) to learn about other systems in their environment.
		Any knowledge obtained like this is not absolute; it is conditional on the ready system $R$ truly being in a state within its assumed variability $\V$.
		As the first row of \cref{tab:repeatable} suggests, giving up on absolute knowledge may just be a necessary price to pay in order to have any reliable knowledge at all.
		
		Nevertheless, the presence of epistemic horizons we derive is independent of how the ready state assumption is justified.
	\end{remark}
	
\subsection{Learning}\label{sec:learning}

As explained in \cref{sec:learning_types}, learning in nomic toy theory is facilitated by a physical interaction.
In our previous work we called such interactions measurements \cite[Definition 2.6]{Fa24}.
Since we consider more general subjects in the present manuscript (\cref{rem:comparison}), we also need to generalise our definition of measurements to accommodate such subjects.

\begin{definition}[Measurement]\label{def:measurement}
	A \newterm{measurement} in nomic toy theory is a Poisson map $M \colon \B \oplus \R \to \B \oplus \S$ together with a manifest variable $E \colon \S \to \E$ and a ready state variability $\V \subseteq \R$.
\end{definition}
That is, $\S$ is a toy subject and $\R$ is a ready system (see \cref{sec:subjects}).

In the following proofs, we often use a specific notation for the matrix blocks of $M$.
For example, with respect to the decomposition of its input via $\B \oplus \R$ and output via $\B \oplus \S$, we split the matrix $M$ into submatrices denoted as
\begin{equation}\label{eq:blocks}
	M = \begin{pmatrix}
			M_{\B\B} & M_{\B \R} \\
			M_{\S\B} & M_{\S \R}
		\end{pmatrix}
\end{equation}
where the first and second indices label the output and input of the submatrix respectively.
For example, $M_{\S\B}$ is a linear map $\B \to \S$, which expresses the effect of the initial ontic state of the object on the final ontic state of the subject.
Furthermore, if we restrict these to a subspace, such as the ready state variability $\V \subseteq \R$, we denote the restricted version of $M_{\B \R}$, for instance, by $M_{\B \V}$.

\begin{remark}[Omitting the affine shifts]\label{rem:zero_affine_shift}
	Note that the interaction $M$ in the definition of a measurement is not the most general physical transformation in nomic toy theory, because it is a linear map with vanishing constant shift.
	
	This is a mere matter of convenience to simplify the terminology.
	Indeed, one could consider measurements given by arbitrary physical transformations in nomic toy theory and obtain equivalent results on epistemic horizons to the ones we derive here.
	The constant shift does not enable any additional learning, as it merely moves the learned information around.
	For example, one effect that including the affine shifts would have on our results is that in \cref{thm:prediction_complete_MV_Poisson_RS} the preparation support could be not just an arbitrary coisotropic subspace $\F$, but an arbitrary affine subspace given by $\F + b$ for a coisotropic subspace $\F$ and an ontic state $b \in \B$.
\end{remark}

As explained in \cref{sec:learning_types}, a given interaction can facilitate different kinds of learning. 
\newterm{Retrodiction} refers to learning about the state of the object prior to the interaction (i.e.\ learning about its past). 
The correlation of interest is one between the \emph{initial} state of the toy object $B$ and the final value of the manifest variable of the toy subject $S$.
\newterm{Prediction} refers to learning about the state of the object after the interaction (i.e.\ learning about its future). 
The relevant correlation is between the \emph{final} state of the toy object $B$ and the final value of the manifest variable of the toy subject $S$.
The last option we consider is learning by \newterm{repeatable measurements}, which in some sense includes retrodiction and prediction simultaneously.
The reason is that for repeatable measurements, the information learned about the past must also be present in the final state of the toy object and is thus also predicted.
We discuss repeatable measurements further in \cref{sec:bitemporal}. 

For both retrodiction and prediction, we can distinguish between properties of the object which can be inferred from the value of subject's manifest variable, and those which cannot. 
We call the former fixed variables, because they are defined by the criterion that, given the value of $E$, their value is unique (i.e.\ fixed).
In the following definition, we use the two maps $\beta \colon \B \oplus \V \to \B$ and $\epsilon \colon \B \oplus \V  \to \E$ defined via \eqref{eq:beta} and \eqref{eq:epsilon} respectively.
Moreover, $\pi_\B \colon \B \oplus \V \to \B$ is the projection map $b \oplus v \mapsto b$, as usual.

\begin{definition}[Fixed variables]\label{def:fix_var}
	Consider a  measurement $M \colon \B \oplus \R \to \B\oplus \S$ with manifest variable $E \colon \S \to \E$ and a ready state $\V$. 
	A variable $Z \colon \B \to \Z$ of $B$ is said to be \newterm{retrodictively fixed} if there exists a function $f \colon \E \to \Z$ satisfying
	\begin{equation}\label{eq:meas_var}
		Z \comp \pi_\B = f \comp \epsilon.
	\end{equation}
	$Z$ is said to be \newterm{predictively fixed} if there exists a function $g \colon \E \to \Z$ satisfying
	\begin{equation}\label{eq:var_fix_prep}
		Z \comp \beta = g \comp \epsilon.
	\end{equation}
\end{definition}
The left-hand side of \cref{eq:meas_var} gives the value of  $Z$ before the measurement took place, while $\epsilon$ on the right-hand side specifies the value of the subject's manifest variable after the measurement. 
This condition thus expresses that the \emph{initial value} of $Z$ can be inferred from the value of the subject's manifest variable.

\Cref{eq:var_fix_prep} specifies a condition on the values of $Z$ in the image of $\beta$, 
namely that the \emph{final value} of $Z$ (the left-hand side) can be inferred from the value of the manifest variable.
Note that no requirement is placed on the values of $Z$ outside of the preparation support $\im(\beta) \subseteq \B$, which consists of those ontic states that cannot occur as final states after the measurement.

Fixed variables correspond to properties of the toy object $B$, which the toy subject can in principle learn by performing the measurement $M$ while having access to the ready system $R$ with variability $\V$.

The next result formalises the idea that if, among two toy subjects, one has either a more informative manifest variable or access to more informative ready states, then it can learn at least as much as the other subject.

\begin{proposition}[Relating learnable properties for distinct subjects]\label{lem:subject_retrodictive_power}
	Consider a measurement interaction $M$ and two subjects, $S$ and $S'$, with identical ontic state space given by $\S$.
	Assume that their manifest variables satisfy 
	\begin{equation}\label{eq:manifest_inclusion}
		\begin{tikzcd}
			\S \arrow[d, "E'"'] \arrow[rd, "E", bend left] &  \\
			\E' \arrow[r, "c"]                              & \E              
		\end{tikzcd}
	\end{equation}
	for some function $c$, and their ready states satisfy $\V \supseteq \V'$.
	Then every variable retrodictively fixed for subject $S$ is also retrodictively fixed for subject $S'$.
	The same holds for predictively fixed variables.
\end{proposition}

\begin{proof}
	Consider a variable $Z \colon \B\to \Z$ retrodictively fixed for subject $S$, i.e.\ such that the diagram
	\begin{equation}\label{eq:retr_fixed_cd}
		\begin{tikzcd}
			\B \oplus \V \arrow[r, "\pi_\B"] \arrow[d, "\epsilon"'] & \B \arrow[d, "Z"] \\
			\E \arrow[r, "f"]                              & \Z               
		\end{tikzcd}
	\end{equation}
	commutes.
	Denote the inclusion map $\V' \hookrightarrow \V$ by $\iota$. 
	Precomposing \eqref{eq:retr_fixed_cd} by $\iota$, we obtain the commutative diagram
	\begin{equation}
		\begin{tikzcd}
			\B \oplus \V' \arrow[rd, "\id \oplus \iota"] \arrow[rrd, "\pi_\B", bend left] \arrow[d, "\epsilon'"] &                                                &                   \\
			\E' \arrow[rd, "c"']                                                                       & \B \oplus \V \arrow[r, "\pi_\B"] \arrow[d, "\epsilon"'] & \B \arrow[d, "Z"] \\
			& \E \arrow[r, "f"]                              & \Z               
		\end{tikzcd}
	\end{equation}
	The trapezoid on the left commutes by \eqref{eq:manifest_inclusion}. 
	From the outermost paths we infer 
		\begin{equation}
		Z \comp \pi_\B = f \comp c \comp \epsilon',
	\end{equation}
	namely $Z$ is retrodictively fixed for subject $S'$. 
	
	The proof for predictively fixed variables is analogous upon replacing the projection ${\pi_\B \colon \B \oplus \V \to \B}$ by $\beta$, and replacing the projection $\pi_\B \colon \B \oplus \V' \to \B$ by $\beta'$. 
\end{proof}

In the following two subsections, we characterise the sets of retrodictively and predictively fixed variables.  
While \cref{sec:measured} generalises \cite[Section 2.5]{Fa24}, the material in \cref{sec:predicted} is new since it pertains to prediction.

\subsection{Retrodicted variables}\label{sec:measured}

If we restrict \cref{eq:meas_var} to the kernel of $\epsilon$, we find $Z \comp \pi_\B|_{\ker(\epsilon)} = f(0)$.
In other words, every retrodictively fixed variable is  constant on the image of this kernel under the projection map $\pi_\B \colon \B \oplus \V \to \B$, which is a subspace of $\B$ given by those states $b \in \B$ for which there exists a $v \in \V$ such that $\epsilon(b \oplus v) = 0$.

\begin{definition}[Retrodictively inaccessible subspace]\label{def:retr_inaccessible}
	Given a measurement, the subspace $\pi_\B\bigl( \ker (\epsilon) \bigr)$ is termed the \newterm{retrodictively inaccessible subspace} and denoted by $\I_\ret$.
\end{definition}

The retrodictively inaccessible subspace is made up of those degrees of freedom of the initial state of the toy object which cannot influence the value of the subject's manifest variable.  
It carries the information that cannot be retrodicted based on the value of the manifest variable. 
This intuition is made precise in \cref{prop:fixed_from_measured} below.

\begin{definition}[Retrodicted variable]
	\label{def:measured}
	Given a measurement, we define the \newterm{retrodicted variable} to be the canonical quotient map
	\begin{equation}
		L_\ret \colon \B \to \newfaktor{\B}{\I_\ret}
	\end{equation}
	which maps each ontic state $b \in \B$ to its equivalence class $b + \I_\ret$.
\end{definition}
Note that every retrodicted variable is linear and surjective.
Moreover, it depends on both the manifest variable as well as the ready state. 
Indeed, one of the key goals of this article is to study this dependence (see \cref{tab:retrodiction,tab:repeatable}) in order to understand the epistemic limitations faced by different kinds of subjects.

In \cref{rem:comparison} we mention that the notion of measured variables from our earlier work \cite[Definition 2.9]{Fa24} is in fact the same as that of retrodicted variables.
Let us connect these two now.
In \cite[Definition 2.8]{Fa24}, the free manifest subspace $\F$ is defined as the orthogonal complement of the subspace $\im(E \comp M_{\S \V}) \subseteq \E$, which depends on a choice of an inner product. 
However, as we show next, choosing $\F$ to be any complement yields a variable that is equivalent to the retrodicted variable from \cref{def:measured}.

\begin{proposition}[Retrodicted is measured]\label{prop:concrete_measurable}
	Consider a measurement $M$ and let $\C$ be the subspace $\im(E \comp M_{\S \V})$.
	Choose any subspace $\F \subseteq \E$ satisfying ${\E = \F \oplus \C}$ and $\F \cap \C = \{0\}$.\footnote{If $\E$ is an inner product space, we can choose $\F$ to be the orthogonal complement of $\C$.}
	Then the kernel of\/\footnotemark{}
	\footnotetext{Recall that $M_{\S \B}$ is a notation from \cref{eq:blocks}.}%
	\begin{equation}
		\pi_\F \comp E \comp M_{\S \B} \colon \B \to \F
	\end{equation}
	is the retrodictively inaccessible subspace $\I_\ret$.
\end{proposition}
That is, if we denoted $\pi_\F \comp E \comp M_{\S \B}$ by $M_{\F\B}$, then restricting this map to its image gives a linear map $\overline{M_{\F\B}} \colon \B \to \im(M_{\F\B})$ that satisfies 
\begin{equation}
	L_\ret = \varphi \comp \overline{M_{\F\B}}
\end{equation}
for a vector space isomorphism $\varphi \colon \im(M_{\F\B}) \to \newfaktor{\B}{\I_\ret}$. 
In other words, up to a vector space isomorphism, we can compute the retrodicted variable $L_{\rm ret}$ without finding the retrodictively inaccessible subspace $\I_{\rm ret}$.
All it takes is to compute $M_{\F\B}$ for \emph{any} chosen free manifest subspace $\F$ and restrict it to its image.
\begin{proof}
	To show this, first note that by definition of $\C$, we have
	\begin{equation}
		\pi_\F (s) = 0  \quad \iff \quad  s \in \im(E \comp M_{\S \V})  \quad \iff \quad  \exists v \in \V \; : \; s = \epsilon (v)
	\end{equation}
	for every $s \in \E$.
	Furthermore, by the definition of $\epsilon$, we also have
	\begin{equation}
		\begin{split}
			b \in \ker(\pi_\F \comp E \comp M_{\S \B})  \quad &\iff \quad  \pi_\F \comp \epsilon (b) = 0 \\
				&\iff \quad  \exists v \in \V \; : \; \epsilon(b) = \epsilon(v) \\
				& \iff \quad  b \in \pi_{\B}(\ker \epsilon),
		\end{split}
	\end{equation}
	which completes the proof.
\end{proof}

We now generalise \cite[Proposition 2.10]{Fa24}, which identifies the retrodicted variable as the most informative retrodictively fixed variable.
It thus characterises the set of retrodictively fixed variables for a given measurement as the set of all functions that can be obtained from $L_\ret$ by post-processings.

\begin{proposition}[Characterisation of retrodictively fixed variables]\label{prop:fixed_from_measured}
	A variable $Z$ is retrodictively fixed if and only if there is a function ${f \colon \linefaktor{\B}{\I_\ret} \to \Z}$ such that we have
	\begin{equation}\label{eq:fixed_from_measured_1}
		Z = f \comp L_\ret .
	\end{equation}
\end{proposition}
\begin{proof}
	Let us first show that any variable of the form $f \comp L_\ret$ is retrodictively fixed.
It suffices to consider $f$ to be the identity, so that we have to show $L_\ret \comp \pi_\B = h \comp \epsilon$ for some function $h \colon \E \to \linefaktor{\B}{\I_\ret}$.
	Let $v_1, v_2 \in \B \oplus \V$ be two arbitrary initial states that yield identical observations, namely $\epsilon(v_1) = \epsilon(v_2)$. 
	Since $L_\ret$ is linear, we have 
	\begin{equation}
		L_\ret \bigl( \pi_\B(v_1)\bigr) - L_\ret \bigl( \pi_\B(v_2)\bigr) = L_\ret \bigl( \pi_\B(v_1 - v_2)\bigr) = 0,
	\end{equation}
	as $\pi_\B(v_1 - v_2)$ is an element of $\I_\ret$ and so equivalent to $0$ under the inaccessibility equivalence relation.
	This means that $L_\ret \comp \pi_\B$ is constant on the fibers of $\epsilon$. 
	Thus, there is a function $h \colon \E \to \linefaktor{\B}{\I_\ret}$ satisfying the required $L_\ret \comp \pi_\B = h \comp \epsilon$. 
	
	Conversely, assume that a variable $Z \colon \B \to \Z$ satisfies $Z \comp \pi_\B = h \comp \epsilon$. 
	To show that $Z = f \comp L_\ret$ for some function $f$, we must establish that $Z$ is constant on the inaccessibility equivalence classes.
	To this end,  consider two ontic states $b_1, b_2 \in \B$ satisfying $b_1 - b_2 \in \I_\ret$.
	This means that there is a $x \in \ker(\epsilon)$ such that $b_1 - b_2 = \pi_\B(x)$. 
	The fact that $Z$ is retrodictively fixed implies, for an arbitrary $v \in \V$, that 
	\begin{align}
		Z (b_2) &= Z \comp \pi_\B(b_1 \oplus v - x)  \\
		&= h \comp \epsilon(b_1 \oplus v - x) = h \comp \epsilon(b_1 \oplus v) \\
		&= Z \comp \pi_\B (b_1 \oplus v) = Z(b_1) .
	\end{align}
	Together, we have shown that $b_1 - b_2 \in \I_\ret$ implies $Z(b_1) = Z(b_2)$.
By the universal property of quotients, there exists a well-defined function $f \colon \linefaktor{\B}{\I_\ret} \to \Z$ given by $f([b]_\ret) = Z(b)$ such that $Z = f \comp L_\ret$ holds.
\end{proof}

\subsection{Predicted variables}\label{sec:predicted}

Let us now turn to predictively fixed variables and their characterisation.
Notice that \cref{eq:pred_var_fix} provides a non-trivial constraint on the variable $Z$ only on a subset of its domain, namely the image of $\beta$. 
While the initial ontic state of $B$ is unconstrained, the set of possible states after the measurement is $\im(\beta)$. 
This fact is independent of the final value of subject's manifest variable, and is accessible to anyone who knows 
\begin{itemize}
	\item that the initial state of the ready system $R$ is within $\V$, and
	\item the details of the measurement interaction. 
\end{itemize}
This is why we refer to it as \emph{non-empirical} knowledge (see \cref{sec:learning_types}), as opposed to the \emph{empirical} knowledge the subject gains by having access to its own manifest variable.

\begin{definition}[Preparation support]\label{def:support}
	Given a measurement, the image of $\beta$ defined in \eqref{eq:beta} is called the \newterm{preparation support}.
\end{definition}

Restricting \cref{eq:var_fix_prep} to the kernel of $\epsilon$, we find $Z \comp \beta|_{\ker(\epsilon)} = g(0)$.
That is, every predictively fixed variable is constant on the image of the kernel of $\epsilon$ under $\beta$. 
Note that this image is a subspace of the preparation support.

\begin{definition}[Predictively inaccessible subspace]\label{def:pre_inaccessible}
	For a given measurement, the subspace $\beta(\ker \epsilon) \subseteq \B$ is termed the \newterm{predictively inaccessible subspace} and denoted by $\I_{\rm pre}$.
\end{definition}

The predictively inaccessible subspace is made up of those final ontic states of the toy object which originate from joint initial states leaving no trace within subject's manifest variable.  
We can thus interpret it as the carrying information that cannot be predicted based on the value of the manifest variable alone. 
This intuition is made precise in \cref{prop:fix_predicted}. 
	
\begin{definition}[Predicted variable]\label{def:predicted}
	Given a measurement $M$, we define the variable \newterm{predicted by} $M$ to be the canonical quotient map
	\begin{equation}\label{eq:predicted}
		L_\pre \colon \im(\beta) \to \newfaktor{\im(\beta)}{\I_\pre},
	\end{equation}
	which maps each ontic state $b \in \B$ to its equivalence class $b + \I_\pre$.
\end{definition}
\begin{remark}\label{rem:predicted_lin}
	We can think of the predicted variable as a pair of a subspace of $\B$ (the preparation support) and a (surjective) linear variable $L_\pre$ defined on the preparation support.
	Given a choice of an inner product on $\im(\beta)$, we can compute $L_\pre$ as the orthogonal projection $\im(\beta) \to (\I_\pre)^\perp$ up to a vector space isomorphism $(\I_\pre)^\perp \cong \linefaktor{\im(\beta)}{\I_\pre}$.
\end{remark}

Just like the concept of retrodicted variables, predicted variables also depend on the manifest variable $E$ and ready state variability $\V$. 
The scope of predicted variables, depending on a choice of a class of toy subjects, is summarised in \cref{tab:prediction}.
Moreover, as before, the predicted variable characterises the set of all predictively fixed variables for a given measurement.
Namely, the latter are given by arbitrary post-processings of a variable $L$ given by $L_\pre$ on the preparation support and the identity map outside of the preparation support. 

\begin{proposition}[Characterisation of predictively fixed variables]\label{prop:fix_predicted}
	Consider a measurement $M$ and the variable $L$ given by
	\begin{equation}\label{eq:predicted+}
		L(b) \coloneq 
			\begin{cases}
				L_\pre(b)  &  \text{if } b \in \im(\beta) \\
				b  &  \text{if } b \in \B \setminus \im(\beta),
			\end{cases}
	\end{equation}
	whose codomain is denoted by $\T$.  
	A variable $Z \colon \B \rightarrow \Z$ is predictively fixed by $M$ if and only if we have
	\begin{equation}
		Z = f \comp L,
	\end{equation}
	where $f \colon \T \to \Z$ is an arbitrary function.
\end{proposition}
Note that the codomain $\T$ of $L$ can be chosen to be the (disjoint) union of $\linefaktor{\im(\beta)}{\I_\pre}$ (i.e.\ the codomain of $L_\pre$) and $\B \setminus \im(\beta)$ (i.e.\ the set of states outside of the preparation support).
\begin{proof}
	We first show that any variable of the form $f \comp L$ satisfies \cref{eq:var_fix_prep}.
	It suffices to show that $L$ satisfies it.
	Let $v_1, v_2 \in \B \oplus \V$ be two arbitrary valid initial states that they yield identical observations in the sense that we have $\epsilon(v_1) = \epsilon(v_2)$. 
	Since $L$ is linear when restricted to the image of $\beta$, we have 
	\begin{equation}
		L \bigl( \beta(v_1)\bigr) - L \bigl( \beta(v_2)\bigr) = L \bigl( \beta(v_1 - v_2)\bigr) = 0
	\end{equation}
	by the definition of the predicted variable, since $\beta(v_1 - v_2)$ is an element of $\I_\pre$.
	This means that the composite map $L \comp \beta$ is constant on the fibers of $\epsilon$. 
	Thus, there is a function $h \colon \E \to \T$ satisfying $L \comp \beta = h \comp \epsilon$. 
	Setting $g \coloneq f \comp h$ establishes 
	\begin{equation}
		Z \comp \beta = f \comp L \comp \beta = f \comp h \comp \epsilon = g \comp \epsilon,
	\end{equation}
	i.e.\ that $Z$ is predictively fixed.
	
	Conversely, assume that a variable $Z \colon \B \to \Z$ satisfies $Z \comp \beta = g \comp \epsilon$. 
	To show that $Z = f \comp L$ holds for some function $f$, we must establish that $Z$ is constant on the fibers of $L$.
	
	To this end, we consider two ontic states $b_1, b_2 \in \B$ satisfying $L(b_1) = L(b_2)$ and distinguish two cases which correspond to the two regimes in \eqref{eq:predicted+}.

	If at least one of them is outside the image of $\beta$, this means they are equal and thus we trivially have $Z(b_1) = Z(b_2)$.
	
	Otherwise, we can assume they are both in the image of $\beta$ and so we have $s_i = \beta(v_i)$ for some $v_1, v_2 \in \B \oplus \V$.
	Furthermore, them being in the same fiber means $b_1 - b_2 \in \I_\pre$, which is equivalent to $b_1 - b_2 = \beta(v_0)$ for some $v_0 \in \ker(\epsilon)$. 
	The fact that $Z$ is predictively fixed  implies
	\begin{align}
		Z (b_2) &= Z \comp \beta(v_1 - v_0)  \\
			&= g \comp \epsilon(v_1 - v_0) = g \comp \epsilon(v_1) \\
			&= Z \comp \beta (v_1) = Z(b_1) .
	\end{align}
	Together, we have shown that $L(b_1) = L(b_2)$ implies $Z(b_1) = Z(b_2)$.
	Therefore, there exists a well-defined function $f \colon \T \to \Z$, given by $f(L(b)) \coloneq Z(b)$, such that $Z = f \comp L$. 
\end{proof}

\subsection{Consequences of Repeatability}\label{sec:bitemporal}

Having defined retrodicted variables in \cref{sec:measured}, we can now specify in detail what we mean by a measurement being repeatable, thus elaborating on \cref{sec:learning_types}.
Intuitively, the idea is that the retrodicted variable is not changed if the same measurement has been applied just before.
This requirement can be expressed as $L_{\rm ret} \comp \pi_\B = L_{\rm ret} \comp \pi_\B \comp M$, restricted to the set of valid initial states $\B \oplus \V$.
This leads to the following definition.

\begin{definition}[Repeatable measurement]\label{def:simple_repeatable_meas}
	A measurement is \newterm{repeatable} if and only if it satisfies
	\begin{equation}\label{eq:repeatable_def}
		L_{\rm ret}(b) = L_{\rm ret}\bigl( \beta(b \oplus v) \bigr)
	\end{equation}
	for all initial ontic states of the object $b \in \B$ and all $v \in \V$.
\end{definition}

Since the retrodicted variable $L_{\rm ret}$ is linear by definition, an equivalent way to express repeatability is via two matrix equations:
\begin{align}
	\label{eq:repeatability}	L_{\rm ret} \comp M_{\B\B} &= L_{\rm ret} \\
	\label{eq:no_disturbance}	L_{\rm ret} \comp M_{\B\V} &= 0.
\end{align}
Here we use the notation for blocks of the interaction matrix as in \cref{eq:blocks} and the fact that $\beta$ is given by
\begin{equation}
	\beta(b \oplus v) = M_{\B\B} b + M_{\B\V} v .
\end{equation}
To interpret these equations, recall that $L_{\rm ret}$ encodes all the information that the subject can learn about the initial state of the object through this interaction.
Moreover, the ready state variability $\V$ is the subspace of $\R$ that contains all valid initial ontic states of the ready system $R$. 
The specific initial state of $\V$ is not known to the subject, even if the fact that the initial state of $R$ lies in $\V$ is known.

\Cref{eq:no_disturbance} ensures that the uncontrollable variability in $\V$ does not influence the information learned by the second measurement.\footnote{In particular, it implies that $L_{\rm ret} \comp M_{\B\B}$ is a variable retrodictively fixed (\cref{def:fix_var}) by the second measurement.}
Thanks to this, we can interpret \cref{eq:repeatability} as stating that the value of $L_{\rm ret}$ for the pre-measurement state of $B$ (the right-hand side) is the same as its value for the post-measurement state (the left-hand side).
In other words, the two identical measurements applied in succession necessarily give rise to the same retrodictive information learned by their respective subjects.

We now show that repeatable measurements allow the toy subject to simultaneously learn about the past and the future, in the sense that $\I_\ret \supseteq \I_\pre$ holds.
In particular, this means that the subject can predict everything about the object which it can also retrodict (cf.\ \cref{cor:retrodiction_implies_prediction}).
However, since there is not an equality in general, there may be additional information that can be predicted on top of what can be retrodicted.

\begin{theorem}[Repeatability entails bitemporal learning]
	\label{thm:prediction_repeatable_kernels}
	For every repeatable measurement we have\footnotemark{}
	\footnotetext{Here, the notation in the right-hand side refers to the span of the union of the two subspaces.}%
	\begin{equation}\label{eq:prediction_repeatable_1}
		\I_\ret \cap \im(\beta) = \I_\pre + \beta \bigl( \{0\} \oplus \V \bigr).
	\end{equation}
\end{theorem}
\begin{proof}	
	By repeatability (which can be expressed as $L_\ret \comp \beta = L_\ret \comp \pi_\B$), we have
	\begin{equation}
		L_\ret\bigl(\I_\pre\bigr) = L_\ret \bigl(\beta(\ker \epsilon) \bigr) = L_\ret \bigl( \pi_\B(\ker \epsilon) \bigr) = L_\ret\bigl(\I_\ret\bigr) = \{0\}
	\end{equation}
	so that $\ker(L_\ret) \supseteq \I_\pre$ follows.
	Using $\ker(L_\ret) = \I_\ret$ (\cref{def:measured}), this implies
	\begin{equation}\label{eq:prediction_repeatable_2}
		\I_\ret \cap \im(\beta) \supseteq \I_\pre,
	\end{equation}
	as $\im(\beta) \supseteq \I_\pre$ holds by definition.
	
	To establish the $\supseteq$ inclusion in \cref{eq:prediction_repeatable_1}, it thus suffices to further show 
	\begin{equation}
		\beta(0 \oplus v) \in \I_\ret
	\end{equation}
	for every $v \in \V$.
	By $\beta(0 \oplus v) = M_{\B\V} v$ and $\I_\ret = \ker(L_\ret)$, this is a direct consequence of the ``no disturbance'' condition, i.e.\ \cref{eq:no_disturbance}.
	
	For the converse set inclusion, 
	assume $b \in \I_{\ret}$ and $b \in \im(\beta)$.
	The latter means that there exists an initial state $b' \oplus v' \in \B \oplus \V$ such that $b = \beta(b' \oplus v')$ holds.
	Applying the repeatability condition \eqref{eq:repeatable_def} gives
	\begin{equation}
		L_\ret \bigl(b' - \beta(b' \oplus v') \bigr) = 0,
	\end{equation}
	i.e.\ we have
	\begin{equation}
		b' - \beta(b' \oplus v') \in \ker(L_\ret) = \I_{\ret},
	\end{equation}
	which means $b' - b \in \I_{\ret}$.
	Since $b$ is an element of $\I_{\ret}$ by assumption, and $\I_{\ret}$ is a linear subspace, it follows that $b' \in \I_{\ret} = \pi_\B(\ker \epsilon)$ also holds.
	There must thus exist a state $v_0 \in \V$ satisfying $b' \oplus v_0 \in \ker \epsilon$.
	By the definition of the predictively inaccessible subspace, mapping this kernel element through $\beta$ gives 
	\begin{equation}\label{eq:prediction_repeatable_3}
		\beta(b' \oplus v_0) \in \I_{\pre}.
	\end{equation}
	Now, we can write
	\begin{equation}
		b = \beta(b' \oplus v_0) + \beta(0 \oplus v)
	\end{equation}
	where $ v \coloneq v' - v_0 \in \V$, which establishes
	\begin{equation}
		\I_\ret \cap \im(\beta) \subseteq \I_\pre + \beta \bigl( \{0\} \oplus \V \bigr)
	\end{equation}
	and completes the proof.
\end{proof}

As \cref{eq:prediction_repeatable_1} shows, there are two sources of discrepancy between the predicted and retrodicted information. 

One is the preparation support $\im(\beta)$, which just means that some of the properties learned about the initial state of $B$ can be trivially true for its post-measurement state by virtue of the fact that all valid post-measurement states have these properties. 
In other words, empirical information about the past can turn into non-empirical information about the future.

The other is the influence of the ready state variability on the object's post-measurement state, i.e.\ $\beta \bigl( \{0\} \oplus \V \bigr)$.
We can interpret this as saying that in principle there can be additional information predicted for a repeatable measurement, if one can create a correlation between the final state of $B$ and the subject's manifest variable by ``broadcasting'' the ready state variability to both. 
Such information is not required to be retrodicted, because the correlation is generated by the measurement interaction, and repeatability does not imply that a second subject applying the same measurement procedure will be able to retrodict this information.

\begin{corollary}\label{thm:outcome_relation}
	For a repeatable measurement, the inclusion map $\iota \colon \im(\beta) \hookrightarrow \B$ induces a well-defined linear map
	\begin{equation}\label{eq:outcome_relation}
		\begin{split}
			[\iota] \colon \linefaktor{\im(\beta)}{\I_\pre} &\to \linefaktor{\B}{\I_\ret} \\
			[b]_\pre &\to [b]_\ret,
		\end{split}
	\end{equation}
	where $[b]_\pre$ and $[b]_\ret$ denote the equivalence classes of an $b \in \im(\beta)$ in the respective quotient spaces.
\end{corollary}
\begin{proof}
	The fact that $[\iota]$ is well-defined is equivalent to 
	\begin{equation}
		\forall \, b , b' \in \im(\beta) \quad : \quad b' - b \in \I_\pre  \implies  b' - b \in \I_\ret,
	\end{equation}
	which is a direct consequence of \cref{thm:prediction_repeatable_kernels}.
\end{proof}

If the inclusion
\begin{equation}\label{eq:lower_bound}
	\beta \bigl( \{0\} \oplus \V \bigr) \subseteq \I_\pre
\end{equation}
also holds, then the map $[\iota]$ is in fact injective, which also implies ${\linefaktor{\im(\beta)}{\I_\pre} \cong \linefaktor{\im(\beta)}{\I_\ret}}$.
While \eqref{eq:lower_bound} holds for many measurements of interest (it only fails for $M^{\rm pred}$ among the four options in \cref{tab:basic-interactions}), we do not currently know of a useful characterisation of those that do not satisfy it.

\begin{corollary}
	\label{thm:prediction_repeatable}
	For every repeatable measurement, we have
	\begin{equation}\label{eq:prediction_repeatable}
		L_\ret(b) = [\iota] \comp L_{\pre}(b) \qquad \forall \, b \in \im(\beta),
	\end{equation}
	where $[\iota]$ is the map from \eqref{eq:outcome_relation}.
\end{corollary}

\begin{proof}
	Using \cref{def:predicted,def:measured} as well as the notation from \eqref{eq:outcome_relation}, we have
	\begin{equation}
		L_\ret(b) = [b]_\ret = [\iota] \bigl( [b]_\pre \bigr) = [\iota] \bigl( L_\pre(b) \bigr)
	\end{equation}
	for any $b \in \im(\beta)$, which is precisely \cref{eq:prediction_repeatable}.
\end{proof}

\begin{corollary}\label{cor:retrodiction_implies_prediction}
	For every repeatable measurement, the retrodicted variable is also predictively fixed.
\end{corollary}
By \cref{prop:fixed_from_measured}, this also means that every retrodictively fixed variable is predictively fixed in the repeatable case.
\begin{proof}
	This follows by combining \cref{thm:prediction_repeatable} with \cref{prop:fix_predicted}.
	In particular, given $L$ from \cref{eq:predicted+}, we have 
	\begin{equation}
		L_\ret = f \comp L
	\end{equation}
	where $f \colon \T \to \linefaktor{\B}{\I_{\ret}}$ is given by
	\begin{equation}
		f(t) = 
		\begin{cases}
			[\iota](t)  &  \text{if } t \in \linefaktor{\im(\beta)}{\I_{\pre}} \\
			L_{\ret}(t)  &  \text{if } t \in \B \setminus \im(\beta),
		\end{cases}
	\end{equation}
	where we use the map $[\iota] \colon \linefaktor{\im(\beta)}{\I_\pre} \to \linefaktor{\B}{\I_\ret}$ from \cref{thm:outcome_relation}.
	This implies that $L_\ret$ is predictively fixed.
\end{proof}

The idea of repeatable measurements is that their outcome is the same when the measurement is repeated. 
In our framework, this idea translates as follows: 
Given two successive identical\footnotemark{} measurements of a toy object $B$, the respective toy subjects involved in these two measurements would agree on what the value of their retrodicted variable is. 
\footnotetext{That is, the two measurements are implemented by the same interaction $M$, they make use of isomorphic ready systems, and the two subjects have isomorphic ontic state spaces and the same manifest variable $E$. 
However, they do not share the initial state of the ready system, since these are by assumption uncontrollable.}%
That is, they agree which element of the partition $\linefaktor{\B}{\I_\ret}$ of $\B$ was occupied prior to their measurement.
Mathematically, we can describe this information by the map $L_\ret \comp \pi_\B \colon \B \oplus \V \to \linefaktor{\B}{\I_\ret}$ (cf.\ \cref{eq:repeatable_2}).

Since the measurements are implemented sequentially, the ontic state prior to the second measurement is equal to the ontic state after the first measurement.
This is mathematically described by the fact that the initial state for the second measurement is given by 
\begin{equation}
	\beta(b \oplus v_1) \oplus v_2 \in \B \oplus \V,
\end{equation}
where $b \oplus v_1$ is the initial state for the first measurement.
That is why repeatable measurements satisfy
\begin{equation}\label{eq:repeatable_2}
	L_\ret \bigl( \pi_\B(b \oplus v_1) \bigr) = L_\ret \Bigl( \pi_\B \bigl(\beta(b \oplus v_1) \oplus v_2 \bigr) \Bigr)
\end{equation}
for all $b \in \B$ and  $v_1, v_2 \in \V$.
The left-hand side represents the outcome of the first measurement and the right-hand side represents the outcome of the second measurement.
By the definition of $\pi_\B$ via $b \oplus v \mapsto b$, \cref{eq:repeatable_2} is equivalent to \cref{eq:repeatable_def}.

We can now show that the value of the predicted variable is also preserved when a repeatable measurement is applied in succession, but only when we discard the additional predicted information via $[\iota]$.
The relevant property to establish is obtained by replacing $L_\ret \comp \pi_\B$ with $[\iota] \comp L_\pre \comp \beta$ in \cref{eq:repeatable_2}.
\begin{proposition}[Stability of prediction in the repeatable case]\label{prop:repeatable_predicted}
	For every repeatable measurement, the predicted variable satisfies\/\footnotemark{}
	\footnotetext{In \cref{eq:repeatable_predicted} as well as in the proof below, we implicitly restrict the codomains of $\beta$, $M_{\B \B}$, and $M_{\B \V}$ to the domain of $L_\pre$, which is the preparation support $\im(\beta)$.}
	\begin{equation}\label{eq:repeatable_predicted}
		[\iota] \comp L_\pre \bigl( \beta(b \oplus v_1) \bigr) = [\iota] \comp L_\pre \Bigl( \beta \bigl(\beta(b \oplus v_1) \oplus v_2 \bigr) \Bigr)
	\end{equation}
	for all $b \in \B$ and all $v_1, v_2 \in \V$.
\end{proposition}
\begin{proof}
	Using $\beta = M_{\B \B} + M_{\B \V}$, \cref{thm:prediction_repeatable}, and the linearity of $L_\ret$, we can write \cref{eq:repeatable_predicted} as 
	\begin{equation}
		L_\ret M_{\B \B} b + L_\ret M_{\B \V} v_1 = L_\ret M_{\B \B}^2 b + L_\ret M_{\B \B} M_{\B \V} v_1 + L_\ret M_{\B \V} v_2
	\end{equation}
	which is equivalent to the following three equations
	\begin{align}
		\label{eq:repeatable_predicted_1} L_\ret M_{\B \B} &= L_\ret M_{\B \B}^2 \\
		\label{eq:repeatable_predicted_2} L_\ret M_{\B \V} &= L_\ret M_{\B \B} M_{\B \V} \\
		\label{eq:repeatable_predicted_3} 0 &= L_\ret M_{\B \V}.
	\end{align}
	Since the measurement is repeatable, we have $L_\ret = L_\ret M_{\B \B}$ and $0 = L_\ret M_{\B \V}$, from which all of the equations follow.
\end{proof}

\section{Proofs of Epistemic Horizons}\label{sec:proof_EH}

\subsection{Retrodiction}\label{sec:proofs_retrodictive}

\begin{proof}[\bfseries Proof of \cref{thm:Poisson_MV_Poisson_RS}]
	There are two statements to show.
	First, we consider an arbitrary Poisson variable $Z \colon \B \to \Z$ and show that it can be retrodicted.
	This amounts to constructing a measurement whose retrodicted variable is $Z$.
	Second, we show that for arbitrary subjects of the above kind, the retrodicted variable must be a Poisson variable.
	
	\paragraph{Part 1 (sufficiency):}
	Consider a generic toy object $B$ with the underlying symplectic vector space $\B$ and an arbitrary Poisson variable $Z \colon \B \to \Z$.
	Recall that $Z$ is surjective by definition.
	
	We then choose a ready system $R$ and a toy subject $S$ as follows:
	Their ontic state spaces are each given by $\Q \oplus \P$, where $\Q$ and $\P$ are vector spaces isomorphic to $\Z$.
	We let the symplectic form have the matrix representation \eqref{eq:sympl_form} with respect to the decomposition $\S = \Q \oplus \P = \R$, so that both $\Q$ and $\P$ are Lagrangian subspaces.
	The ready state variability of the ready system is chosen to be $\V \coloneq \P$ and the manifest variable of the subject is $\pi_\Q \colon \B \to \Q$. 
	
	We now explicitly construct a transformation $\B \oplus \R \to \B \oplus \S$ that measures $Z$. 
	In matrix notation, it is given by
	\begin{equation}\label{eq:M_Pson}
		M^{\rm Pson} \coloneq  \begin{pmatrix}
			\id_\B & 0 & \Omega_{\B} Z^T \\
			Z & \id_\Q & 0 \\
			0 & 0 & \id_\P
		\end{pmatrix},
	\end{equation} 
	in the block form relative to the decomposition $\B \oplus \Q \oplus \P$ of both its input and output spaces.
	
	To show that $M^{\rm Pson}$ is Poisson, we compute
	\begin{equation}
		M^{\rm Pson} \Omega \left( M^{\rm Pson} \right)^T = 
		\begin{pmatrix}
			\Omega_{\B} & 0 & 0 \\
			0 & Z \Omega_{\B} Z^T & \id  \\
			0 & - \id & 0
		\end{pmatrix} = \Omega,
	\end{equation}
	where we use $Z \Omega_\B Z^T = 0$, i.e.\ that $Z$ is a Poisson variable.
	
	The fastest way to compute the variable retrodicted by $M^{\rm Pson}$ is to use \cref{prop:concrete_measurable}.
	We can choose the free manifest subspace $\F$ from its statement to be $\Q$, so that the retrodicted variable is $Z$ as required.
	
	\paragraph{Part 2 (necessity):}
	For the other direction, consider an arbitrary measurement interaction $M \colon \B \oplus \R \to \B \oplus \S$ together with a Poisson ready state variability $\V \subseteq \R$ and a subject with a Poisson manifest variable $E \colon \S \to \E$.
	
	We choose a free manifest subspace $\F \subseteq \E$ as in \cref{prop:concrete_measurable}, coming also with a projection $\pi_\F \colon \E \to \F$ with respect to the decomposition $\E = \F \oplus \C$.
	By \cref{prop:concrete_measurable}, the retrodicted variable satisfies $L_{\rm ret} = \overline{M_{\F \B}}$, where $M_{\F \B} \coloneq \pi_\F \comp M_{\E \B}$, $M_{\E \B} \coloneq E \comp M_{\S \B}$, and $\overline{M_{\F \B}}$ denotes the restriction of its codomain to its own image.
	Thus, to complete this part of the proof, it suffices to show that $M_{\F \B}$ must be a Poisson variable.
	
	Given a choice of a Darboux basis $\{q_1, \ldots, q_n, p_1, \ldots, p_n\}$ of $\R$ as in \cref{sec:toy_systems}, we also choose an inner product on $\R$, such that this basis is an orthonormal one.
	Then we can define an orthogonal complement $\W \subseteq \R$ of the ready state variability $\V$ as well as the orthogonal projection $\pi_{\W} \colon \R \to \W$.
	Since the kernel of $\pi_{\W}$ is the coisotropic subspace $\V$, it is a Poisson variable.
	
	Let us now use the fact that $M$ is a Poisson map, i.e.\
	\begin{equation}
		M \Omega_{\B \oplus \R} M^T = \Omega_{\B \oplus \S}.
	\end{equation}
	Multiplying this equation by $\pi_\F E \pi_{\S}$ from the left and by its transpose from the right gives
	\begin{equation}\label{eq:ops_FF}
		M_{\F \B} \Omega_{\B} M_{\F \B}^T + M_{\F \R} \Omega_{\R} M_{\F \R}^T = \pi_\F E \Omega_{\S} E^T \pi_{\F}^T.
	\end{equation}
	Since $E$ is a Poisson variable, the right-hand side vanishes.
	Moreover, by the defining properties of a free manifest subspace (see \cref{prop:concrete_measurable}), we have $M_{\F \V} = 0$ so that also
	\begin{equation}
		M_{\F \R} = M_{\F \V} \pi_\V + M_{\F \W} \pi_\W = M_{\F \W} \pi_\W
	\end{equation}
	holds.
	Now we use the fact that $\pi_\W$ is a Poisson variable to obtain
	\begin{equation}
		M_{\F \R} \Omega_{\R} M_{\F \R}^T = M_{\F \W} \pi_\W \Omega_{\R} \pi_\W^T M_{\F \W}^T = 0.
	\end{equation}
	Putting all of these arguments together, \cref{eq:ops_FF} thus implies
	\begin{equation}
		M_{\F \B} \Omega_{\B} M_{\F \B}^T = 0,
	\end{equation}
	i.e.\ $M_{\F \B}$ is a Poisson variable, as we wanted to show.
\end{proof}

\begin{proof}[\bfseries Proof of \cref{thm:Poisson_MV_trivial_RS}]
	In the case of trivial ready states, we have $\V = \R$ and $\F = \E$ (using the notation from \cref{sec:subjects,sec:measured}).
	By \cref{prop:concrete_measurable}, the retrodicted variable (for a given measurement) is $E \comp M_{\S \B}$. 
	
	For the swapping interaction $M^{\swap}$ from \eqref{eq:swap}, the block $M_{\S \B}$ is given by the identity map.
	Since $E$ can be an arbitrary (maximal) Poisson variable, we obtain that any maximal Poisson variable can be retrodicted.
	To retrodict other Poisson variables, one can either invoke \cref{lem:Lag_rotation} or use the measurement $M^{\rm Pson}$ from the proof of \cref{thm:Poisson_MV_Poisson_RS}.
	
	To show that no other variable can be retrodicted, we use the fact that every trivial ready state is also a Poisson ready state.
	Specifically, because subjects with access to Poisson ready states cannot retrodict variables other than Poisson ones by \cref{thm:Poisson_MV_Poisson_RS}, the same limitation follows for subjects with access to trivial ready states only.
\end{proof}

\begin{proof}[\bfseries Proof of \cref{thm:Poisson_MV_fixed_RS}]
	Consider a generic toy object $B$ with a decomposition $\B = \G \oplus \sfH$ of its underlying vector space such that we can write its symplectic form as
	\begin{equation}
		\Omega_{\B} = 
		\begin{pmatrix}
			0 & \id \\
			-\id & 0
		\end{pmatrix}
	\end{equation}
	with respect to it, as usual.
	
	We choose the toy subject $S$ and ready system $R$ to be toy systems with twice as many degrees of freedom as the toy object.
	Namely, their underlying vector spaces of ontic states, and their symplectic form, are given by 
	\begin{equation}
		\R = \S \coloneq \Q \oplus \P   \qquad  \text{and}  \qquad
		\Omega_{\S} \coloneq \begin{pmatrix}
			0 & \id  \\
			- \id & 0
		\end{pmatrix}
	\end{equation}
	respectively, where $\Q$ and $\P$ are (canonically isomorphic) copies of the vector \mbox{space $\B$}.
	Note that even though $\B \oplus \B$ and $\S$ are isomorphic as vector spaces, they are not isomorphic as symplectic vector spaces, since the former has symplectic form given by $\Omega_{\B} \oplus \Omega_{\B}$, which is distinct from $\Omega_{\S}$ given above.
	
	The manifest variable of the subject is the projection $\pi_\Q \colon \Q \oplus \P \to \Q$ onto its first component, which is a maximal Poisson variable (since $\Q$ is a Lagrangian subspace of $\S$).
	Moreover, the ready state is complete, i.e.\ its variability is $\{0\}$. 
	
	With these choices, we can construct a Poisson map, for which the identity variable $\id_{\B}$ is retrodicted.
	It is specified by a symplectic linear map $M^{\rm full} \colon \B \oplus \R \to \B \oplus \R$ as follows:
	\begin{equation}\label{eq:M_full}
		\begin{aligned}
			M^{\rm full}_{\B\B} &= \id & &\qquad\qquad &
			M^{\rm full}_{\B\R} &=
			\begin{pmatrix}
				\id & 0 
			\end{pmatrix} \\
			M^{\rm full}_{\S\B} &=
			\begin{pmatrix}
				\id  \\
				0 
			\end{pmatrix} & &\qquad\qquad &
			M^{\rm full}_{\S\R} &=
			\begin{pmatrix}
				\pi_{\G} & \Omega_{\B} \\
				- \Omega_{\B} & 0
			\end{pmatrix},
		\end{aligned}
	\end{equation}
	where 
	\begin{equation}\label{eq:object_projection}
		\pi_{\G} = 
		\begin{pmatrix}
			\id & 0 \\
			0 & 0
		\end{pmatrix}
	\end{equation}
	is a linear map $\B \to \B$ that projects onto the Lagrangian subspace $\G$ of $\B$.
	
	To show that $M^{\rm full}$ is a Poisson map, we compute
	\begin{equation}
		M^{\rm full} \Omega \left( M^{\rm full} \right)^T = 
		\begin{pmatrix}
			\Omega_{\B} & 0 & 0 \\
			0 & \Omega_{\B} - \pi_{\G} \Omega_{\B} - \Omega_{\B} \pi_{\G} & - \Omega_{\B}^2 \\
			0 & \Omega_{\B}^2 & 0
		\end{pmatrix}
	\end{equation}
	and note that we have $\Omega_{\B}^2 = - \id$ and $\pi_{\G} \Omega_{\B} + \Omega_{\B} \pi_{\G} = \Omega_{\B}$.
	
	Note that by the assumption of a complete ready state, the map $\epsilon$ is given by $\pi_\Q \comp M^{\rm full}_{\S\B}$, which is the identity on $\B$.
	As a consequence, we have shown that the identity variable of any system can be retrodicted by subjects with Poisson manifest variables and access to complete ready states.
	
	To show that any other surjective linear variable can also be retrodicted, we use \cref{lem:Lag_rotation}.
	Consider an arbitrary surjective linear variable $Z \colon \Q \to \Z$. 
	Furthermore, define a toy subject with underlying state space $\T = \Z \oplus \W$ for which both $\Z$ and $\W$ are Lagrangian subspaces.
	Choose the manifest variable to be $\pi_\Z$.
	By \cref{lem:Lag_rotation}, there is a Poisson map $N \colon \R \to \T$, whose $\Z \Q$ block is given by $Z$.
	Applying this physical transformation after the measurement interaction $M$, we obtain a measurement 
	\begin{equation}\label{eq:comp_meas}
		(\id_{\B} \oplus N) \comp M^{\rm full} \colon \B \oplus \R \to \B \oplus \T,
	\end{equation}
	whose retrodicted variable is $Z$.
	
	In conclusion, we have shown that any surjective linear variable $Z$ can be retrodicted for toy subjects with Poisson manifest variables and access to complete ready states.
\end{proof}

\begin{theorem}[EH for retrodiction; complete manifest variable]\label{thm:complete_MV_retrodictive}
	For toy subjects whose manifest variable is complete and ready state is trivial, every surjective linear variable, whose kernel is a symplectic subspace of $\B$, can be retrodicted.
	The same holds for subjects whose ready state is Poisson, as well as for those whose ready state is complete.\footnotemark{}
\end{theorem}
\footnotetext{We do not know whether other variables, one whose kernel is not a symplectic subspace, can also be retrodicted for Poisson ready states.
The case of complete ready states is covered by \cref{thm:repeatable_complete_MV_complete_RS}.
However, answering this question is irrelevant as far as the epistemic horizon is concerned, because identity maps can be retrodicted according to \cref{thm:complete_MV_retrodictive}, and so every variable is retrodictively fixed by \cref{prop:fixed_from_measured}.}%

\begin{proof}[\bfseries Proof of \cref{thm:complete_MV_retrodictive}]
	Consider a generic toy object $B$, a ready system with identical state space $\R = \B$ and variability $\V \subseteq \B$, and a toy subject also with the same state space, $\S = \B$, and a complete manifest variable $\id \colon \S \to \S$.
	Then we can use the swapping interaction $M^{\rm swap}$ from \eqref{eq:swap} to show that the identity variable can be retrodicted. 
	Indeed, the kernel of $\epsilon \colon \B \oplus \V \to \S$ is given by $\{0\} \oplus \V$ and so the retrodictively inaccessible subspace $\I_{\rm ret} = \pi_\B(\ker\epsilon)$ is $\{0\}$.
	In conclusion, for any of the three types of subjects, the identity variable can be retrodicted.
	
	Moreover, we can post-compose $M^{\rm swap}$ with $\id_\B \oplus N$ where $N \colon \S \to \T$ is any  Poisson map and $\T$ is any symplectic vector space. 
	Then the variable retrodicted by the composite measurement is $N$.
	By \cref{lem:ops_forms}, $N$ is Poisson if and only if its kernel is a symplectic subspace.
	As a consequence, every variable whose kernel is symplectic can be retrodicted regardless of the ready state.
\end{proof}

\subsection{Repeatable Measurements}\label{sec:proofs_repeatable}

\begin{proof}[\bfseries Proof of \cref{thm:repeatable:Poisson_MV_trivial_RS}]
	It suffices to show that the retrodictively inaccessible subspace $\I_{\rm ret}$ must be the whole of $\B$ for subjects with a complete manifest variable.
	The same conclusion then follows for those with a Poisson manifest variable, since post-composing $\epsilon$ with a Poisson manifest variable $E$ can only enlarge the kernel of $\epsilon$.
	
	For the trivial ready state, we have $\V = \R$, and for a complete manifest variable we have 
	\begin{equation}
		M(b \oplus r) = \beta(b \oplus r) \oplus \epsilon(b \oplus r).
	\end{equation}
	Since $M$ is Poisson, it is surjective by \cref{lem:ops_surj}.
	Thus for every $b'\in \B$, there are $b \in \B$ and $r \in \R$ satisfying
	\begin{equation}
		M(b \oplus r) = b' \oplus 0.
	\end{equation}
	By definition of the retrodictively inaccessible subspace (\cref{def:retr_inaccessible}), $b$ is an element of $\I_{\rm ret}$.
	
	Moreover, we have
	\begin{equation}
		L_{\rm ret}(b) = L_{\rm ret} \comp \beta (b \oplus r) = L_{\rm ret} (b')
	\end{equation}
	by repeatability and the previous two equations.
	In particular, by definition of the retrodicted variable (\cref{def:measured}), this means $b' - b \in \I_{\rm ret}$ and consequently $b' \in \I_{\rm ret}$.
	Since $b'$ was arbitrary, we obtain $\I_{\rm ret} = \B$, which implies that the retrodicted variable is constant.
\end{proof}

\begin{proof}[\bfseries Proof of \cref{thm:repeatable_Poisson_MV_Poisson_RS}]
	The proof of \cref{thm:Poisson_MV_Poisson_RS} shows that any Poisson variable can be retrodicted by $M^{\rm Pson}$ from \cref{eq:M_Pson}.
	So it suffices to show that $M^{\rm Pson}$ is in fact a repeatable measurement.
	Specifically, \cref{eq:repeatability} follows by $M^{\rm Pson}_{\B\B} = \id_\B$ while \cref{eq:no_disturbance} follows from 
	\begin{equation}
		L_{\rm ret} \comp M^{\rm Pson}_{\B \V} = Z \Omega_{\B} Z^T = 0
	\end{equation}
	where the first equation holds by $Z = L_{\rm ret}$ and the second holds because $Z$ is a Poisson variable.
	
	Since every variable that can be repeatably retrodicted can also be retrodicted by a generic measurement, the epistemic horizon from \cref{thm:Poisson_MV_Poisson_RS} implies that no other variables can be repeatably retrodicted.
\end{proof}

\begin{proof}[\bfseries Proof of \cref{thm:repeatable_Poisson_MV_complete_RS}]
	We show in the proof of \cref{thm:Poisson_MV_fixed_RS} that the variable retrodicted by the interaction $M^{\rm full}$ defined via \cref{eq:M_full} is the identity map $\id_\B$.
	Since we also have $M^{\rm full}_{\B \B} = \id_\B$, the measurement is repeatable.
	In particular, note that the ``no disturbance'' condition{\,---\,}\cref{eq:no_disturbance}{\,---\,}holds trivially because the initial ontic state of the ready system is completely fixed, i.e.\ $\V = \{0\}$.
	Therefore, the identity variable can be repeatably retrodicted.
	                                                    
	To show that any other (surjective linear) variable can also be repeatably retrodicted, we proceed as in the proof \cref{thm:Poisson_MV_fixed_RS}. 
	Specifically, the composite measurement from \eqref{eq:comp_meas}, whose retrodicted variable is (an arbitrary) surjective linear map, is also repeatable.
	To see this, we note that its $\B \B$ block is still equal to the identity map.
\end{proof}

\begin{proof}[\bfseries Proof of \cref{thm:repeatable_complete_MV_Poisson_RS}]
	The fact that every Poisson variable can be repeatably retrodicted follows from the construction of $M^{\rm Pson}$ from \cref{eq:M_Pson}, which was used in the proof of \cref{thm:Poisson_MV_Poisson_RS}.
	As noted in the proof of \cref{thm:repeatable_Poisson_MV_Poisson_RS}, this measurement is repeatable for a Poisson manifest variable.
	The same is true for a complete manifest variable, because the variable retrodicted by it is the same.
	Indeed we can compute the retrodicted variable easily thanks to \cref{prop:concrete_measurable} and verify that it is given by $Z${\,---\,}an arbitrary Poisson variable.
	
	To prove the converse, consider an arbitrary repeatable measurement $M$ and choose a free manifest subspace $\F \subseteq \S$, as in \cref{prop:concrete_measurable}.
	By this result, the retrodicted variable satisfies $L_{\rm ret} = \overline{M_{\F \B}}$, where $\overline{M_{\F \B}}$ is as in the second part of the proof of \cref{thm:Poisson_MV_Poisson_RS}.
	
	Thus it suffices to show that $M_{\F\B}$ is a Poisson variable, i.e.\ that we have $M_{\F\B} \Omega_{\B} M_{\F\B}^T = 0$.
	To this end, we again choose an inner product on $\R$, as in the proof of \cref{thm:Poisson_MV_Poisson_RS}.
	Then we can define an orthogonal complement $\W \subseteq \R$ of the ready state variability $\V$ as well as the orthogonal projection $\pi_{\W} \colon \R \to \W$.
	Since the kernel of $\pi_{\W}$ is the coisotropic subspace $\V$, it is a Poisson variable.
	
	Let us use the fact that $M$ is Poisson, namely 
	\begin{equation}\label{eq:M_trans_sympl}
		M \Omega M^T = \Omega, \qquad \text{where }
		\Omega = \begin{pmatrix}
			\Omega_\B & 0 \\
			0 & \Omega_{\R} 
		\end{pmatrix} .
	\end{equation} 
	Its $\B\F$ block is
	\begin{equation} \label{eq:sympl2} 
		M_{\B\B} \Omega_\B M_{\F\B}^T  +  M_{\B\R} \Omega_\R M_{\F\R}^T  =  0 .
	\end{equation}
	Multiplying by $M_{\F\B}$ from the left, and using the repeatability condition of \cref{eq:repeatability} gives 
	\begin{equation}\label{eq:contr2}
		M_{\F\B} \Omega_\B M_{\F\B}^T + M_{\F\B} M_{\B\R} \Omega_\R M_{\F\R}^T = 0 .
	\end{equation}
	By the no disturbance condition of \cref{eq:no_disturbance}, we then have 
	\begin{equation}
		M_{\F\B} M_{\B\R} = M_{\F\B} M_{\B\W} \pi_{\W} + M_{\F\B} M_{\B\V} \pi_\V = M_{\F\B} M_{\B\W} \pi_{\W}
	\end{equation}
	where $\pi_\V \colon \R \to \V$ is the orthogonal projection onto $\V$. 
	By the definition of the free manifest subspace we also have $M_{\F \V} = 0$, which implies
	\begin{equation}
		M_{\F \R} = M_{\F \W} \pi_{\W}.
	\end{equation}
	Using both of these in \cref{eq:contr2} gives
	\begin{equation}
		M_{\F\B} \Omega_\B M_{\F\B}^T + M_{\F\B} M_{\B\W} \pi_{\W} \Omega_\R \pi_{\W}^T M_{\F\W}^T = 0.
	\end{equation}
	Since $\pi_{\W}$ is a Poisson variable, it satisfies $\pi_{\W} \Omega_\R \pi_{\W}^T = 0$ and thus we get that $M_{\F\B}$ must also be a Poisson variable.
	This completes the proof of the theorem.
\end{proof}

\begin{proof}[\bfseries Proof of \cref{thm:repeatable_complete_MV_complete_RS}]
	This proof is essentially identical to that of \cref{thm:repeatable_Poisson_MV_complete_RS}, because the change in manifest variable does not change the variable retrodicted by $M^{\rm full}$.
	In particular, in this case we have $\epsilon = M^{\rm full}_{\S \B}$.
	By the definition of the interaction matrix via \cref{eq:M_full}, the kernel of $\epsilon$ is thus the identity which can be retrodicted.
\end{proof}

\subsection{Prediction}\label{sec:proofs_predictive}

\begin{proof}[\bfseries Proof of \cref{thm:prediction_Poisson_MV_trivial_RS}]
	The proof follows the same steps as that of \cref{thm:repeatable:Poisson_MV_trivial_RS}.	
	In particular, we can restrict our attention to complete manifest variables.
	
	The fact that the preparation support $\im(\beta)$ must be the whole of $\B$ follows from the surjectivity of the measurement (\cref{lem:ops_surj}) and the fact that for trivial ready states, the variability is unrestricted: $\V = \R$.
	
	Surjectivity of $M$ also implies that every output like $b' \oplus 0 \in \B \oplus \S$ is in the image of $M$, which implies $b' \in \beta(\ker \epsilon) = \I_\pre$.
	Therefore, we have $\I_\pre = \B$, which means that the predicted variable is necessarily constant.
\end{proof}

\begin{proof}[\bfseries Proof of \cref{thm:prediction_Poisson_MV_Poisson_RS}]
	We split the argument into two parts. 
	In the first, we show that for every measurement by a toy subject with a Poisson manifest variable and Poisson ready state, the predictively inaccessible subspace is necessarily coisotropic.
	By \cref{lem:Poisson}, this is equivalent to showing that the canonical quotient map $Z \colon \B \to \linefaktor{\B}{\I_\pre}$ is Poisson (i.e.\ that it satisfies $Z \Omega Z^T = 0$).
	In the second part, provided an arbitrary coisotropic subspace $\I$ of $\B$ and an arbitrary linear subspace $\F$ satisfying $\F \supseteq \I$,\footnotemark{} we then construct a measurement whose preparation support coincides with $\F$ and whose predictively inaccessible subspace coincides with $\I$.
	\footnotetext{Note that the subspace $\F$ must necessarily be coisotropic as well, since taking symplectic complements reverses subspace inclusions.}%
	
	\paragraph{Part 1 (necessity):}	
	While for a generic Poisson ready state, the variability $\V$ is a coisotropic subspace of $\R$, in this proof we can restrict our attention to the case of $\V$ being a Lagrangian subspace.
	To see this, recall that every coisotropic subspace $\tilde{\V}$ contains a Lagrangian subspace $\V$ within. 
	In this proof, we show (the contrapositive of) the implication
	\begin{center}
		$\I \subseteq \B$ is not coisotropic $\implies$ the quotient projection $Z \colon \B \to \linefaktor{\B}{\I}$ is not predictively fixed,
	\end{center}
	assuming that $\V$ is Lagrangian.
	By \cref{lem:subject_retrodictive_power}, if $Z$ is not predictively fixed for a Lagrangian $\V$, then it is also not predictively fixed for the coisotropic $\tilde{\V}$ satisfying $\tilde{\V} \supseteq \V$.
	This in turn implies that $\I$ cannot be a predictively inaccessible subspace, provided that $\tilde{\V}$ is the kernel of the ready state.
	To summarise this paragraph, in the rest of the proof we can safely assume $\V$ to be a Lagrangian subspace of $\R$.
	
	To simplify the argument structure, we choose decompositions of the vector spaces of the ready system and the toy subject 
	\begin{equation}
		\R = \W \oplus \V  \qquad \qquad  \S = \Q \oplus \Y
	\end{equation}
	where $\V$ is the ready state variability, $\Y$ is the kernel of the manifest variable $E \colon \S \to \E$, both Lagrangian subspaces, and their complements are chosen (together with isomorphisms $\W \cong \V$ and $\Q \cong \Y$) so that we have
	\begin{equation}
		\Omega_{\R} = \begin{pmatrix}
			0 & \id \\
			-\id & 0
		\end{pmatrix} 
		\qquad \qquad
		\Omega_{\S} = \begin{pmatrix}
			0 & \id \\
			-\id & 0
		\end{pmatrix}
	\end{equation}
	in the corresponding block notation.
	Using the above decompositions and the block notation for the interaction matrix $M$ in consideration, we have
	\begin{equation}
		E \comp M_{\Y \B} = 0 =  E \comp M_{\Y \V}
	\end{equation}
	by the definition of $\Y$.
	
	Note that, for the predictively inaccessible subspace $\I_\pre$ associated to $M$, the quotient map $Z \colon \B \to \linefaktor{\B}{\I_\pre}$ is predictively fixed by \cref{prop:fix_predicted}, i.e.\ we have $Z \beta = g \epsilon$ for some function $g$.
	What is more, the map ${g \colon \E \to \linefaktor{\B}{\I_\pre}}$ can be chosen to be linear.
	To simplify the notation further, we denote $\Z \coloneq \linefaktor{\B}{\I_\pre}$ and define the linear maps
	\begin{equation}
		M_{\Z \B} \coloneq g \comp E \comp M_{\Q \B}  \quad \text{and} \quad  M_{\Z \V} \coloneq g \comp E \comp M_{\Q \V}.
	\end{equation}
	Using 
	\begin{equation}
		\beta = M_{\B \B} + M_{\B \V} \quad \text{and} \quad \epsilon = E \comp M_{\Q \B} + E \comp M_{\Q \V},
	\end{equation}
	the fact that $Z$ is predictively fixed can be succinctly expressed via 
	\begin{equation}\label{eq:pred_var_fix}
		Z M_{\B \B} = M_{\Z \B}  \quad \text{and} \quad Z M_{\B \V} = M_{\Z \V}.
	\end{equation}
	We employ these repeatedly in the computations below.
	
	Let us now use the fact that $M$ is a Poisson map to show that $Z$ must be a Poisson variable.
	The defining equation, $M \Omega M^T = \Omega$, gives the following equations as its $\B \B$, $\B \Q$, $\Q \B$, and $\Q \Q$ blocks: 
	\begin{align}
		\label{eq:symp_1} M_{\B\B} \Omega_\B M_{\B\B}^T + M_{\B\W} M_{\B\V}^T - M_{\B\V} M_{\B\W}^T &= \Omega_\B \\
		\label{eq:symp_2} M_{\B\B} \Omega_\B M_{\Q\B}^T + M_{\B\W} M_{\Q\V}^T - M_{\B\V} M_{\Q\W}^T &= 0 \\
		\label{eq:symp_4} M_{\Q\B} \Omega_\B M_{\B\B}^T + M_{\Q\W} M_{\B\V}^T - M_{\Q\V} M_{\B\W}^T &= 0 \\
		\label{eq:symp_5} M_{\Q\B} \Omega_\B M_{\Q\B}^T + M_{\Q\W} M_{\Q\V}^T - M_{\Q\V} M_{\Q\W}^T &= 0 . 
	\end{align}
	Applying $Z$ and $Z^T$ to \cref{eq:symp_1} on the left and right respectively gives 
	\begin{equation}
		Z \Omega_\B Z^T =  M_{\Z\B} \Omega_\B M_{\Z\B}^T + Z M_{\B\W} M_{\Q\V}^T Z^T - Z M_{\Q\V}M_{\B\W}^T Z^T.
	\end{equation}
	Using \cref{eq:symp_2,eq:symp_4} respectively, we find
	\begin{align}
		Z M_{\B\W} M_{\Q\V}^T Z^T &= - M_{\Z\B} \Omega_\B M_{\Z\B}^T + M_{\Z \V} M_{\Z \W}^T \\
		- Z M_{\Q\V}M_{\B\W}^T Z^T &= - M_{\Z\B} \Omega_\B M_{\Z\B}^T - M_{\Z \W} M_{\Z \V}^T.
	\end{align}
	Putting the last three equations together and applying \cref{eq:symp_5}, we obtain
	\begin{equation}
		Z \Omega_\B Z^T = - M_{\Z\B} \Omega_\B M_{\Z\B}^T + M_{\Z \V} M_{\Z \W}^T - M_{\Z\W} M_{\Z \V}^T = 0.
	\end{equation}
	Altogether, this shows that $Z$ is a Poisson variable, and thus that $\I_\pre$ is a coisotropic subspace of $\B$.
	Consequently, the predicted variable given by the quotient map ${\im(\beta) \to \linefaktor{\im(\beta)}{\I_\pre}}$ is a Poisson partial variable by \cref{def:partial_variable}.
	
	\paragraph{Part 2 (sufficiency):}
	For the converse, we construct the measurement that saturates the epistemic horizon.
	To this end, consider two arbitrary coisotropic subspaces $\I$ and $\F$ of $\B$, that satisfy $\I \subseteq \F$.
	Denoting the symplectic complements by $(\ph)^\omega$, we have 
	\begin{equation}\label{eq:compl_incl}
		\F^\omega \subseteq \I^\omega \subseteq \I \subseteq \F.
	\end{equation}
	
	Because both quotients $\linefaktor{\B}{\F}$ and $\linefaktor{\F}{\I}$ are in general non-trivial, both non-empirical and empirical information\footnote{Non-empirical information predicted corresponds to the quotient $\linefaktor{\B}{\F}$, while the empirical information predicted by the subject can be identified with the quotient $\linefaktor{\F}{\I}$.} are non-trivial and have to be encoded in the final state of $\B$.
	To disentangle these, we decompose $\B$ into a direct sum of two subsystems $\B = \B_1 \oplus \B_2$, where  
	$\B_1$ will be used for predicting non-empirical information and $\B_2$ for the empirical information.
	
	Specifically, we choose $\B_1$ to be a symplectic subspace of $\B$, such that $\F^\omega$ is a Lagrangian subspace thereof.
	While not unique, such a $\B_1$ exists by Darboux Theorem, since $\F^\omega$ is an isotropic subspace of $\B$.
	We then define $\B_2$ to be the symplectic complement of $\B_1$, so that it is another symplectic subspace of $\B$ and satisfies $\B_1 \cap \B_2 = \{0\}$ and $\B = \B_1 \oplus \B_2$.
	
	By definition, we then have $\F \cap \B_1 = \F^\omega$, and using \eqref{eq:compl_incl} we also get $\I \cap \B_1 = \F^\omega$.
	As we can see, when restricted to the subsystem $\B_1$, our prospective preparation support $\F$ and predictively inaccessible subspace $\I$ coincide. 
	That is, $\B_1$ is indeed going to carry non-empirical information only.
	
	On the other hand, by $(\F^\omega)^\omega = \F$ we have $\B_2 \subseteq \F$, from which $\F\cap \B_2 = \B_2$ follows.
	We can also define the subspace $\I_2 \coloneq \I \cap \B_2$, which is coisotropic within $\B_2$.
	Since the prospective preparation support for the second subsystem is full, this system is going to carry empirical information only.
	In particular, the variable of $\B_2$ that we want to encode in the manifest variable is the quotient projection
	\begin{equation}\label{eq:target_measured}
		Z \colon \B_2 \to \newfaktor{\B_2}{\I_2}.
	\end{equation}
	Since $\I_2$ is coisotropic, $Z$ is a Poisson variable.
	
	What is more, $\F$ and $\I$ are both aligned with the decomposition of $\B$.
	That is, we have
	\begin{equation}\label{eq:alignment}
		\F = (\F^\omega)^\omega = (\F \cap \B_1)^\omega = \F^\omega + \B_1^\omega = \F^\omega \oplus \B_2,
	\end{equation}
	which, using $\I \subseteq \F$, implies
	\begin{equation}\label{eq:alignment_2}
		\I = \F \cap \I = (\F^\omega \cap \I) \oplus (\B_2 \cap \I ) = \F^\omega \oplus \I_2.
	\end{equation}
	
	Let us now consider both the ready system and the toy subject to have identical underlying ontic state space given by $\R = \R_1 \oplus \R_2$.
	The first subsystem is identical to the corresponding object subsystem: $\R_1 \cong \B_1$.
	The second one is constructed to be able to encode the necessary information in its Lagrangian subspace about $\B_2$.
	That is, we set $\R_2 \cong \Q_2 \oplus \P_2$, where both $\Q_2$ and $\P_2$ are Lagrangian subspaces of $\R_2$ (with symplectic form given by \eqref{eq:sympl_form}) and isomorphic to the codomain of $Z$ from \eqref{eq:target_measured}.
	We also choose a decomposition of $\R_1$ into two Lagrangian subspaces via $\R_1 = \Q_1 \oplus \P_1$.
	The ready state variability of the ready system is given by ${\F^\omega \oplus \P_2 \subseteq \R}$ and the manifest variable of the subject is the projection ${\pi_{\Q_1} \oplus \pi_{\Q_2}}$. 
	
	As the measurement interaction $M \colon \B \oplus \R \to \B \oplus \R$, we consider the direct sum (up to reordering of subsystems) of
	\begin{equation}
		M_1 \colon \B_1 \oplus \R_1 \to \B_1 \oplus \R_1  \qquad \text{and} \qquad  M_2 \colon \B_2 \oplus \R_2 \to \B_2 \oplus \R_2
	\end{equation}
	where $M_1$ is the swapping interaction introduced in \cref{eq:swap} and $M_2$ is the measurement of an arbitrary Poisson variable introduced in \cref{eq:M_Pson}. 
	In particular, we choose $M_2$ to be the measurement which has $Z$ from \eqref{eq:target_measured} as its retrodicted variable.
	
	Let us now compute the preparation support of $M$.
	It is given by $\im(\beta_1) \oplus \im(\beta_2)$ for 
	\begin{align}
		\beta_1 \colon \B_1 \oplus \F^\omega &\to \B_1 & \beta_1(b \oplus t) &= t, \\
		\label{eq:beta_2} \beta_2 \colon \B_2 \oplus \P_2 &\to \B_2 & \beta_2(b' \oplus p) &= b' + \Omega_{\B_2} Z^T p.
	\end{align}
	Consequently, the preparation support of $M$ is $\F^\omega \oplus \B_2$, which coincides with $\F$ by \cref{eq:alignment}.
	
	By the direct sum splitting of $M$, its predictively inaccessible subspace is given by the direct sum of the predictively inaccessible subspaces of $M_1$ and $M_2$ respectively.
	For the swapping interaction, we have
	\begin{equation}
		\epsilon_1 \colon \B_1 \oplus \F^\omega \to \Q_1  \qquad \qquad \epsilon_1(b \oplus t) = \pi_{\Q_1} b,
	\end{equation}
	so that the kernel of $\epsilon_1$ is $\P_1 \oplus \F^\omega$ and its image under $\beta_1$, i.e.\ the predictively inaccessible subspace of $M_1$, is $\F^\omega$.
	
	To compute the predictively inaccessible subspace of $M_2$, note that we have
	\begin{equation}
		\epsilon_2 \colon \B_2 \oplus \P_2 \to \Q_2  \qquad \qquad \epsilon_2(b' \oplus p) = Zb',
	\end{equation}
	so that $\ker(\epsilon_2) = \ker(Z) = \I_2$ holds.
	Since $Z$ is a Poisson variable, we have
	\begin{equation}
		Z \Omega_{\B_2} Z^T p = 0
	\end{equation}
	for all $p \in \P_2$, so that the term $\Omega_{\B_2} Z^T p$ from \eqref{eq:beta_2} is necessarily an element of the kernel of $Z$.
	Combining these two facts and recalling our choice of $Z$ from \eqref{eq:target_measured}, we find that the predictively inaccessible subspace of $M_2$ is 
	\begin{equation}
		\beta_2 \bigl( \ker(\epsilon_2) \bigr) = \ker(Z) = \I_2.
	\end{equation}
	In conclusion, the predictively inaccessible subspace of $M$ is $\F^\omega \oplus \I_2$, which coincides with $\I$ by \cref{eq:alignment_2}.
	
	This completes the proof, since we have shown, provided arbitrary coisotropic $\I \subseteq \F$, that there is a measurement $M$ with $\I$ and $\F$ as its predictively inaccessible subspace and preparation support respectively.
\end{proof}

\begin{proof}[\bfseries Proof of \cref{thm:prediction_Poisson_MV_complete_RS}]
	This follows by combining \cref{thm:repeatable_Poisson_MV_complete_RS} with \cref{thm:prediction_repeatable_kernels}.
	Specifically, the measurement $M^{\rm full}$ used in the proof of the former result has full preparation support, which follows directly from its definition via \cref{eq:M_full}.
	Therefore, \cref{eq:prediction_repeatable_1} reads
	\begin{equation}
		\I_\ret = \I_\pre + \beta(\{0\} \oplus \V) = \I_\pre,
	\end{equation}
	where the second equation follows because $\V = \{0\}$ is the definition of a complete ready state.
	Consequently, the variable predicted by $M^{\rm full}$ is the identity map.
	
	The same argument can be applied to the composite measurement from \eqref{eq:comp_meas}, so that also any other surjective linear map can be predicted.
\end{proof}

\begin{theorem}[EH for prediction; complete manifest variable; Poisson ready state]
	\label{thm:prediction_complete_MV_Poisson_RS}
	For toy subjects whose manifest variable is complete and ready state is Poisson, a variable can be predicted if and only if the preparation support $\im(\beta)$ and the predictively inaccessible subspace $\I_\pre$ satisfy ${\bigl(\im(\beta)\bigr)^\omega \subseteq \I_\pre}$. 
\end{theorem}

\begin{proof}[\bfseries Proof of \cref{thm:prediction_complete_MV_Poisson_RS}]
	We split the proof into two parts, first showing these conditions must be satisfied by any predictable variable and then constructing a measurement with these properties.
	
	\paragraph{Part 1 (necessity):}
	Consider an arbitrary element $a$ of the symplectic complement of the preparation support, i.e.\ an element of $\bigl(\im(\beta)\bigr)^\omega$.
	Our task is to show that $a$ lies necessarily within the predictively inaccessible subspace $\I_\pre$.
	
	To this end, by \cref{lem:ops_forms}~\ref{it:ops_kernel}, there exist $c \in \B$ and $r \in \R$ satisfying
	\begin{equation}\label{eq:u_preimage}
		c \oplus r \in \bigl(\ker(M)\bigr)^\omega \qquad \text{and} \qquad M(c \oplus r) = a \oplus 0.
	\end{equation}
	Then, by \cref{lem:ops_forms}~\ref{it:ops_forms}, we have
	\begin{equation}\label{eq:ops_formb_2}
		\omega_{\B \oplus \R} \bigl( c \oplus r, b \oplus v \bigr) = \omega_{\B \oplus \S} \bigl( M(c \oplus r), M(b \oplus v) \bigr)
	\end{equation}
	for all $b \in \B$ and all $v \in \V$.
	Note that we can choose $b$ and $v$ arbitrarily because \cref{lem:ops_forms} also guarantees that $\B \oplus \R$ splits into two symplectic subspaces via ${\ker(M) \oplus \bigl(\ker(M)\bigr)^\omega}$ and the component of $b \oplus v$ within the kernel of $M$ contributes by $0$ to both sides of \cref{eq:ops_formb_2}.
	
	Since the manifest variable is complete, we have
	\begin{equation}
		M(b \oplus v) = \beta(b \oplus v) \oplus \epsilon(b \oplus v)
	\end{equation}
	so that the above properties imply
	\begin{equation}
		\begin{split}
			\omega_{\B} (c, b) + \omega_{\R} (r,v) &= \omega_{\B \oplus \S} \bigl( a \oplus 0, \beta(b \oplus v) \oplus \epsilon(b \oplus v) \bigr) \\
			&= \omega_{\B} \bigl(a, \beta(b \oplus v) \bigr) + \omega_{\S}\bigl(0, \epsilon(b \oplus v) \bigr),
		\end{split}
	\end{equation}
	where both the summands in the final expression are equal to $0$.
	The first one vanishes by our assumption of $a \in \bigl(\im(\beta)\bigr)^\omega$ and the second one by the bilinearity of the symplectic form.
	Setting $v = 0$, we thus get 
	\begin{equation}
		\forall b \in \B \quad \omega_{\B} (c, b) = 0,
	\end{equation}
	which means $c = 0$ by the non-degeneracy of the symplectic form.
	On the other hand, choosing $b = 0$ gives
	\begin{equation}
		\forall v \in \V \quad \omega_{\R} (r, v) = 0,
	\end{equation}
	which means $r \in \V^\omega$.
	Since the ready state is Poisson, we have $\V^\omega \subseteq \V$ and thus \eqref{eq:u_preimage} gives
	\begin{equation}
		a = \beta(0 \oplus r) \qquad  0 = \epsilon(0 \oplus r)
	\end{equation}
	for $r \in \V$.
	Consequently, $a$ is indeed an element of $\beta(\ker \epsilon)$ as we wanted to show.
	
	\paragraph{Part 2 (sufficiency):}	
	For the converse, consider an arbitrary coisotropic subspace $\F$ of $\B$ and another subspace $\I$
	satisfying
	\begin{equation}\label{eq:compl_incl_2}
		\F^\omega \subseteq \I \subseteq \F.
	\end{equation}
	Our aim is to construct a measurement interaction $M$, so that given a Poisson ready state and a complete manifest variable, we get
	\begin{equation}
		\im(\beta) = \F \qquad \text{and} \qquad  \I_\pre = \I.
	\end{equation}
	
	Just as in the proof of \cref{thm:prediction_Poisson_MV_Poisson_RS}, we split the toy object into two subsystems $\B_1$ and $\B_2$.
	Like before, $\B_1$ is a choice of a symplectic subspace for which $\F^\omega$ is Lagrangian and $\B_2$ is the symplectic complement thereof.
	The very same argument then also establishes 
	\begin{equation}\label{eq:alignment_3}
		\F = \F^\omega \oplus \B_2 \qquad \qquad \qquad \I = \F^\omega \oplus \I_2
	\end{equation}
	with respect to this decomposition of $\B$, where $\I_2 \coloneq \I \cap \B_2$ is now not necessarily a coisotropic subspace of $\B_2$.
	We also choose a complementary subspace $\C_2$ of $\B_2$ of dimension ${\dim(\B_2) - \dim(\I_2)}$, such that we have 
	\begin{equation}\label{eq:S2_decomp}
		\B_2 \cong \I_2 \oplus \C_2.
	\end{equation}
	
	We choose the ready system and the toy subject to have underlying ontic state space given by $\R = \R_1 \oplus \R_2$ and $\S = \S_1 \oplus \S_2$ respectively.
	The first subsystem of each is identical to the corresponding object subsystem: $\R_1 \cong \B_1 \cong \S_1$.
	We set the ready state variability of $\R_1$ to be $\F^\omega$.
	
	The second ready subsystem is $\R_2 = \Q_2 \oplus \P_2$, where both $\Q_2$ and $\P_2$ are Lagrangian subspaces of $\R_2$ (with symplectic form given by \eqref{eq:sympl_form}), each of which carries a (vector space) isomorphism with $\B_2$ and thus they also inherit the decomposition from \eqref{eq:S2_decomp}.
	We set the ready state variability of $\R_2$ to be $\P_2$.
	
	The second subject subsystem is $\S_2 = \G_2 \oplus \sfH_2$, where both $\G_2$ and $\sfH_2$ are Lagrangian subspaces of $\S_2$, each isomorphic to $\C_2$.
	In particular, this means that we have a surjective linear map $L \colon \B_2 \to \C_2$ given in matrix form by
	\begin{equation}\label{eq:L_proj}
		L = \begin{pmatrix} 0 & \id \end{pmatrix}
	\end{equation}
	with respect to the decomposition \eqref{eq:S2_decomp}, so that the kernel of $L$ is $\I_2$.
	In particular, we then have
	\begin{equation}\label{eq:coisometry}
		L L^T = \id_{\C_2}.
	\end{equation}
	The manifest variable is complete by theorem assumptions, and so it must be given by the identity map $\id_{\S_2}$.
	
	As the measurement interaction $M \colon \B \oplus \R \to \B \oplus \S$, we consider the direct sum (up to reordering of subsystems) of
	\begin{equation}
		M_1 \colon \B_1 \oplus \R_1 \to \B_1 \oplus \S_1  \qquad \text{and} \qquad  M^{\rm pred} \colon \B_2 \oplus \R_2 \to \B_2 \oplus \S_2
	\end{equation}
	where $M_1$ is the swapping interaction introduced in \cref{eq:swap} and
	\begin{equation}
		M^{\rm pred} \colon \B_2 \oplus \Q_2 \oplus \P_2 \to \B_2 \oplus \G_2 \oplus \sfH_2
	\end{equation}
	is given as a matrix by
	\begin{equation}\label{eq:M_pred}
		M^{\rm pred} = \begin{pmatrix} \id_{\B_2} & 0 & \Omega_{\B_2} \\ L & L & L K \\ 0 & 0 & L \end{pmatrix}
	\end{equation}
	where $K$ is a matrix satisfying both
	\begin{equation}\label{eq:piecewise_sympl_form}
		K^T \bigl( \ker(L) \bigr) \subseteq \ker(L) \qquad \text{and} \qquad K - K^T = \Omega_{\B_2}.
	\end{equation}
	
	Let us show that such a $K$ exists.
	To this end, let us use decomposition \eqref{eq:S2_decomp} to write $K$ and $\Omega_{\B_2}$ in block form:
	\begin{equation}
		K = \begin{pmatrix}
				K_{11} & K_{12} \\
				K_{21} & K_{22}
			\end{pmatrix}
		\qquad \qquad	
		\Omega_{\B_2} = \begin{pmatrix}
			\Omega_{11} & -\Omega_{21}^T \\
			\Omega_{21} & \Omega_{22}
		\end{pmatrix}.
	\end{equation}
	Using \cref{eq:L_proj}, the first condition in \eqref{eq:piecewise_sympl_form} follows if we set $K_{12} = 0$.
	The second condition amounts to 
	\begin{equation}
		\begin{pmatrix}
			K_{11} - K_{11}^T & -K_{21}^T \\
			K_{21} & K_{22} - K_{22}^T
		\end{pmatrix} =
		\begin{pmatrix}
			\Omega_{11} & -\Omega_{21}^T \\
			\Omega_{21} & \Omega_{22}
		\end{pmatrix},
	\end{equation}
	which can be satisfied by a suitable choice of $K$ since both $\Omega_{11}$ and $\Omega_{22}$ are antisymmetric.
	
	Having defined $M^{\rm pred}$, let us establish that it is a valid physical transformation in nomic toy theory. 
	Using the block matrix expression of the symplectic form on $\B_2 \oplus \R_2$:
	\begin{equation}
		\Omega = \begin{pmatrix} \Omega_{\B_2} & 0 & 0 \\ 0 & 0 & \id_{\B_2} \\ 0 & -\id_{\B_2} & 0 \end{pmatrix}
	\end{equation}
	we compute
	\begin{equation}
		M^{\rm pred} \, \Omega \, (M^{\rm pred})^T = \begin{pmatrix} \Omega_{\B_2} & 0 & 0 \\ 0 & L \bigl[ \Omega_{\B_2} - K + K^T \bigr] L^T & L L^T \\ 0 & - L L^T & 0 \end{pmatrix}.
	\end{equation}
	By \cref{eq:coisometry,eq:piecewise_sympl_form}, this equals the matrix expression for the symplectic form of $\B_2 \oplus \S_2$, thus proving that $M^{\rm pred}$ is a Poisson map.
	
	We now compute the preparation support of $M$.
	As we have seen in the proof of \cref{thm:prediction_Poisson_MV_Poisson_RS}, $\im(\beta_1)$ is $\F^\omega$.
	Furthermore, note that $\beta^{\rm pred} \colon \B_2 \oplus \P_2 \to \B_2$ is given by
	\begin{equation}\label{eq:sigma_pred}
		\beta^{\rm pred} = \begin{pmatrix} \id_{\B_2} & \Omega_{\B_2} \end{pmatrix}
	\end{equation}
	and thus $\im(\beta^{\rm pred}) = \B_2$, so that we obtain
	\begin{equation}
		\im(\beta_1) \oplus \im(\beta^{\rm pred}) = \F^\omega \oplus \B_2 = \F
	\end{equation}
	by \eqref{eq:alignment_3}.
	Therefore, the preparation support of $M$ is $\F$ as required.
	
	As we have seen in the proof of \cref{thm:prediction_Poisson_MV_Poisson_RS}, the predictively inaccessible subspace of $M_1$ is $\F^\omega$.
	On the other hand, $\epsilon^{\rm pred} \colon \B_2 \oplus \P_2 \to \G_2 \oplus \sfH_2$ is given by
	\begin{equation}
		\epsilon^{\rm pred}(b \oplus p) = \begin{pmatrix} L b + L K p \\ L p \end{pmatrix}.
	\end{equation}
	Setting this equal to $0$ gives $p \in \ker(L)$ and $L b = - L K p$ as conditions for $b \oplus p$ to be elements of $\ker(\epsilon^{\rm pred})$.
	Mapping this kernel by $\beta^{\rm pred}$ and then $L$ gives
	\begin{equation}
		L \comp \beta^{\rm pred} (b \oplus p) = L b + L \Omega_{\B_2} p = L (\Omega_{\B_2} - K) p = - L K^T p = 0,
	\end{equation}
	where the first equality is by \eqref{eq:sigma_pred}, second one by $b \oplus p \in \ker(\epsilon^{\rm pred})$, and the last two both by \eqref{eq:piecewise_sympl_form}.
	This establishes that the predictively inaccessible subspace of $M^{\rm pred}$ is contained within $\ker(L) = \I_2$.
	But we also know that $b \oplus 0$ is an element of $\ker(\epsilon^{\rm pred})$, as long as $b \in \ker(L)$, and thus by \eqref{eq:sigma_pred} we get that the predictively inaccessible subspace of $M^{\rm pred}$ is in fact equal to $\I_2$.
	Altogether, we have 
	\begin{equation}
		\I_\pre = \F^\omega \oplus \I_2 = \I
	\end{equation}
	by \eqref{eq:alignment_3}, so that the predictively inaccessible subspace of $M$ is indeed $\I$ as required, concluding the proof.
\end{proof}

\end{document}